\documentclass[preprint,pre,nobibnotes]{revtex4}%
\usepackage{amsfonts}
\usepackage{amsmath}
\usepackage{amssymb}
\usepackage{graphicx}%
\providecommand{\U}[1]{\protect\rule{.1in}{.1in}}
\newtheorem{theorem}{Theorem}

\newenvironment{proof}[1][Proof]{\noindent\textbf{#1.} }{\ \rule{0.5em}{0.5em}}
\begin{document}
\preprint{UATP/1004}
\title{Non-equilibrium thermodynamics.\ II: Application to inhomogeneous systems}
\author{P.D. Gujrati}
\email{pdg@uakron.edu}
\affiliation{Department of Physics, Department of Polymer Science, The University of Akron,
Akron, OH 44325 USA}

\begin{abstract}
We provide an extension of a recent approach to study non-equilibrium
thermodynamics [Phys. Rev. E \textbf{81}, 051130 (2010), to be denoted by I in
this work] to inhomogeneous systems by considering the latter to be composed
of quasi-independent subsystems. The system $\Sigma$ along with the
(macroscopically extremely large) medium $\widetilde{\Sigma}$\ forms an
isolated system $\Sigma_{0}$. Starting from the Gibbsian formulation of the
entropy for $\Sigma_{0}$, which is valid even when $\Sigma_{0}$ is out of
equilibrium, we derive the Gibbsian formulation of the entropy of $\Sigma$,
which need not be in equilibrium. We show that the additivity of entropy
requires \emph{quasi-independence} of the subsystems, which requires that the
interaction energies between different subsystems must be \emph{negligible} so
that the energy also becomes additive. The thermodynamic potentials of
subsystems such as the Gibbs free energy that continuously decrease during
approach to equilibrium are determined by the field parameters (temperature,
pressure, etc.) of the medium and exist no matter how far the subsystems are
out of equilibrium so that their field variables may not even exist. This and
the requirement of quasi-independence make our approach different from the
conventional approach due to de Groot and others as discussed in the text. As
the energy depends on the frame of reference, the thermodynamic potentials and
Gibbs fundamental relation, but not the entropy, depend on the frame of
reference. The possibility of relative motion between subsystems described by
their\ net linear and angular momenta gives rise to viscous dissipation. The
concept of internal equilibrium introduced in I is developed further here and
its important consequences are discussed for inhomogeneous systems. The
concept of internal variables (various examples are given in the text) as
variables that cannot be controlled by the observer for non-equilibrium
evolution is also discussed. They are important because the internal
equilibrium in the presence of internal variables is lost if internal
variables are not used in thermodynamics. It is argued that their affinity
vanishes only in equilibrium. Gibbs fundamental relation, thermodynamic
potentials and irreversible entropy generation are expressed in terms of
observables and internal variables. We use these relations to eventually
formulate the non-equilibrium thermodynamics of inhomogeneous systems. We also
briefly discuss the case when bodies form an isolated system without any
medium to obtain their irreversible contributions and show that this case is
no different than when bodies are in an extremely large medium.

\end{abstract}
\date[July 28, 2010]{}
\maketitle

\section{ Introduction}

\subsection{Nature of the Problem}

In an earlier paper \cite{GujratiNETI}, which we will refer to as I in this
work, we have considered some of the consequences of applying the second law
of thermodynamics to an isolated system $\Sigma_{0}$, which consists of a
macroscopic system of interest $\Sigma$ containing a fixed number $N$ of
particles (atoms or molecules) surrounded by an extremely large medium
$\widetilde{\Sigma}$; see \ Fig. \ref{Fig_Systems}. From now on, it will be
implicitly assumed that the medium is extremely large to be unaffected by the
system. (Later in this work, we will also consider $\Sigma$ to have a fixed
volume $V$ instead of $N$ or fixed $V$ and $N$. Furthermore, we will also
consider briefly the case of many similar size systems forming the isolated
system $\Sigma_{0}$ without the extremely large medium $\widetilde{\Sigma}$.)
The motivation has been to develop this approach to obtain a non-equilibrium
thermodynamic description of the open system $\Sigma$ under various
conditions. Throughout this work, we will use \emph{body} to refer to any of
the three kinds systems: the isolated system, the medium or the system. All
quantities related to $\Sigma_{0},\widetilde{\Sigma}$ and $\Sigma$ will be
denoted by a suffix $0$,$\sim$ over the top, and without any suffix,
respectively. All quantities related to a body will be denoted without any
suffix, as if we are dealing with an open system.\ Similarly, in this work, we
will say that the system is \emph{open} when it is in a medium. Even though it
is not the common usage, this should not cause any confusion as the context
will be clear.%
\begin{figure}
[ptb]
\begin{center}
\includegraphics[
height=2.5322in,
width=5.2243in
]%
{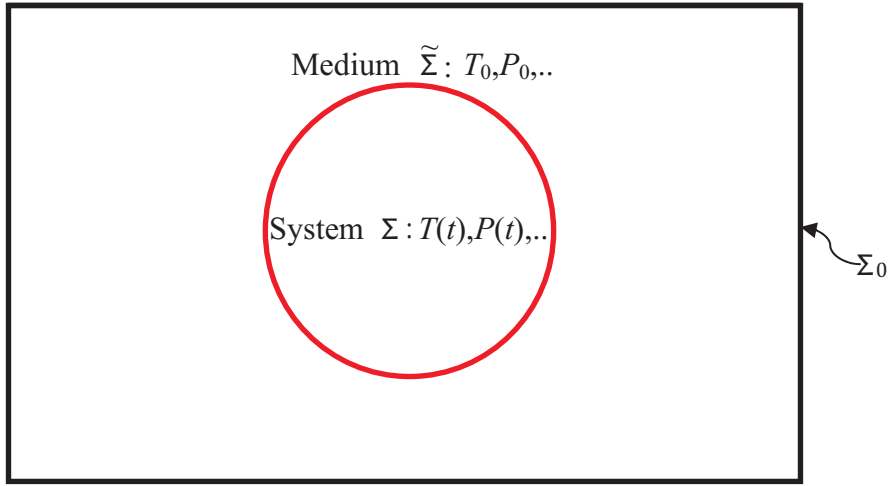}%
\caption{Schematic representation of a macroscopically large system $\Sigma$
and the medium $\widetilde{\Sigma}$ surrounding it to form an isolated system
$\Sigma_{0}$. The system is a very small part of $\Sigma_{0}$. The medium is
described by its fixed fields $T_{0},P_{0},$ etc. while the system, if in
internal equilibrium (see text) is characterized by $T(t),P(t),$ etc.
\ \ \ \ }%
\label{Fig_Systems}%
\end{center}
\end{figure}

To avoid complications due to external shear, we had only considered $\Sigma$
under no external shear in I. This restriction is easily removed as we will do
here. The isolated system will still have \emph{no} external force acting on
it to ensure its isolation; see below also. We will now allow forces acting at
the surface $\partial V$ of the system or any of its subsystems; see Fig.
\ref{Fig_two_boxes_modified_system}(a). These forces must balance the internal
stress tensor for mechanical equilibrium. Thus, the force $t_{i}df$ acting on
a surface element $df$ must equal the stress force $\tau_{ij}n_{j}df$
(summation over repeated indices implied), and we have
\cite{Landau-FluidMechanics,Landau-Elasticity} for the surface traction force%
\begin{equation}
t_{i}=\tau_{ij}n_{j}. \label{Force_stress_relation}%
\end{equation}%
\begin{figure}
[ptb]
\begin{center}
\includegraphics[
trim=1.057052in 0.399253in 0.000000in 0.000000in,
height=4.3716in,
width=4.4105in
]%
{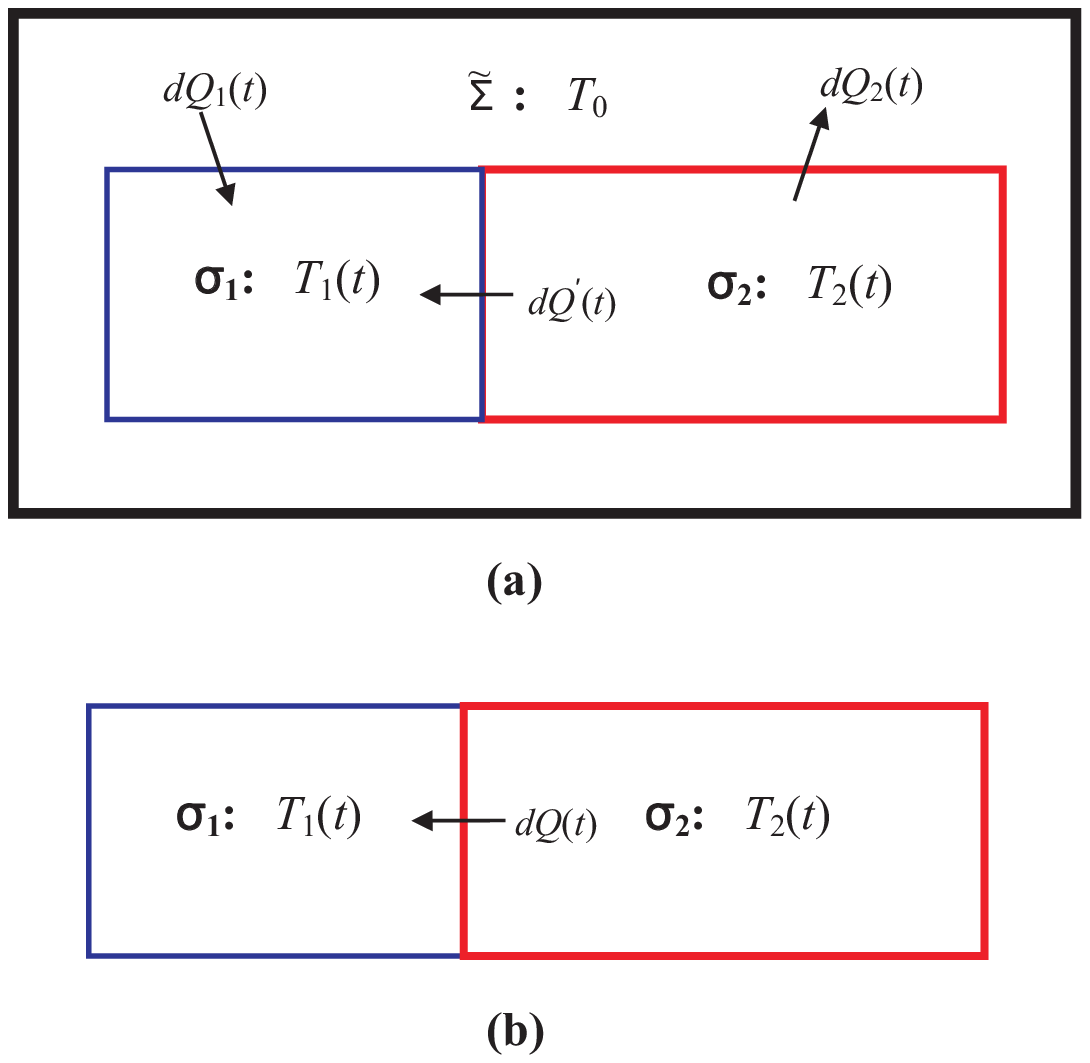}%
\caption{We show schemtically the two subsystems $\sigma_{1}$ and $\sigma_{2}$
($T_{2}(t)>T_{1}(t)$) forming the system $\Sigma$ in an extensively large
medium in (a) and by themselves forming an isolated system without an
extensively large medium in (b). The heat output $dQ(t)$ in (a) by $\sigma
_{2}$ is the sum of $dQ^{\prime}(t)$ and $dQ_{2}(t)$, while the heat intake by
$\sigma_{1}$ is the sum of $dQ^{\prime}(t)$ and $dQ_{1}(t)$. As we are dealing
with an isolated system, the heat input and output must be equal. Therefore,
$dQ_{1}(t)\equiv dQ_{2}(t)$. The equality of the heat input and output is also
true in (b). As the heat transfer occurs between objects does not occur
isothermally, there is irreversible entropy generation due to each heat
transfer. We will study this issue later in Sect.
\ref{marker_comparable_bodies}.}%
\label{Fig_two_boxes_modified_system}%
\end{center}
\end{figure}
Here, $\mathbf{n}$ is the outward unit normal at the surface element $df$
surrounding a point on the surface. This condition must be satisfied at every
point on the surface $\partial V.$ The net force and torque acting on the
system are given by
\[
\mathbf{F}\equiv%
{\displaystyle\oint\limits_{\partial V}}
\mathbf{t}df,\ \ \mathbf{K}\equiv%
{\displaystyle\oint\limits_{\partial V}}
\mathbf{r}^{(\text{s})}\times\mathbf{t}df,
\]
respectively; here, for convenience, $\mathbf{r}^{(\text{s})}$ is taken to be
the radius vector of the surface element with respect to the center of mass of
the system. The external forces are responsible for the deformation of the
system, and also result in the translation and rotation of the system. Let us
consider an infinitesimal volume element $dV$ of mass $dm$, which is moving
with a velocity $\mathbf{v}(t)$ and rotating with an angular velocity
$\boldsymbol{\omega}(t)$ and has an intrinsic angular momentum $\mathbf{m}%
(t)dV$. The linear and angular momenta of the system in some fixed frame are
given by%
\[
M\mathbf{V(}t\mathbf{)\equiv}%
{\textstyle\int\nolimits_{V(t)}}
\mathbf{v}(t)dm,\ \mathbf{M}(t)\mathbf{\equiv}M\mathbf{R(}t\mathbf{)\times
V(}t\mathbf{)+}%
{\textstyle\int\nolimits_{V(t)}}
\mathbf{r}(t)\times\mathbf{v}(t)dm+%
{\textstyle\int\nolimits_{V(t)}}
\mathbf{m}(t)dV,
\]
respectively; here $M=m_{0}N$\ is the mass of the system ( $m_{0}$ being the
mass of a particle), which we consider fixed for fixed $N$, and $\mathbf{R(}%
t\mathbf{)}$ and $\mathbf{V}(t)$ are the location and the translational
velocity of the center of mass in this frame. If the frame is taken to be the
center of mass frame, then $\mathbf{R(}t\mathbf{)}$ and $\mathbf{V}(t)$ are
zero. If the body as a whole is \emph{stationary}, then the energy of the body
is known as the internal energy. Such a stationary situation was considered in
our previous work \cite{GujratiNETI}, where no motions were considered. This
limitation will be removed here so as to allow for relative motions
(translation and rotation) between the system and the medium. These relative
motions are the additional sources of viscous dissipation and give rise to
additional irreversibility. The irreversibility due to temperature difference
(such as between the system and the medium) has already been considered in I.

We should mention here the recent somewhat comprehensive investigation carried
out by Bouchbinder and Langer \cite{Langer} who also consider a system under
external shear; however, our approaches and emphases are very different. We
should also mention earlier very different equilibrium-like attempts by
Lubchenko and Wolynes \cite{Wolynes} and by \"{O}ttinger \cite{Ottinger}.
Mention should also be made of a very interesting phenomenological approach by
Oono and Paniconi \cite{Oono} on steady state thermodynamics, which was later
advanced by Sasa and Tasaki \cite{Sasa}. The classical local non-equilibrium
thermodynamics due to de Donder \cite{Donder,deGroot,Prigogine,Prigogine-old}
is close in spirit to our approach and that of Bouchbinder and Langer
\cite{Langer}. It will be the standard formulation with which we will compare
and contrast our approach initiated in \cite{GujratiNETI}. Therefore, for the
sake of continuity and clarity, we briefly discuss the classical formulation
involving local equilibrium in Sect. \ref{Marker_Local_Thermodynamics}. We
note that there are other versions of non-equilibrium thermodynamics usually
known as the extended, rational and GENERIC non-equilibrium thermodynamics
\cite{Demirel,Muschik}; however, we do not discuss these formulations in this work.

It is well known that the second law of thermodynamics for the isolated system
$\Sigma_{0}$ states that%
\begin{equation}
\frac{dS_{0}(t)}{dt}\geq0, \label{Second_Law_0}%
\end{equation}
where $S_{0}(t)$ denotes the entropy of $\Sigma_{0}$ at some instant $t$. It
should be stressed that the isolation of $\Sigma_{0}$ requires that there be
something outside of\ $\Sigma_{0}$ from which it is isolated. Therefore, we
will assumes that $\Sigma_{0}$\ is confined to a finite though extremely large
volume $V_{0}$ \cite{Gujrati-Symmetry}. The isolation requires that we neglect
all interactions, such as gravitational interactions, of $\Sigma_{0}$ with the
outside that cannot be shielded. All interactions with the outside should be
relatively very weak to be \emph{negligible}. All relevant interactions must
be confined inside the volume $V_{0}$. This to not to be taken as a weakness
of our approach as including the interactions with outside will only make
$\Sigma_{0}$ an open system, so that our investigation of an open system,
which is our primary concern, can then be applied to it.

We should emphasize the following important point. The law of increase of
entropy in Eq. (\ref{Second_Law_0}) should give a pause to those readers who
believe that the concept of entropy is meaningful only for an equilibrium
state and that the entropy cannot be defined for a non-equilibrium state. This
belief is unfounded. The mere fact that the second law dictates the approach
to equilibrium clearly implies that the entropy exists even when the system is
out of equilibrium. This issue and the related history, in particular the
Gibbs formulation of the non-equilibrium entropy, has been reviewed recently
\cite[and references cited therein]{Gujrati-Symmetry}, and we refer the reader
to this for further details. We should stress that the Gibbs formulation of
the entropy requires that the dynamics of the system and that of the isolated
system are not deterministic, but rather stochastic; see
\cite{Gujrati-Symmetry} for further details.

\subsection{Important Restrictions in I and Their Removal in the Current Work
\label{marker_important_assumptions_I}}

An assumption that was implicit, but not stated, in I was that there was no
relative motion between the system and the medium and that the isolated system
was stationary. This is normally the case in practice in which the various
bodies $(\Sigma_{0},\widetilde{\Sigma}$ and $\Sigma)$ are stationary in the
laboratory frame of reference; the latter we will denote by $\mathcal{L}$ in
this work. In general, one can allow for a translation and rotation of a body,
as done here because we now wish to study their effects on the deformation of
the system. Because of the stationary assumption, there is no difference
between the energy and the internal energy for $\Sigma_{0},\widetilde{\Sigma}$
and $\Sigma$. This is a very common but useful assumption as the entropies
depend on the internal energies and not on the energies \cite{Landau}. The
latter energies may contain the contribution from the translation and rotation
of the system as a whole. In contrast, the internal energies are the energies
systems have in a frame in which they are stationary. Whenever we discuss both
energies together in the following, the internal energy will be denoted by a
superscript "i" to distinguish it from the energy, which is denoted without
the superscript; otherwise, it will be clear which energy we are considering.
For a stationary system in the lab frame $\mathcal{L}$,\ the internal energy
is the same as the energy. The translation of a body as a whole merely affects
the energy, but not its thermodynamic properties. However, the rotation of a
body as a whole gives rise to centrifugal potential energy that modifies the
energies of microstates and has to be carefully incorporated in any
thermodynamic investigation \cite[see Sect. 34]{Landau} as done in Sect.
\ref{marker_Rotating_Systems}; see also Appendices \ref{Appd_Frames} and
\ref{Appd_Rotating_System_0}.

\subsection{Present Goal and the Layout}

We had focused exclusively on the system as a whole without worrying about its
internal parts in I. As entropy palys a central role in our development, we
discuss the formulation of entropy and thermodynamic averages in Sect.
\ref{marker_averages}, where we show the entropy itself as a thermodynamic
average. Under the assumption of internal equilibrium in the system, there
were no relative motions between its various parts that could add additional
irreversible processes inside the system. In this work, we will remove this
limitation and treat the system $\Sigma$\ as inhomogeneous as glasses normally
exhibit both spatial and temporal inhomogeneity; see
\cite{Wolynes2,Chandler,Gujrati-Sokolov} for some recent investigations. We do
this by considering $\Sigma$ to be composed of a collection of a large number
$N_{\text{S}}$ of \emph{subsystems} $\sigma_{k}$, $k=1,2,\cdots,N_{\text{S}},$
which may be different from each other to allow for inhomogeneity and for
relative motions and shear forces between different subsystems in terms of
internal variables. Each subsystem is still \emph{macroscopically} large so
that we can introduce a legitimate entropy function $s_{k}$; see Sect.
\ref{marker_inhomogeneous_arbitrary} for further elaboration.

Apart from the \emph{observables} that can be manipulated by the observer,
there also appear \emph{internal variables} often used in describing glasses,
as is well known from the early works of Davies and Jones
\cite{Jones,Gujrati-book}. The latter variables cannot be manipulated by the
observer and were briefly introduced in I, but not explored. We will also
remedy this situation here and consider glass as an inhomogeneous system with
internal variables. Specifically, we treat translations and rotations of
various parts of a system as internal variables that are generated by surface
traction forces. The alternative approach is to use the traction forces and
the strains instead; see for example \cite{Langer}. The phenomenological
ideology introduced by Davies and Jones \cite{Jones}, which has been recently
reviewed by \"{O}ttinger \cite{Ottinger}, is by now standard and has been
discussed in several textbooks; see for example
\cite{Nemilov-Book,Gutzow-Book}. The observables and internal variables will
be collectively called \emph{state variables}; see Sect. \ref{marker_averages}
for proper definitions of these terms.

In Sect. \ref{marker_internal_equilibrium}, we discuss the consequences of
internal equilibrium and its similarity with and differences from the concept
of local equilibrium \cite{Donder,deGroot,Prigogine,Prigogine-old} discussed
in Sect. \ref{Marker_Local_Thermodynamics}. In particular, we argue in the
form of Theorem \ref{marker_Uniform_Motion} that the system can only sustain a
uniform translation and rotation in internal equilibrium. It is assumed here
that there are no additional conditions (such as the potential flow in a
superfluid) on the velocity. The proof is trivial but the theorem has far
reaching consequences. For example, it implies that the uniform rotation must
be about a principle axis of inertia. A simple way to understand internal
equilibrium is to think of the system as follows. We first disconnect it form
the medium with which it is not in equilibrium, and connect it to another
medium whose conjugate field variables ($T,P$, etc.) are exactly the same as
that of the system. The system remains in equilibrium with this medium so that
there will be no irreversible process in the system. We discuss the
generalization of equilibrium Maxwell relation to systems in internal
equilibrium. The condition for additivity and quasi-independence is considered
in Sect. \ref{marker_entropy_additivity}, where we prove that the Gibbs
entropy formulation in Eq. (\ref{Gibbs entropy}) is also applicable to an open
system, which is assumed to be \emph{quasi-independent} of the medium. Various
thermodynamic potentials are identified in Sect.
\ref{marker_thermodynamic_potentials} that are in accordance with the second
law. It is here that we see a clear distinction between our approach and the
conventional non-equilibrium theory exploiting the local equilibrium concept
\cite{Donder,deGroot,Prigogine,Prigogine-old}. Internal variables are
discussed in Sect. \ref{Marker_Internal Variables}. We prove that the chemical
potential or the affinity associated with an internal variable must be zero
when the system is in equilibrium. We also prove that the entropy expressed
solely in terms of observables when there are independent internal variables
must explicitly depend on time so that while the system is in internal
equilibrium with respect to all state variables, it is not in internal
equilibrium with respect to only observables. A system undergoing uniform
translation and rotation is studied in Sect. \ref{marker_Rotating_Systems}
where we also develop the Gibbs fundamental relation for such a system. We
then apply the results from this study to an isolated system in which the
system and the medium undergo relative translational motion in Sect.
\ref{marker_homogeneous}, but the system is homogeneous. An inhomogeneous
system with relative motions between its subsystems is studied in the next
section. We also discuss in this section the case of several bodies, each in
internal equilibrium but different from others, that form an isolated body
together by themselves without a medium; see Fig.
\ref{Fig_two_boxes_modified_system}(b). These bodies are macroscopic in size,
but are not extensively as large as a medium. We discuss a direct method of
calculating the irreversible entropy generation in each body in terms of the
equilibrium state of all the bodies. We find that the same results are also
obtained by bringing all bodies in contact with a medium as shown in Fig.
\ref{Fig_two_boxes_modified_system}(a). This equivalence is used to prove
Theorem \ref{Theorem_Isolated_Bodies}. All these investigations are extended
to include extra observables and internal variables. The final section
contains concluding discussion and a brief summary of our results.

\section{Local Non-equilibrium Thermodynamics: A Brief
Review\label{Marker_Local_Thermodynamics}}

As entropy is the central quantity appearing in the second law in Eq.
(\ref{Second_Law_0}), we will pay close attention to its determination,
although this is usually not done in classical local non-equilibrium
thermodynamics \cite{Donder,deGroot,Prigogine,Prigogine-old}, where its
existence is taken as a postulate with an implicit assumption that it is
always \emph{additive} \cite{deGroot}. The entropy $S$ of a system is defined
in terms of the \emph{local} entropy density $s(\mathbf{r})$ per unit volume:%
\begin{equation}
S\equiv%
{\textstyle\int\nolimits_{V}}
s(\mathbf{r})dV. \label{Entropy_Additivity_Thermodynamics}%
\end{equation}
The local temperature $T$ and pressure $P$ are assumed continuous functions of
the location $\mathbf{r}$ and time $t$, and are postulated to always exist.
The \emph{additivity} of the energy $E$ results in%
\begin{equation}
E\equiv%
{\textstyle\int\nolimits_{V}}
\left[  e(\mathbf{r})+%
\frac12
m_{0}\rho(\mathbf{r})\mathbf{v}^{2}(\mathbf{r})+\mathbf{m(\mathbf{r})\cdot
}\boldsymbol{\omega}(\mathbf{r})+\psi(\mathbf{r})\right]  dV.
\label{Energy_Additivity_Thermodynamics}%
\end{equation}
Here, $e(\mathbf{r})dV$ is the internal energy, $%
\frac12
m_{0}\rho(\mathbf{r})\mathbf{v}^{2}(\mathbf{r})dV$ and $\mathbf{m(\mathbf{r}%
)\cdot}\boldsymbol{\omega}(\mathbf{r})dV$\ the translational and rotational
kinetic energy density, respectively, and $\psi(\mathbf{r})dV$ the additional
energy contribution due to interactions not included in the internal energy
density in a volume $dV$ of the system; the local mass and angular momentum
densities are given by $m_{0}\rho$ and $\mathbf{m(\mathbf{r})},$ respectively;
compare with Eq. (\ref{subsystem_energy_relation0}) derived later after
limiting it to the volume element $dV$.

The functional form of the entropy density depends on the nature of the
system. For example, for a simple system containing a fixed number of
structureless particles, it is assumed to be a function only of the internal
energy density $e$ and the local number density $\rho$; see for example,
\cite[see Eq. (III.14)]{deGroot}%
\begin{equation}
s(\mathbf{r})=s(e(\mathbf{r}),\rho(\mathbf{r})).
\label{local_equilibrium_entropy}%
\end{equation}
The local Gibbs free energy density $\widehat{g}$ is given by
\begin{equation}
\widehat{g}=e-Ts+P. \label{Local_Gibbs_Free_Energy}%
\end{equation}
whether local equilibrium exists or not. (The unconventional use of the symbol
$\widehat{g}$ instead of $g$ will become clear later when we discuss the Gibbs
free energy, which follows from the second law and which continuously decrease
as the system approaches equilibrium.) However, no direct method of
calculating the entropy is given in this approach except by assuming as
another postulate the validity of Gibbs fundamental relation, which for a
simple system with no internal variables reads
\begin{equation}
Td(s/\rho)=d(e/\rho)+Pd(1/\rho); \label{Gibbs_fundamental_relation}%
\end{equation}
this postulate is a consequence of assuming \emph{local equilibrium} \cite[see
Eq. (III.15)]{deGroot}. The presence of the local temperature $T$ and pressure
$P$ in the fundamental relation imposes strong conditions on the nature of the
entropy in that its partial derivatives are related to the given $T$ and $P$
under local equilibrium, which follow from Eq.
(\ref{Gibbs_fundamental_relation}).

We avoid the above issues in the conventional non-equilibrium thermodynamics
\cite{Donder,deGroot,Prigogine,Prigogine-old} by first identifying the entropy
of the system in terms of microstate probabilities as described below, see Eq.
(\ref{Gibbs entropy}), and then use the concept of \emph{internal equilibrium}
to introduce the temperature and pressure in terms of this entropy; the latter
are defined only when there is internal equilibrium \cite{GujratiNETI}. This
thus avoids the issues inherent in the conventional approach. Our approach is
\emph{not} local in that we always deal with quantities $S,E,$ etc. related to
macroscopically large systems or subsystems as opposed to the conventional
thermodynamics which deals with local quantities $s,e,$ etc. As a consequence,
not only $S,E,$ etc. but also the temperatures, pressures, etc. associated
with these systems or subsystems will not \emph{always} be continuous
functions of space at the interfaces. Thus, we will not impose continuity in
space on these quantities, which makes our approach distinct from the
traditional local non-equilibrium approach of de Donder
\cite{Donder,deGroot,Prigogine,Prigogine-old} where these quantities are
always treated as continuous. In the latter approach, a system can be broken
into subsystems, each sufficiently small, yet large enough to be in internal
equilibrium to satisfy Gibbs fundamental relation. This hypothesis is known as
the local equilibrium hypothesis. We add another requirement, that of
quasi-independence of the subsystems in our approach, which we believe to be
extremely important. Only this requirement ensures that the entropy retains
the additivity property and also remains a state variable, as we will discuss
later in Sect. \ref{marker_entropy_additivity}.

\subsection{Helmholtz Theorem\label{marker_Helmholtz_theorem}}

To accommodate relative motion, we need to allow for surface traction forces
$t_{i}$, which then give rise to internal forces causing the deformation of
the system. These forces can be related to the induced stress tensor
$\tau_{ij}$ as shown in Eq. (\ref{Force_stress_relation}), and will result in
a motion of the system due to non-zero net force and torque acting on the
system. It is well known that the local motions for a deformable system can be
described as a combination of three distinct types of
motions\ \cite{Sommerfeld,Batchelor}:

\begin{enumerate}
\item[(a)] a pure translation

\item[(b)] a pure strain, and

\item[(c)] a pure rotation
\end{enumerate}

by expressing the instantaneous difference in the velocity $\delta\mathbf{v}$
at two nearby points separated by a displacement vector $\delta\mathbf{r}$ as
\begin{equation}
\delta v_{i}=\psi_{ij}\delta x_{j}+\widehat{\omega}_{ij}\delta x_{j},
\label{velocity_decomposition}%
\end{equation}
where the symmetric tensor
\[
\psi_{ij}\equiv%
\frac12
(\frac{\partial v_{i}}{\partial x_{j}}+\frac{\partial v_{j}}{\partial x_{i}})
\]
represents the rate of strain tensor and the antisymmetric tensor
\[
\widehat{\omega}_{ij}\equiv%
\frac12
(\frac{\partial v_{j}}{\partial x_{i}}-\frac{\partial v_{i}}{\partial x_{j}%
})\equiv e_{ijk}\widehat{\omega}_{k}%
\]
is the vorticity tensor, and represents the axial vector $\widehat
{\boldsymbol{\omega}}\boldsymbol{=}%
\frac12
\boldsymbol{\partial}\times\mathbf{v}$\ associated with $\widehat{\omega}%
_{ij}.$ The second term $\widehat{\omega}_{ij}\delta x_{j}$ in Eq.
(\ref{velocity_decomposition}) represents the components of the vector
$\widehat{\boldsymbol{\omega}}\times\delta\mathbf{r}$\textbf{.} One should
think of $\delta\mathbf{v}$ as the relative velocity between two neighboring
points separated by $\delta\mathbf{r}$.

The first contribution represents a pure straining motion while the second
contribution represents a rigid-body rotation. For example, a simple shearing
motion in which plane layers of the system slide over each other can be
treated as a combination of a pure strain (with no rate of volume change) and
a rotation \cite{Batchelor}.

\subsection{Stress Tensor}

The motion at the local level can also be studied directly by considering the
stress tensor. The stress tensor is normally expressed as a sum of the
non-dissipative and dissipative or viscous contributions
\cite{Landau-Elasticity}:%
\begin{equation}
\tau_{ij}=\sigma_{ij}+\sigma_{ij}^{\prime},
\label{stress_tensor_decomposition}%
\end{equation}
in which the viscous contribution $\sigma_{ij}^{\prime}$\ is some function
that depends on the velocity gradients, i.e., on
\[
\partial v_{i}/\partial x_{j},\partial^{2}v_{i}/\partial x_{j}\partial
x_{k},\text{ etc.}%
\]
Thus, we can express it as
\[
\sigma_{ij}^{\prime}=A_{ijkl}f_{kl}(\partial v_{m}/\partial x_{n},\partial
^{2}v_{m}/\partial x_{n}\partial x_{p},...),
\]
where $A_{ijkl}$ does not depend on the velocity distribution and $f_{kl}$ is
a function of the velocity gradients$.$ For example, in a linear approximation
using only $\partial v_{i}/\partial x_{j}$, $f_{kl}$ is taken to be
\[
f_{kl}=f_{0}\partial v_{k}/\partial x_{l}%
\]
with a constant $f_{0}$, which could be conveniently absorbed in $A_{ijkl}.$
In this approximation, we see from Eq. (\ref{velocity_decomposition}) that
$\sigma_{ij}^{\prime}$ depends on both the rate of strain tensor and the
vorticity $\widehat{\boldsymbol{\omega}}$.

In general, we can partition $\sigma_{ij}^{\prime}$\ into a symmetric and an
antisymmetric part, the latter due to the presence of intrinsic rotation of
the system \cite{deGroot} and describing the role of the rotational viscosity.
Thus, we can also partition $\tau_{ij}$\ into a symmetric and an antisymmetric
part. We refer the reader to Chapter 12 in the monograph of de\ Groot and
Mazur \cite{deGroot}. Of course, there may be symmetry reasons such as the
isotropy of the system that would forbid the dependence on vorticity, in which
case there would be no antisymmetric part. It is possible to show
\cite{deGroot} that the rate of change of the intrinsic angular momentum is
determined solely by the antisymmetric part of $\tau_{ij}^{\prime}.$ Thus, the
absence of this part will imply the conservation of the orbital and intrinsic
angular momentum separately.

It can now be shown \cite[Eqs. (XII.18) and (XII.18)]{deGroot} that the
antisymmetric part
\[
\tau_{ij}^{\text{a}}\equiv%
\frac12
\left(  \tau_{ij}-\tau_{ji}\right)
\]
contributes a term proportional to%
\[
\boldsymbol{\tau}\cdot(\widehat{\boldsymbol{\omega}}\mathbf{-}%
\boldsymbol{\omega})
\]
to the rate of change of the internal energy $e$ and to the entropy
production. Here, $\boldsymbol{\tau}\ $\ is the vector associated with
$\tau_{ij}^{\text{a}}$
\[
\tau_{ij}^{\text{a}}=e_{ijk}\tau_{k}%
\]
and $\boldsymbol{\omega}$ represents the angular velocity of rotation of the
system; cf. Eq. (\ref{Energy_Additivity_Thermodynamics}). Thus, this
contribution vanishes for uniform rotation $\widehat{\boldsymbol{\omega}%
}\mathbf{=}\boldsymbol{\omega}$ as expected. For $\widehat{\boldsymbol{\omega
}}\mathbf{\neq}\boldsymbol{\omega}$, there would be precession of the local
volume element \cite{Landau-Mechanics} about the direction of
$\boldsymbol{\omega}$, so that the rotational viscosity would play an
important role until $\widehat{\boldsymbol{\omega}}$\textbf{ }becomes equal
to\textbf{ }$\boldsymbol{\omega}$.

\subsection{Energy Balance and the First Law}

The decomposition in Eq. (\ref{stress_tensor_decomposition}) allows us to
break the surface traction also in two terms related to the individual
contribution:%
\[
\mathbf{t}=\boldsymbol{\sigma}+\boldsymbol{\sigma}^{\prime},
\]
where the two new terms are defined similar to that in Eq.
(\ref{Force_stress_relation}). The non-viscous contribution is sometimes
expressed by taking out the isotropic pressure term as follows%
\[
\boldsymbol{\sigma}=-P\mathbf{n}+\boldsymbol{\varepsilon},
\]
assuming that the thermodynamic pressure can be defined. The rate at which the
work is done \emph{on} the system is given by the average (over all
microstates) of the stress tensor at the surface $\partial V$ of the system:%
\begin{equation}
\frac{dW^{\prime}(t)}{dt}=\overline{%
{\displaystyle\oint\limits_{\partial V(t)}}
t_{j}v_{j}^{(\text{s})}df}=\overline{%
{\displaystyle\oint\limits_{\partial V(t)}}
\sigma_{j}v_{j}^{(\text{s})}df}+\overline{%
{\displaystyle\oint\limits_{\partial V(t)}}
\sigma_{j}^{\prime}v_{j}^{(\text{s})}df}. \label{Total_Work0}%
\end{equation}
where $df_{i}$ and $v_{j}^{(\text{s})}$\ are the components of $d\mathbf{f}%
(t)$\ and the velocity $\mathbf{v}^{(\text{s})}$ of the surface element,
respectively. When the pressure can be defined, this rate can be expressed as%
\begin{equation}
\frac{dW^{\prime}(t)}{dt}=-P(t)\frac{dV(t)}{dt}+\overline{%
{\displaystyle\oint\limits_{\partial V(t)}}
\varepsilon_{j}v_{j}^{(\text{s})}df}+\overline{%
{\displaystyle\oint\limits_{\partial V(t)}}
\sigma_{j}^{\prime}v_{j}^{(\text{s})}df}. \label{Total_Work}%
\end{equation}

In this work, we will not consider latent heats and chemical or nuclear
reactions within the body. In this case, the rate at which the heat is added
to the system is obtained by considering the heat flux through the surface
$\partial V$ and is given by%
\[
\frac{dQ(t)}{dt},
\]
where $dQ(t)$ is the heat added to the system in time $dt$. We can then write
down for the rate of change of the energy due to the dynamics in the system as%
\begin{equation}
\frac{dE(t)}{dt}=\frac{dQ(t)}{dt}+\frac{dW^{\prime}(t)}{dt},
\label{Internal_Energy_Balance_Eq}%
\end{equation}
which is a restatement of the first law of thermodynamics; see also
\cite{Langer}.

It should be remarked that the balance equation
(\ref{Internal_Energy_Balance_Eq}) is in an integral form for the entire
system and should be contrasted with the differential form commonly stated in
textbooks; see for example \cite{deGroot}. The latter formulation is valid for
infinitesimal volumes. Here, we are not interested in such a local
description. Our main focus is to consider regions of the system that are
macroscopically large enough so that they can be treated as quasi-independent.
Under this condition, the entropy of the system can be approximated to a high
accuracy by adding the entropies of the subsystems. We discuss this point
further in Sect. \ref{marker_homogeneous_arbitrary}.

\subsection{Need for Internal Equilibrium\label{marker_internal_equilibrium_0}%
}

If the system as a whole is stationary, then the average velocity
$\overline{\mathbf{V(}t\mathbf{)}}=0$. In this case, the energy of the system
is the internal energy. We will consider this case below in this section for simplicity.

The first law statement in Eq. (\ref{Internal_Energy_Balance_Eq}) does not by
itself allow us to determine the way the entropy of the system changes. This
law is valid even if the system is not in internal equilibrium. In the absence
of internal equilibrium, there is no way to determine the change in the
entropy from this law. We need to relate $dE(t)$ with the change $dS(t)$ in
the entropy to determine the latter. However, if we now assume that the system
is in internal equilibrium (see I and Sect. \ref{marker_internal_equilibrium}
below), then the entropy no longer has an explicit $t$-dependence. In this
case, we can write down the differential of the entropy $S(E,V,N)$ of a
monatomic system of neutral particles with fixed $N$\ as \cite{GujratiNETI}
\begin{equation}
dS=\frac{1}{T(t)}dE+\frac{P(t)}{T(t)}dV; \label{Internal_Equilibrium_Relation}%
\end{equation}
compare this equation with Eq. (\ref{Gibbs_fundamental_relation}). We can now
use Eq. (\ref{Internal_Energy_Balance_Eq}) in this equation to determine the
rate of change of the entropy. However, as we discuss in Sect.
\ref{marker_internal_equilibrium}, in this case there is no \emph{additional}
irreversible entropy production arising from viscous stress tensor. Thus, the
last contribution in Eqs. (\ref{Total_Work0}) or (\ref{Total_Work}) vanishes
\cite[Sect. 7.4.2]{Kuiken}. If there is no shearing at the surface, then the
second term in Eq. (\ref{Total_Work}) also vanishes, and we obtain the
standard relation
\begin{equation}
T(t)\frac{dS(t)}{dt}=\frac{dQ(t)}{dt}; \label{Entropy_Production}%
\end{equation}
this identification was used in I \cite{GujratiNETI}.

\subsection{Reversible Entropy Change and Irreversible Entropy
Production\label{marker_entropy_production}}

We wish to emphasize, as was done in the previous work \cite[see Eqs. (16) and
(18) in particular]{GujratiNETI} that this assumption of internal equilibrium
does not mean that the\emph{ irreversible entropy production} $d_{\text{i}%
}S(t)$ in the system is absent unless the system happens to be in equilibrium
with the medium. To see this most easily, we recount from \cite{GujratiNETI}
that
\begin{equation}
d^{(\text{E})}S(t)=\frac{dQ(t)}{T(t)} \label{Entropy_Change}%
\end{equation}
is the change in the entropy due to heat transfer $dQ(t)$ to the system that
is in internal equilibrium and depends on its instantaneous temperature
$T(t)$. We have used the superscript E to indicate that the entropy change we
are discussing is due to energy (heat) transfer for which the associated
conjugate variable, see Eq. (\ref{Field_Variables}), is $y_{\text{E}}=$
($\partial S/\partial E$)$_{V,N}=1/T(t)$ and takes the value $y_{0\text{E}}=$
$1/T_{0}$ in equilibrium$.$ It follows from this that the reversible entropy
change $d_{\text{e}}^{(\text{E})}S(t)$ in the system depends on its
\emph{equilibrium} temperature $T_{0}$, which is also the \emph{constant}
temperature of the medium, and is given by%
\begin{equation}
d_{\text{e}}^{(\text{E})}S(t)=\frac{dQ(t)}{T_{0}},
\label{Reversible_Entropy_Change}%
\end{equation}
regardless of the instantaneous temperature of the system. This results in
\cite{GujratiNETI}
\begin{equation}
d_{\text{i}}^{(\text{E})}S(t)=dQ(t)\left(  \frac{1}{T(t)}-\frac{1}{T_{0}%
}\right)  =F_{\text{E}}(t)dQ(t)\geq0 \label{Irreversible_Entropy_Change}%
\end{equation}
in all cases. Here,
\begin{equation}
F_{\text{X}}(t)\equiv F\left[  y(t)\right]  \equiv y(t)-y_{0}, \label{Force_y}%
\end{equation}
with $y=\allowbreak y_{0}$ representing the equilibrium value of the conjugate
field $y$ associated with the observable $X$, represents the
\emph{thermodynamic force} for the flow of $X$. \ Since $T(t)$ in Eq.
(\ref{Entropy_Change}) undergoes a change $dT\propto dQ(t)$ due to the heat
transfer, the heat transfer is not isothermal. Thus, it should not be
surprising that there is an irreversible entropy generation as part of it. On
the other hand, the determination of $d_{\text{e}}S(t)$ in Eq.
(\ref{Reversible_Entropy_Change}) requires the heat transfer to be
\emph{isothermal} for the process to be \emph{reversible}. What is surprising
is that $d_{\text{e}}S(t)$ is determined not by the current state of the
system at $t$, but its equilibrium state in the future so that the concept of
causality is inapplicable \cite{Gujrati-Symmetry}. It is this particular
aspect of $d_{\text{e}}S(t)$ that is the cornerstone of the second law of
thermodynamics: $d_{\text{i}}S(t)\geq0$ also depends on the future. We will
make use of this observations later.

The situation with other extensive variables like volume $V$ is no different.
As shown in \cite{GujratiNETI}, the entropy change and the reversible entropy
change due to a change $dV$ are given by
\begin{equation}
d^{(\text{V})}S(t)=\frac{P(t)}{T(t)}dV(t);\text{\ }d_{\text{e}}^{(\text{V}%
)}S(t)=\frac{P_{0}}{T_{0}}dV(t); \label{Entropy_Change_V}%
\end{equation}
It follows from this that the irreversible entropy generation $d_{\text{i}%
}^{(\text{V})}S(t)$ is given by an identical formulation as above for
$d_{\text{i}}^{(\text{E})}S(t)$
\begin{equation}
d_{\text{i}}^{(\text{V})}S(t)=dV(t)\left(  \frac{P(t)}{T(t)}-\frac{P_{0}%
}{T_{0}}\right)  =dV(t)F_{\text{V}}(t)\geq0, \label{Entropy_Change_V_i}%
\end{equation}
with $F_{\text{V}}(t)$ given by Eq. (\ref{Force_y}) with $y_{\text{V}%
}=P(t)/T(t)$, see also Eq. (\ref{Field_Variables_0}), represents the
thermodynamic force for the "flow" of volume $V$.

It should be stressed that the validity of Eq.
(\ref{Internal_Equilibrium_Relation}) follows from the continuity of the
entropy and the existence of the internal equilibrium. This is a general
relation and is valid for all systems whose macrostate can be described by the
three variables $E,V,$ and $N.$ On the other hand, the form of Eq.
(\ref{Entropy_Production}) depends on particular form of the processes that go
on in the system. Thus, it is process-specific; recall that we have eliminated
various processes in deriving this equation. We can incorporate the missing
additional contributions in Eq. (\ref{Entropy_Production}) by properly
introducing internal variables to describe the microstates of the system and
generalizing Eq. (\ref{Internal_Equilibrium_Relation}) to include other
extensive observables. The issue\ of internal variables is taken up in Sect.
\ref{Marker_Internal Variables}, and the generalization of Eq.
(\ref{Internal_Equilibrium_Relation}) is taken up in
Sect.\ref{marker_Rotating_Systems}.

\section{Entropy and Averages\label{marker_averages}}

\subsection{Isolated System}

The entropy of an isolated system such as $\Sigma_{0}$, whether in equilibrium
or not, is given by the Gibbs formulation%
\begin{equation}
S_{0}(t)=-%
{\textstyle\sum\limits_{\alpha}}
p_{\alpha}(t)\ln p_{\alpha}(t), \label{Gibbs entropy}%
\end{equation}
where $p_{\alpha}(t)$ is the time-dependent probability of the $\alpha$-th
microstate of the isolated system $\Sigma_{0}$; see a recent review
\cite{Gujrati-Symmetry} of this formulation for an isolated system, regardless
of how far it is from equilibrium. It is assumed that the dynamics of the
system is stochastic and not deterministic, as the latter dynamics makes the
above entropy a constant of motion in direct contradiction with the second
law. In a deterministic, i.e. Hamiltonian dynamics, a microstate uniquely
evolves into a microstate, while in a stochastic dynamics, a microstates
evolves into several microstates with certain probabilities. The unique
evolution is time-reversible, which results in the entropy being a constant of
motion. In a stochastic dynamic, the entropy can only increase
\cite{Gujrati-Symmetry}. Since the system is isolated, we do not allow
external forces acting on it; we of course neglect weak stochastic forces
acting on it. Thus, any deformation if it occurs is due to internal forces. We
have defined the entropy as a $\emph{dimensionless}$ quantity, which is
equivalent to setting the Boltzmann constant $k_{\text{B}}$ equal to $1$. The
collection $\boldsymbol{\alpha}=\left\{  \alpha\right\}  $ of these
microstates along with their \emph{non-zero} probabilities represents a\emph{
macrostate }$\mathcal{M}_{0}^{\boldsymbol{\alpha}}$ or simply $\mathcal{M}%
_{0}$ of $\Sigma_{0}$. As equilibrium and non-equilibrium thermodynamics is an
experimental science, a macrostate of any body is specified in terms of a set
of some \ extensive \emph{observables} $\mathbf{X}$ that can be manipulated by
the observer. The same set of observables are also used to identify the
microstates. Most often, the microstates are identified by the energy, volume,
and the number of particles that form the elements of $\mathbf{X}$, because of
their special role in thermodynamics. If there are other \emph{extensive}
mechanical quantities (quantities that are not purely thermodynamic in nature)
such as the total linear and angular momenta, then the microstates are
characterized by all these extensive quantities, collectively denote by
$\mathbf{X}$; see the discussion in Sect. \ref{Marker_Internal Variables} for
more details. Apart from these observables, a body can also be specified by a
set $\mathbf{I}$\emph{ }of \emph{internal variables}
\cite{Donder,deGroot,Prigogine,Prigogine-old} that have been found very useful
in describing glassy dynamics and will be discussed later in Sect.
\ref{Marker_Internal Variables}. Indeed, one needs internal variables to also
explain the time-evolution of an isolated system towards equilibrium since all
its observables remain constant. Thus, the macrostate of the isolated system
can use the internal variables for its specification. We will take these
internal variables to be also extensive and call both of them as \emph{state
variables }and collectively denote them by $\mathbf{Z}$.\ Taking internal
variables as extensive allows us to deal all state variables on equal footing,
so that generalization from observables to internal variables becomes almost trivial.

Let us continue with the discussion of the isolated system. In general,
microstate probabilities $p_{\alpha}(t)$ are functions of the state variables
$\mathbf{Z}_{0}$ along with $t$. As a consequence, the entropy $S_{0}%
(\mathbf{Z}_{0}(t),t)$ is also going to be a function of $\mathbf{Z}_{0}(t)$
and $t$. There are situations, when the entropy can also depend on some
\emph{external parameters} such as the angular velocity of the rotation of the
frame of reference. These parameters need not necessarily be extensive just as
$t$ is not. The number of state variables are too limited for a complete
microscopic description of the system, but is sufficient to describe the
macroscopic conditions of a system.

The observables remain fixed for the isolated system $\Sigma_{0}$. The
internal variable $\mathbf{I}_{0}(t)$, if present in the isolated system, is
normally a function of time; its time-dependence describes the temporal
evolution of $\Sigma_{0}$. While the microstate $\alpha$, hence\ the value of
the state variable $\mathbf{Z}_{0\alpha}$ in the microstate $\alpha$, does not
vary with time, the \emph{average }$\mathbf{Z}_{0}(t)$ for the macrostates
varies with $t$:%
\begin{equation}
\mathbf{Z}_{0}(t)\equiv%
{\textstyle\sum\limits_{\alpha}}
p_{\alpha}(t)\mathbf{Z}_{0\alpha}. \label{Average 0}%
\end{equation}
It should be pointed out that the entropy, as formulated in Eq.
(\ref{Gibbs entropy}) is also an \emph{average} of $(-\ln p)$
\cite{Gujrati-Symmetry}, the negative of the index of probability $\ln p$.
There will be times, when we will also use an overbar such as in
$\overline{\mathbf{Z}}_{0}(t)$ to indicate such averages for the sake of
clarity. For common thermodynamic quantities such as average energy, volume,
etc. the normal practice is to not use the overbar (unless clarity is needed)
as it is mostly these average quantities that we deal with.

\subsection{An Arbitrary Body}

It should also be stressed that the microstates for a body remain the same
whether the body is isolated or not. We can apply Eqs. (\ref{Gibbs entropy})
and (\ref{Average 0}) to determine the entropy and the average quantity for
any body, isolated or not (such as the system $\Sigma$ or the medium
$\widetilde{\Sigma}$, neither of which is isolated). In the following, we will
always use $i$ to denote a microstate of a body but reserve $\alpha$ to denote
the microstate of the isolated system. The entropy and the average energy of a
macrostate of a body\ is given by
\begin{subequations}
\label{0}%
\begin{align}
S(t)  &  \equiv-%
{\textstyle\sum\limits_{i}}
p_{i}(t)\ln p_{i}(t),\label{System_Entropy}\\
E(t)  &  \equiv%
{\textstyle\sum\limits_{i}}
p_{i}(t)E_{i}, \label{System_Energy}%
\end{align}
where $i$ denotes one of its microstates, whose probability is denoted by
$p_{i}(t)>0$. While we can certainly allow microstates with probabilities
$p_{i}(t)=0$ in Eq. (\ref{0}), we find it convenient to only allow microstates
with non-zero probabilities in the sum. Microstates with non-zero
probabilities will be identified as \emph{allowed} \cite{Gujrati-Symmetry} in
this work.

While there cannot be any doubt about the validity of Eq. (\ref{System_Energy}%
), one may feel some reservation about Eq. (\ref{System_Entropy}) for the
entropy of an open system. Therefore, we will give a direct proof of Eq.
(\ref{System_Entropy}) in Sect. \ref{marker_entropy_additivity}.

\section{Internal Equilibrium Thermodynamics}

\subsection{Equiprobability Concept and Consequences:\ No Internal Variables}

We will first consider the case when there are no internal variables. As is
the normal practice (see I for details), we assume that the medium is in
\emph{internal equilibrium }even if the system and the medium are not in
equilibrium. This assumption is similar to the assumption of local equilibrium
in the conventional nonequilibrium thermodynamics noted in Sect.
\ref{marker_averages}. As discussed in I and in the review
\cite{Gujrati-Symmetry}, the condition for the internal equilibrium to be met
is that the entropy has the maximum possible value at each instant for
the\emph{ instantaneous} average value $\mathbf{X}_{\text{IE}}\equiv
\mathbf{X}(t)\ $of the observable of the body. For the case considered in I,
they are $\widetilde{E}_{\text{IE}}\equiv\widetilde{E}(t)$ and $\widetilde
{V}_{\text{IE}}\equiv\widetilde{V}(t)$ for the medium; the number of particles
of the medium is not allowed to change at all, so that $\widetilde
{N}_{\text{IE}}\equiv\widetilde{N}$ is always kept fixed. It is easy to see
from Gibbs' formulation of the entropy in Eqs. (\ref{Gibbs entropy}) or
(\ref{System_Entropy}) that this happens \emph{if and only if } all the
microstates that are allowed ($p_{i}(t)>0$) at that instant are
\emph{equiprobable}. Another way to think about the internal equilibrium is to
imagine isolating the medium by disconnecting it from the system. This will
keep $\widetilde{\mathbf{X}}$ \emph{fixed} at $\widetilde{\mathbf{X}%
}_{\text{IE}}$. Then the entropy of the isolated medium cannot change anymore.
In other words, it is in \emph{equilibrium}. We can apply the same idea to any
body in internal equilibrium. The body will remain in equilibrium if it is isolated.

Let us follow the consequences of this concept beyond what was discussed in I.
\end{subequations}
\begin{enumerate}
\item[(1)] We allow for the possibility of external forces including stresses
acting on the system and produced by the medium; see Fig. \ref{Fig_Systems}.
Under internal equilibrium,
\begin{equation}
p_{i}(t)=1/W(t),\ \ \forall i, \label{Internal_equilibrium}%
\end{equation}
where $W(t)$ is the number of allowed microstates \cite{Gujrati-Symmetry} at
that instant. This immediately leads to the Boltzmann entropy
\begin{equation}
S(t)=\ln W(t) \label{Boltzmann entropy}%
\end{equation}
for a system in internal equilibrium, a common assumption in non-equilibrium
thermodynamics; see for example Bouchbinder and Langer \cite{Langer}. Since
the entropy is maximum at each instant $t$, it cannot increase further if all
observables are kept \emph{fixed }at their values $\mathbf{X}_{\text{IE}}$ at
that moment $t$, when the internal equilibrium is achieved. To ensure that
$\mathbf{X}$ is held fixed at $\mathbf{X}_{\text{IE}}$, we \emph{isolate} the
system from the medium so that $\Sigma$ turns into an isolated system. Let its
entropy be denoted by $S_{\text{IS}}(t)$ after isolation. As the entropy is
already at its maximum, it cannot change. In other words, $S_{\text{IS}}\equiv
S(\mathbf{X}_{\text{IE}})\ $at \emph{fixed} $\mathbf{X}_{\text{IE}}$\ must be
independent of time.

The above argument proves that the entropy $S_{\text{IE}}\equiv S(\mathbf{X}%
_{\text{IE}})$ has no \emph{explicit} $t$-dependence when the system is in
internal equilibrium:
\[
S_{\text{IE}}\equiv S(E_{\text{IE}},V_{\text{IE}},\cdots,N)=S(E(t),V(t),\cdots
,N).
\]

Its implicit time dependence when $\Sigma$ is treated as an open system comes
from the temporal evolution of $\mathbf{X}(t)$. Thus, the (maximum) entropy in
Eq. (\ref{Boltzmann entropy}) will change with time as $\mathbf{X}(t)$ changes
with time in the open system. For the isolated situation, $S_{\text{IS}}\equiv
S(\mathbf{X}_{\text{IE}})\ $at \emph{fixed} $\mathbf{X}_{\text{IE}}$\ will
remain constant.

\item[(2)] Since the entropy is maximum for fixed $\mathbf{X}_{\text{IE}}$,
there cannot be any \emph{additional} irreversible entropy production
$d_{\text{i}}S^{(\text{IE})}$ anymore%
\begin{equation}
\left.  d_{\text{i}}S_{\text{IE}}\right\vert _{\mathbf{X}_{\text{IE}}}=0\text{
in internal equilibrium (IE); } \label{IE_entropy_production}%
\end{equation}
here, the symbol $\left.  {}\right\vert _{\mathbf{X}_{\text{IE}}}$means that
$\mathbf{X}_{\text{IE}}$ are held fixed.

\item[(3)] When the system is in internal equilibrium, its various parts must
be in equilibrium with each other. Otherwise, there would be irreversible
entropy generation.

\item[(4)] It also follows from (3) that all the arguments that one uses to
follow the consequences of equilibrium can be applied to different parts of
$\Sigma$ that are in equilibrium. For example, the arguments that establish
that a system in equilibrium can only sustain uniform translation and rotation
\cite[Sect. 10]{Landau} can be applied without any change to a system in
internal equilibrium. As this result is going to play an important role in our
approach, we state it as a theorem.

\begin{theorem}
\label{marker_Uniform_Motion}There cannot be any relative motion between
different parts of $\Sigma$ for \emph{fixed} $\mathbf{X}_{\text{IE}}$\ in the
state of internal equilibrium. Thus, a system in internal equilibrium can only
sustain uniform translation and rotation \cite{Landau}.
\end{theorem}

\begin{proof}
We refer the reader to Landau and Lifshitz \cite[Sect. 10]{Landau} for the
details. Their argument applies without any change to an isolated system in
equilibrium. We easily extend their argument by considering our system at some
instant $t$ when it has $\mathbf{X}_{\text{IE}}=\mathbf{X}(t)$. We keep
$\mathbf{X}(t)$ fixed at $\mathbf{X}_{\text{IE}}$ by isolating the system from
the medium. The arguments now apply without any change to the system in
internal equilibrium. This proves the theorem.
\end{proof}

The axis of the uniform rotation must be a principal axis of the instantaneous
moment of inertia of the system. Otherwise, the system will undergo precession
in space \cite{Landau-Mechanics} and the rotation will not be uniform.

\item[(5)] Even with internal equilibrium in the system, there are both
elastic and inelastic or plastic deformations \cite[Sect. 7.4.2]{Kuiken},
which result in viscoelasticity in the system.
\end{enumerate}

If and only if the system is under internal equilibrium, the derivatives of
$S(t)$ with respect to $\mathbf{X}(t)$ have the significance of the
\emph{field} variables, also called the \emph{conjugate} variables, which we
will denote by $\mathbf{y}(t)$ or $\mathbf{Y}(t):$
\begin{equation}
\mathbf{y}(t)\equiv\frac{\mathbf{Y}(t)}{T(t)}\equiv\left(  \frac{\partial
S(t)}{\partial\mathbf{X}(t)}\right)  _{\mathbf{X}^{\prime}(t)}
\label{Field_Variables}%
\end{equation}
where $\mathbf{X}^{\prime}(t)$ denotes all other elements of $\mathbf{X}(t)$
except the one used in the derivative; compare with Eq. (\ref{Medium_Fields}).
The temperature $T(t)$ and the pressure $P(t)$ are defined in the customary
way by
\begin{equation}
y_{E}(t)\equiv\frac{1}{T(t)}=\left(  \frac{\partial S(t)}{\partial
E(t)}\right)  _{\mathbf{X}^{\prime}(t)},y_{V}(t)\equiv\ \ \frac{P(t)}%
{T(t)}=\left(  \frac{\partial S(t)}{\partial V(t)}\right)  _{\mathbf{X}%
^{\prime}(t)} \label{Field_Variables_0}%
\end{equation}
where $\mathbf{X}^{\prime}(t)$ denotes all other elements of $\mathbf{X}(t)$
except $E(t)$ or $V(t),$ respectively in the two derivatives. The pair of
quantities $\mathbf{X}(t),\mathbf{Y}(t)$ or $\mathbf{X}(t),\mathbf{y}(t)$ are
called \emph{conjugate} to each other.

The definitions of the conjugate fields give us an alternative way to
interpret the internal equilibrium. We imagine bringing the system in contact
with another medium whose field variables are also $\mathbf{Y}_{\text{IE}}$,
where
\[
\mathbf{y}_{\text{IE}}\equiv\frac{\mathbf{Y}_{\text{IE}}}{T_{\text{IE}}}%
\equiv\left.  \left(  \frac{\partial S(t)}{\partial\mathbf{X}(t)}\right)
\right\vert _{\mathbf{X}_{\text{IE}}};
\]
here $\left.  {}\right\vert _{\mathbf{X}_{\text{IE}}}$ denotes the value of
the derivative at $\mathbf{X}_{\text{IE}}$. To distinguish this medium from
the first medium that is characterized by $\mathbf{Y}_{0}=(T_{0},P_{0}%
,\cdots)$, we denote the first medium by $\widetilde{\Sigma}(\mathbf{Y}_{0})$,
and the new medium by $\widetilde{\Sigma}(\mathbf{Y}_{\text{IE}})$. The system
$\Sigma$ in internal equilibrium with observables $\mathbf{X}_{\text{IE}}$ is
in equilibrium with the new medium $\widetilde{\Sigma}(\mathbf{Y}_{\text{IE}%
})$. If the system is isolated by disconnecting it from $\widetilde{\Sigma
}(\mathbf{Y}_{\text{IE}})$, then there cannot be any change in the macrostate
of the system as all its observables remain constant at $\mathbf{X}%
_{\text{IE}}$. In other words, the system $\Sigma$ after being isolated
remains in equilibrium if it was originally in internal equilibrium, as noted
earlier. This will not be true if there are internal variables, to which we
will turn to in a moment.

\subsection{Zeroth and the Second Law}

All the above discussion has been for the entire system, but can be easily
extended to a system consisting of various subsystems $\sigma_{k}$, each in
internal equilibrium. Let us consider the system shown in Fig.
\ref{Fig_two_boxes_modified_system}, which consists of two subsystems
$\sigma_{1}$ and $\sigma_{2}$\ whose internal temperatures are $T_{1}(t)$ and
$T_{2}(t)>T_{1}(t)$, respectively. Let their respective energies be $E_{1}(t)$
and $E_{2}(t)$, with their sum denoted by $E(t)$. We consider all other
observables fixed for both subsystems. We first consider the system to be
isolated with no medium, as shown in Fig. \ref{Fig_two_boxes_modified_system}%
(b). Then $E$ remains constant, but not the individual energies. The
irreversible entropy gain for the entire system is%
\begin{equation}
d_{\text{i}}^{(\text{E})}S=dQ\left(  \frac{1}{T_{1}(t)}-\frac{1}{T_{2}%
(t)}\right)  >0 \label{Isolated_S_Generation}%
\end{equation}
during an infinitesimal heat transfer $dQ$ from the hotter subsystem to the
colder subsystem. Here,
\[
F_{\text{E}}(t)\equiv\frac{1}{T_{1}(t)}-\frac{1}{T_{2}(t)}%
\]
plays the role of the thermodynamic force $F_{\text{E}}(t)$ for the isolated
system, and should be contrasted with the same quantity for an open system in
Eq. (\ref{Force_y}). As the system is isolated, this is also the entropy
change $d^{\left(  \text{E}\right)  }S=d^{\left(  \text{E}\right)  }%
S_{1}+d^{\left(  \text{E}\right)  }S_{2}$ for the system, with
\[
d^{\left(  \text{E}\right)  }S_{1}=\frac{dQ}{T_{1}(t)},\ \ \ d^{\left(
\text{E}\right)  }S_{2}=-\frac{dQ}{T_{2}(t)}%
\]
for the two subsystems as follows from Eq. (\ref{Entropy_Change}). We now
consider the system to be in a medium at a fixed temperature $T_{0}$, as shown
in Fig. \ref{Fig_two_boxes_modified_system}(a). We take $T_{0}$ to be the
equilibrium temperature of the isolated subsystems; it is intermediate between
$T_{1}(t)$ and $T_{2}(t)$. The infinitesimal heat given out by the hotter
subsystem is now $dQ=dQ^{\prime}+dQ_{2}.$ The heat gained $dQ^{\prime}+dQ_{1}%
$by the colder subsystem must be exactly the heat loss $dQ$, since we are
dealing with an isolated system. Therefore, $dQ_{1}\equiv dQ_{2}$, so that the
entropy of the medium does not change. As the entropy change of the isolated
system is equal to that of the system, we have%
\begin{equation}
d^{\left(  \text{E}\right)  }S=\frac{dQ^{\prime}}{T_{1}(t)}-\frac{dQ^{\prime}%
}{T_{2}(t)}-\frac{dQ_{2}}{T_{2}(t)}+\frac{dQ_{1}}{T_{1}(t)}=dQ\left(  \frac
{1}{T_{1}(t)}-\frac{1}{T_{2}(t)}\right)  >0, \label{Medium_S_Generation}%
\end{equation}
since $dQ=dQ^{\prime}+dQ_{1}=dQ^{\prime}+dQ_{2}$. This is the same
irreversible entropy gain in Eq. (\ref{Isolated_S_Generation}) for the
isolated system in Fig. \ref{Fig_two_boxes_modified_system}(b). This should
not be surprising as none of the heat transfers is isothermal. Thus, bringing
the isolated system $\Sigma$, which consists of two subsystems, in contact
with a medium, characterized by the equilibrium temperature $T_{0}$ of
$\Sigma$, does not affect the irreversible entropy production. It is easy to
see that the arguments can be extended to many subsystems and to other field
variables. We will not pause here to do that.

At this moment, it is important to follow another important consequence of the
thermodynamic force $F_{\text{E}}(t)$, which vanishes if and only if the
system has come to \emph{thermal equilibrium}. This is the \emph{zeroth} law
of thermodynamics in terms of the internal temperatures of the two subsystems.
Thus, the internal instantaneous temperature plays the role of a thermodynamic
temperature in that the heat always flows from a hotter subsystem to a colder
subsystem. The above result can be easily generalized to many subsystems.

Let us now consider the volumes of the two subsystems to adjust as they come
to equilibrium. All other extensive observables are considered fixed. Then the
same reasoniong as above results in%
\[
d_{\text{i}}^{(\text{V})}S=dV(t)\left(  \frac{P_{1}(t)}{T_{1}(t)}-\frac
{P_{2}(t)}{T_{2}(t)}\right)  >0
\]
where $dV(t)$ is change in the volume of $\sigma_{1}$; the volume of the
isolated system remains unchanged. The corresponding thermodynamic force%
\[
F_{\text{V}}(t)\equiv\frac{P_{1}(t)}{T_{1}(t)}-\frac{P_{2}(t)}{T_{2}(t)}%
\]
vanishes when the system comes to \emph{mechanical equilibrium}. It usually
happens that thermal equilibrium requires mechanical equilibrium in that the
forces exerted on each other by any two subsystems must be equal and opposite.
Thus, the conditions for the equilibrium is that not only the pressure $P(t)$
but also the temperature $T(t)$ have the same value in both subsystems.

\subsection{Presence of Internal Variables}

We can easily extend the above discussion to include internal variables
$\mathbf{I}(t)$ by replacing $\mathbf{X}(t)$ by $\mathbf{Z}(t)$. In the
context of internal variables, their conjugate variables are known as
"\emph{affinity}." The general form of Eq. (\ref{Field_Variables}) is%
\begin{equation}
\mathbf{w}(t)\equiv\frac{\mathbf{W}(t)}{T(t)}\equiv\left(  \frac{\partial
S(t)}{\partial\mathbf{Z}(t)}\right)  _{\mathbf{Z}^{\prime}(t)},
\label{Field_Variables_1}%
\end{equation}
where $\mathbf{Z}^{\prime}(t)$ denotes all other elements of $\mathbf{Z}(t)$
except the one used in the derivative. The affinity $\mathbf{a}(t)$ is given
by
\begin{equation}
\mathbf{a}(t)\equiv\frac{\mathbf{A}(t)}{T(t)}\equiv\left(  \frac{\partial
S(t)}{\partial\mathbf{I}(t)}\right)  _{\mathbf{Z}^{\prime}(t)},
\label{Field_Variables_2}%
\end{equation}
so that $\mathbf{w}(t)$ consists of
\begin{equation}
\mathbf{y}(t)\equiv\left(  \frac{\partial S(t)}{\partial\mathbf{X}(t)}\right)
_{\mathbf{Z}^{\prime}(t)} \label{Field_Variables_3}%
\end{equation}
and $\mathbf{a}(t)$. The generalization of the thermodynamic force in Eq.
(\ref{Force_y}) is given by%
\begin{equation}
F_{\text{Z}}(t)\equiv F\left[  w(t)\right]  \equiv w(t)-w_{0}, \label{Force_w}%
\end{equation}
with $w_{0}$ representing the equilibrium value of $w$ corresponding to the
state variable $Z$.

Let us now consider the system to be in internal equilibrium, while the medium
containing it is $\widetilde{\Sigma}(\mathbf{Y}_{0},\mathbf{A}_{0})$, where
$\mathbf{Y}_{0},\mathbf{A}_{0}$ characterize the medium. If we now disconnect
the system from this medium, but bring it in contact with another medium
$\widetilde{\Sigma}(\mathbf{Y}_{\text{IE}},\mathbf{A}_{\text{IE}})$, where
$\mathbf{Y}_{\text{IE}}$ and $\mathbf{A}_{\text{IE}}$ are the field and
affinity vectors associated with the system in internal equilibrium, then the
system will remain in equilibrium with this medium. This is no different than
what we have said above in the absence of any internal variable $\mathbf{I}$.

But the situation is very different when we try to keep the system isolated.
Since internal variables are not under the control of the observer, they
cannot be manipulated to remain constant after isolation and will continue to
change. Thus, after the isolation, $\mathbf{X}_{\text{IS}}$, the value of
$\mathbf{X}(t)$ at the instance of isolation, will remain constant, but
$\mathbf{I}(t)$\ will not remain fixed at its value $\mathbf{I}_{\text{IS}}$
at the instant it was isolated. This time variation in the internal variables
is what drives this isolated system towards its equilibrium state during which
its entropy will continuously increase. This is very different from the case
above when there were no internal variables. Thus, a system in internal
equilibrium cannot be isolated as an equilibrium system if there are internal
variables present. It can only remain in equilibrium with the medium
$\widetilde{\Sigma}(\mathbf{Y}_{\text{IS}},\mathbf{A}_{\text{IS}})$.

\subsection{Maxwell Relations}

As the concept of internal equilibrium is not that different from the concept
of equilibrium, it should not come as a surprise that there are analogs of
Maxwell relations. We recall that in equilibrium thermodynamics, the standard
Maxwell relations for a system characterized by only $S$ and $V$ (fixed $N$)
are as follows in terms of Jacobians:
\begin{align}
\frac{\partial(T_{0},S,N)}{\partial(V,S,N)}  &  =\frac{\partial(P_{0}%
,V,N)}{\partial(V,S,N)},\ \ \ \frac{\partial(T_{0},S,N)}{\partial(P_{0}%
,S,N)}=\frac{\partial(P_{0},V,N)}{\partial(P_{0},S,N)},\nonumber\\
\frac{\partial(P_{0},V,N)}{\partial(T_{0},V,N)}  &  =\frac{\partial
(T_{0},S,N)}{\partial(T_{0},V,N)},\ \ \ \frac{\partial(T_{0},S,N)}%
{\partial(P_{0},T_{0},N)}=\frac{\partial(P_{0},V,N)}{\partial(P_{0},T_{0},N)}.
\label{Maxwell_Jacobians}%
\end{align}
All four Maxwell relations use the same numerators $\partial(T_{0},S,N)$ and
$\partial(P_{0},V,N)$. They use different denominators. Thus, they can all be
combined into one compact relation that can be simply written as
\begin{equation}
\partial(T_{0},S,N)=\partial(P_{0},V,N). \label{Maxwell_Compact}%
\end{equation}
Here, the relation only has a meaning if each side is divided by one of the
possible denominators $\partial(V,S,N),\partial(P_{0},S,N),\partial
(T_{0},V,N)$ and $\partial(P_{0},T_{0},N)$ on both sides.

We now consider a system. For simplicity, we assume that only one internal
variable, which we denote by $\xi$, characterizes this system. We assume the
system is in internal equilibrium. To simplify the notation, we will suppress
$N$ but use the additional variable $\xi$ along with the other two variables.
By considering the system in a medium $\widetilde{\Sigma}(\mathbf{Y}%
_{\text{IS}},\mathbf{A}_{\text{IS}})$, we recognize that the system is in
equilibrium. Thus, the Maxwell relations now for fixed $\xi$\ can be compactly
represented by%
\begin{equation}
\partial(T,S,\xi)=\partial(P,V,\xi), \label{Maxwell_Relation_TSPV/N}%
\end{equation}
by replacing $T_{0},P_{0}$ by $T(t)=T_{\text{IS}},P(t)=P_{\text{IS}},$ which
for simplicity have been written as $T,P$. The extension to the Maxwell
relation in terms of the internal variable requires considering the pair
$A,\xi$ in place of $T,S$ or $P,V$, where $A$ denotes the conjugate affinity
to $\xi$. For fixed $V$, the Maxwell relation is%
\begin{equation}
\partial(T,S,V)=\partial(A,\xi,V), \label{Maxwell_Relation_TSPV/N_1}%
\end{equation}
and for fixed $S$, the Maxwell relation is%
\begin{equation}
\partial(P,V,S)=-\partial(A,\xi,S). \label{Maxwell_Relation_TSPV/N_2}%
\end{equation}
If we consider the system in the medium $\widetilde{\Sigma}(\mathbf{Y}%
_{\text{IS}},\mathbf{A}_{\text{IS}})$, then we have the standard equilibrium
Maxwell relations similar to \ those in Eq. (\ref{Maxwell_Jacobians}).
However, we are interested in the possible "non-equilibrium" Maxwell relations
when the system is in the medium $\widetilde{\Sigma}(\mathbf{Y}_{0}%
,\mathbf{A}_{0})$. We first consider fixed $\xi$. The Maxwell relation in Eq.
(\ref{Maxwell_Relation_TSPV/N}) turns into the identity

\bigskip%
\begin{equation}
\frac{\partial(T,S,\xi)}{\partial(P_{0},S,\xi)}=\frac{\partial(P,V,\xi
)}{\partial(P_{0},S,\xi)} \label{Maxwell_Enthalpy/n}%
\end{equation}
that is
\[
\left(  \frac{\partial T}{\partial P_{0}}\right)  _{S,\xi}=\left(
\frac{\partial P}{\partial P_{0}}\right)  _{S,\xi}\left(  \frac{\partial
V}{\partial S}\right)  _{P_{0},\xi}-\left(  \frac{\partial V}{\partial P_{0}%
}\right)  _{S,\xi}\left(  \frac{\partial P}{\partial S}\right)  _{P_{0},\xi}%
\]
for the enthalpy and
\begin{equation}
\frac{\partial(T,S,\xi)}{\partial(T_{0},V,\xi)}=\frac{\partial(P,V,\xi
)}{\partial(T_{0},V,\xi)}, \label{Maxwell_Helmholtz/n}%
\end{equation}
that is
\[
\left(  \frac{\partial P}{\partial T_{0}}\right)  _{V,\xi}=\left(
\frac{\partial T}{\partial T_{0}}\right)  _{V,\xi}\left(  \frac{\partial
S}{\partial V}\right)  _{T_{0},\xi}-\left(  \frac{\partial S}{\partial P_{0}%
}\right)  _{V,\xi}\left(  \frac{\partial T}{\partial V}\right)  _{T_{0},\xi}%
\]
for the Helmholtz free energy; the details will be given in a separate
publications \cite{Gujrati-IV}. One can also obtain Maxwell relations at fixed
$V$ or $S$. For example, we find the following Maxwell relation
\[
\frac{\partial(P,V,S)}{\partial(\xi,V,S)}=-\frac{\partial(A,\xi,S)}%
{\partial(\xi,V,S)},
\]
that is
\[
\left(  \frac{\partial P}{\partial\xi}\right)  _{V,S}=\left(  \frac{\partial
A}{\partial V}\right)  _{\xi,S}.
\]
Similarly, from the Maxwell relation%
\[
\frac{\partial(T,S,V)}{\partial(\xi,S,V)}=\frac{\partial(A,\xi,V)}%
{\partial(\xi,S,V)},
\]
we find%
\[
\left(  \frac{\partial T}{\partial\xi}\right)  _{V,S}=-\left(  \frac{\partial
A}{\partial S}\right)  _{\xi,V}.
\]
The Maxwell relations in Eq. (\ref{Maxwell_Relation_TSPV/N}%
-\ref{Maxwell_Relation_TSPV/N_2}) contain the internal fields of the system
and not of the medium \cite{Gujrati-IV} when the system is out of equilibrium.
Obviously, the extensive variables in the relation must refer to the system.

\subsection{Internal Equilibrium Thermodynamics versus Local Thermodynamics}

We will argue in Sect. \ref{marker_entropy_additivity} that the concept of
internal equilibrium, which we adopt, is no different than the concept of
local equilibrium used in Eq. (\ref{local_equilibrium_entropy}) or in the
Gibbs fundamental relation in Eq. (\ref{Gibbs_fundamental_relation}). Despite
this, the two approaches based on the concept of local and internal
equilibrium, respectively, differ in some important ways that will be
elaborated later. Here, we briefly mention these differences. The first
important difference is that our approach is truly a statistical mechanical
approach for non-equilibrium systems. Once the probabilities of microstates
are known, the averages and other moments of all state variables that are used
to identify the microstates and the entropy are determined for the macrostate.
For example, the average fluctuation in $\mathbf{Z}$ for a body is given by%
\begin{equation}
\left[  \Delta\mathbf{Z}(t)\right]  ^{2}\equiv%
{\textstyle\sum\limits_{i}}
p_{i}(t)\left[  \mathbf{Z}_{i}-\overline{\mathbf{Z}}(t)\right]  ^{2},
\label{Fluctuations}%
\end{equation}
where%
\begin{equation}
\overline{\mathbf{Z}}(t)\equiv%
{\textstyle\sum\limits_{i}}
p_{i}(t)\mathbf{Z}_{i} \label{Average}%
\end{equation}
is the average $\overline{\mathbf{Z}}$ for the body; compare with Eq.
(\ref{Average 0}). As the microstate probabilities exist even when the system
is out of equilibrium, these averages including the entropy exist at all times
even if the system is not in internal equilibrium. Their temporal variations
are controlled by the dynamics governing the system and give rise to various
balance equations; see for example Eq. (\ref{Internal_Energy_Balance_Eq}). The
use of a probabilistic approach in the determination of the entropy and other
statistical properties means that the dynamics in the system must be
\emph{stochastic} and not deterministic, as the entropy remains constant in a
deterministic dynamics \cite{Gujrati-Symmetry}. A consequence of the
stochastic nature is that irreversible dissipation becomes an integral part of
the statistical description of any system, which then results in the law of
increase of entropy as captured by Eq. (\ref{Second_Law_0}). It is this law
that was the foundation of our approach in I, and which we develop further in
this work.

The second difference is in the identification of the thermodynamic potentials
for the open system, which has been discussed in I and will be further
elaborated in Sect. \ref{marker_thermodynamic_potentials} and again in Sect.
\ref{marker_inhomogeneous_arbitrary}. It is discovered in our approach that
thermodynamic potentials contain the field parameters (temperature, pressure,
chemical potentials,etc.) of the medium, which determine how far an open
system is from its equilibrium with the medium. These thermodynamic potentials
satisfy the second law in that they do \emph{not} increase in a spontaneous
process. In contrast, the form given in Eq. (\ref{Local_Gibbs_Free_Energy})
for the local Gibbs free energy density or its integral over the volume of the
system does not always satisfy this requirement; see the discussion
surrounding $\widehat{G}(t)$ in Eq. (22) of I. However, there is no
discrepancy for the internal energy in the two approaches as both approaches
give the same Gibbs fundamental relation. This is because the fundamental
relation in both approaches does \emph{not} depend on the\ field parameters of
the medium, but include the instantaneous field parameters of the system.

The third difference is in the reversible and irreversible entropy changes,
which depend on the equilibrium value $w_{0}$ ($y_{0}$ or $a_{0}$)\ of the
conjugate field $w$ ($y$ or $a$), as is easily seen from Eqs.
(\ref{Reversible_Entropy_Change}), (\ref{Irreversible_Entropy_Change}),
(\ref{Entropy_Change_V}) and (\ref{Entropy_Change_V_i}). In the
localthermodynamics, these quantities are determined by the local conjugate fields.

The presence of the medium field variable in the thermodynamic potentials of
the system does not mean that the situation would be any different if we
consider the system or subsystems to form the isolated system $\Sigma_{0}$
without any medium, as shown in Fig. \ref{Fig_two_boxes_modified_system}(b).
In this case, which we consider in Sect. \ref{marker_comparable_bodies}, we
again find that the thermodynamic potentials do not depend on the field
variables of the system or subsystems. The role of the field variables of the
medium $\widetilde{\Sigma}(\mathbf{Y}_{0},\mathbf{A}_{0})$ are now played by
the equilibrium conjugate variables $\mathbf{Y}_{0},\mathbf{A}_{0}$ of the
system or subsystems. This then leads us to the following important theorem:

\begin{theorem}
\label{Theorem_Isolated_Bodies}An isolated system $\Sigma$ is no different
than the open system $\Sigma$ in an extensively large medium $\widetilde
{\Sigma}(\mathbf{Y}_{0},\mathbf{A}_{0})$, provided the medium is appropriately
chosen to represent the equilibrium state (in terms of $\mathbf{Y}%
_{0},\mathbf{A}_{0}$) of the isolated system $\Sigma$.
\end{theorem}

In particular, the reversible entropy change and the irreversible entropy
generation in the two cases (a) and (b) in Fig.
\ref{Fig_two_boxes_modified_system} are exactly the same, as they both depend
on the equilibrium conjugate variables $\mathbf{Y}_{0},\mathbf{A}_{0}$ of the
system or subsystems. An example of this is already seen in Eqs.
(\ref{Reversible_Entropy_Change}) and (\ref{Irreversible_Entropy_Change}). We
defer the proof of this theorem to Sect. \ref{marker_comparable_bodies}.

\section{Additivity of Entropy and Quasi-independence
\label{marker_entropy_additivity}}

For simplicity of discussion, we consider all systems to be stationary in this
section, so that we only deal with internal energies. It was noted in Sect.
\ref{marker_averages} and recently reviewed in \cite{Gujrati-Symmetry}, the
entropy of an isolated body is given by the Gibbs formulation in Eq.
(\ref{Gibbs entropy}), regardless of whether it is in equilibrium or not.
There is no reason to believe that this formulation also applies to an
\emph{open} body under all conditions, though its applicability in equilibrium
is not in dispute \cite{Landau}. We now prove that this formulation also
applies to an open body under a condition that is always taken for granted. We
will specifically consider our system $\Sigma$ at some instant $t$, but the
conclusion is valid for all bodies.\ Let us consider all allowed microstates
of $\Sigma$ with fixed number of particles $N$; we index these microstates by
$i=1,2,\cdots,W(t)$. Theses microstates correspond to all possible energies
and volumes of the system. We use $\widetilde{\alpha}$ to denote the
microstates of $\widetilde{\Sigma}$ whose number of particles $\widetilde{N}$
is also fixed. A specification of the microstates $i$ and $\widetilde{\alpha}$
gives a unique microstate specification $\alpha$ representing a microstate of
the isolated system $\Sigma_{0}$. Hence, the number of microstates
$W_{0}(t)\ $of the $\Sigma_{0}$\ is the product
\begin{equation}
W_{0}(t)=W(t)\widetilde{W}(t), \label{Microstate_Numbers}%
\end{equation}
where $W(t)$ and $\widetilde{W}(t)$ are respectively the number of all
\emph{allowed} microstates \cite{Gujrati-Symmetry} of $\Sigma$\ and
$\widetilde{\Sigma},$ respectively at that instant $t$.\ 

As the concept of microstates does not depend on the nature of interactions
(they exist even in the absence of interaction), the above equation is valid
for all kinds of interactions. Let $E_{0},E(t)$ and $\widetilde{E}(t)$ denote
the internal energies of $\Sigma_{0},\Sigma$\ and $\widetilde{\Sigma}$,
respectively. Let $E_{0}^{(\text{int})}(t)$ denote the mutual interaction
energy between $\Sigma$\ and $\widetilde{\Sigma}$ at that instant. For
short-ranged interactions, this energy is determined by the surface $\partial
V(t)$ of $\Sigma$. For convenience, we assume that this entire area is exposed
to the surrounding medium, as shown in Fig. \ref{Fig_Systems}. If long-ranged
interactions are also present, or if the system size is very small, this
energy may depend on the entire volume $V(t)$ of $\Sigma$. In all cases, this
energy is defined by the following identity%
\begin{equation}
E_{0}\equiv E(t)+\widetilde{E}(t)+E_{0}^{(\text{int})}(t).
\label{Interaction_Energy}%
\end{equation}
Because of the smallness of $\Sigma$ relative to $\Sigma_{0}$,
$E(t)<<\widetilde{E}(t)$. If it happens that
\begin{equation}
E_{0}^{(\text{int})}(t)<<E(t), \label{Quasi-independence}%
\end{equation}
we call $\Sigma$\ and $\widetilde{\Sigma}$ \emph{quasi-independent}. For
quasi-independence, the linear size of the system must be at least as large
as, but hopefully larger than, the correlation length in the system. In this
case, we can neglect their mutual interactions, which is a common practice in
the discipline \cite{Landau}. The quasi-independence of the system and the
medium holds to a very high degree of accuracy for all short-ranged
interactions \cite{GujratiNETI}, provided the system itself is macroscopically
large so that the ratio of its surface to volume is insignificant. In most
cases, this will also ensure that the correlation length is small compared to
the size of the system. If there are also long-ranged interactions, then we
can still have quasi-independence provided these interactions are relatively
weak and shielding occurs and that Eq. (\ref{Quasi-independence}) and the
condition on the correlation length hold simultaneously.

Let us now assume that $\Sigma$\ and $\widetilde{\Sigma}$ quasi-independent.
In this case, the microstates of the two systems are independent of each other
to a very high degree of accuracy and we have (we suppress all state variables
for simplicity of notation) \
\[
p_{\alpha}(t)=p_{i}(t)p_{\widetilde{\alpha}}(t).
\]
Now, using
\[
\ln p_{i}(t)p_{\widetilde{\alpha}}(t)=\ln p_{i}(t)+\ln p_{\widetilde{\alpha}%
}(t),
\]
and the sum rule%
\[%
{\textstyle\sum\limits_{\widetilde{\alpha}}}
p_{\widetilde{\alpha}}(t)=1,\ \
{\textstyle\sum\limits_{i}}
p_{i}(t)=1,
\]
we find that
\[
S_{0}(t)\equiv-%
{\textstyle\sum\limits_{i}}
p_{i}(t)\ln p_{i}(t)-%
{\textstyle\sum\limits_{\widetilde{\alpha}}}
p_{\widetilde{\alpha}}(t)\ln p_{\widetilde{\alpha}}(t),
\]
where the two terms in the above equations represent the entropies of the
system and the medium%
\begin{equation}
S(t)=-%
{\textstyle\sum\limits_{i}}
p_{i}(t)\ln p_{i}(t),\ \ \ \widetilde{S}(t)=-%
{\textstyle\sum\limits_{\widetilde{\alpha}}}
p_{\widetilde{\alpha}}(t)\ln p_{\widetilde{\alpha}}(t), \label{Entropies}%
\end{equation}
respectively. This demonstration justifies the additivity of entropies%
\begin{equation}
S_{0}(t)=S(t)+\widetilde{S}(t) \label{Entropy Sum}%
\end{equation}
as a consequence of \emph{quasi-independence} so that%
\begin{equation}
E_{0}\equiv E(t)+\widetilde{E}(t) \label{Energy_Sum_Approx}%
\end{equation}
also holds to a very good approximation. Note that we have neither assumed the
medium nor the system to be in internal equilibrium in the above demonstration.

If the system and the medium fail to be quasi-independent because their mutual
interaction cannot be neglected, then Eq. (\ref{Quasi-independence}) is
violated. In this case, the presence of this interaction acts as a constraint
on $\Sigma_{0}.$ Consequently, the entropy now will be strictly less than the
above entropy in Eq. (\ref{Entropy Sum}). We denote this difference by
$S_{0}^{(\text{int})}(t)\leq0$ \cite{Gujrati-Athermal-Entropy}, which is
defined by the following identity%
\begin{equation}
S_{0}(t)\equiv S(t)+\widetilde{S}(t)+S_{0}^{(\text{int})}(t)\text{.}
\label{Entropy_Loss_Sum}%
\end{equation}
This identity reduces to the above additivity in Eq. (\ref{Entropy Sum})
provided
\begin{equation}
\left\vert S_{0}^{(\text{int})}(t)\right\vert <<S(t);
\label{Quasi-independence-Entropy}%
\end{equation}
compare with Eq. (\ref{Quasi-independence}). This inequality will in general
hold only if the interaction energy is also negligible.

If the strong inequality in Eq. (\ref{Quasi-independence-Entropy}) is not
satisfied, we have lost the additivity property of the entropy. Let us assume
that the strong inequality is satisfied for some large size of the system
$\Sigma$. Now, as the size of the system decreases, which is what will happen
on the way to considering physically infinitesimal volume elements used in Eq.
(\ref{Entropy_Additivity_Thermodynamics}), there comes a point where the
strong inequalities in Eqs. (\ref{Quasi-independence}) and
(\ref{Quasi-independence-Entropy}) are violated. This will destroy the
additivity of the entropy as exhibited in Eq. (\ref{Entropy Sum}).

The violation of entropy additivity occurs at intermediate sizes of the
system, somewhere between the macroscopic size where Eqs.
(\ref{Quasi-independence}) and (\ref{Quasi-independence-Entropy}) are valid,
and small local or microscopic size containing a small number $dN$ of
particles. For example, for $dN\approx10^{18},$ the surface to volume ratio
for the volume element $dV$ is about $10^{-6}$, implying an almost
imperceptible error\ in neglecting the interaction entropy $S_{0}%
^{(\text{int})}(t)$, provided the linear size of this region is large compared
to not only to the inter-particle separation \cite[p.1]{Landau-FluidMechanics}
but also the correlation length in the system. Under these conditions, the
integrand in Eq. (\ref{Entropy_Additivity_Thermodynamics}) truly refers to a
"physically" infinitesimal volume element containing a very large number of
particles. In this sense, our starting premise is similar to that adopted in
the conventional non-equilibrium thermodynamics
\cite{Donder,deGroot,Prigogine,Prigogine-old}. In particular, our concept of
internal equilibrium is no different than the concept of local equilibrium in
conventional non-equilibrium thermodynamics, as observed in I, except that we
require quasi-independence, which imposes the strong condition that not only
the interaction energy be small but also that the linear size be larger than
the correlation length. At present, there is some evidence that the
correlation length in a glass forming system appears to increasing as the
system approaches the glass transition \cite{Ottinger-2}.

The discussion above also clarifies that the additivity of entropy is a
consequence of the additivity of energy of various parts of the body and that
the interaction energies between them must be negligible. This additivity of
the energy and entropy was adopted in I. However, as we will be interested in
considering parts of $\Sigma$ as subsystems in this work, the additivity of
their entropy requires that their mutual interaction energies be also
negligibly small compared to their individual internal energies, and that
their linear sizes be large compared to the correlation lengths
\cite{Ottinger-2}. These requirements put a strong condition on the sizes of subsystems.

\section{Thermodynamic Potentials\label{marker_thermodynamic_potentials}}

\subsection{Fixed Number of Particles $N$ of the System $\Sigma$}

Despite similarities between our approach and that adopted in conventional
thermodynamics \cite{Donder,deGroot,Prigogine,Prigogine-old}, there were
important differences noted in I. One of these was the discovery that the
differences of the temperature and pressure of $\Sigma$\ and $\widetilde
{\Sigma},$ which are the same as that of $\Sigma$\ and $\Sigma_{0},$ play the
role similar to internal variables. The second difference was that the Gibbs
free energy in our approach exists even if the system is not in internal
equilibrium, and involves the temperature and pressure of $\Sigma_{0}$\ or
$\widetilde{\Sigma};$ of course,we assume$\ $that the medium is in internal
equilibrium; see I and Sect. \ref{marker_internal_equilibrium}. Under this
very weak assumption for the medium, its field variables such as the
temperature $T_{0}$, pressure $P_{0}$, etc. are well defined, and are
unaffected by whatever processes happen to be going on within the system
$\Sigma$ or whether $\Sigma$ is homogeneous or inhomogeneous. When the number
of particles $N$ in the system $\Sigma$ is held fixed, the appropriate
thermodynamic potential is the Gibbs free energy, which is identified as%
\begin{equation}
G(T_{0},P_{0},t)=E(t)-T_{0}S(t)+P_{0}V(t), \label{Gibbs_Free_Energy}%
\end{equation}
where the observables $E(t),S(t)$ and $V(t)$\ have explicit time-dependence
for fixed $N$. The particular form of the Gibbs free energy is in accordance
with the second law in Eq. (\ref{Second_Law_0}), and remains valid even if the
system is so far out of equilibrium that its temperature and pressure cannot
be defined; see also Landau and Lifshitz \cite[see Sect. 20]{Landau}. It is
not surprising, therefore, that it contains the temperature and pressure of
the medium, which are well defined. Theses quantities do not exists for the
system unless it happens to be at least in internal equilibrium. Even then,
the Gibbs free energy is given by Eq. (\ref{Gibbs_Free_Energy}) and contains
the temperature and pressure of $\widetilde{\Sigma}$. As such, it does not
represent a state function of the system. In engineering context, this
quantity is also known as \emph{exergy} or \emph{availability} \cite{Keenan}.

However, our form in Eq. (\ref{Gibbs_Free_Energy}) differs from the local form
of the Gibbs free energy in Eq. (\ref{Local_Gibbs_Free_Energy}), which
contains the local temperature and pressure. One can argue that the
identification of the Gibbs free energy in I was for the entire system, but
that once we account for the inhomogeneity by considering subsystems, the
Gibbs free energy for each subsystem will somehow become consistent with that
in Eq. (\ref{Local_Gibbs_Free_Energy}). This is a reasonable possibility and
we need to investigate this possibility. This issue is deferred to Sect.
\ref{marker_inhomogeneous_arbitrary}.

As the method to identify the Gibbs free energy or other thermodynamic
potentials that follow from the second law in Eq. (\ref{Second_Law_0}) is
going to be employed here several times, we briefly sketch the derivation for
the sake of continuity. The full details are given in I. We do not assume the
existence of the temperature, pressure, etc. of the system to include the
situation in our discussion when the system is far away form equilibrium so
that they are not defined. For simplicity, we consider a monatomic system of
structureless particles under no external shear as in the previous work
\cite{GujratiNETI}. Accordingly, we only consider the energy $E$, volume $V$
and the number of particles $N$ to describe the macrostate of the system at
any instant $t$. No internal variables will be considered at this moment. The
system and the medium are assumed not to be in equilibrium.

We use the additivity in Eq. (\ref{Entropy Sum}) to write the entropy
$S_{0}(t)$ of $\Sigma_{0}$ as the sum of the entropies $S(t)$ of the system
and $\widetilde{S}(t)$ of the medium. If we also assume that the latter is in
internal equilibrium, then we have%
\begin{equation}
S_{0}(E_{0},V_{0},N_{0},t)=S(E,V,N,t)+\widetilde{S}(\widetilde{E}%
,\widetilde{V},\widetilde{N}); \label{Total_Entropy}%
\end{equation}
there is no explicit $t$-dependence in $\widetilde{S}(\widetilde{E}%
,\widetilde{V},\widetilde{N})$ due to its internal equilibrium. With $N$ and
$\widetilde{N}$ fixed, we expand $S_{0}$ in terms of the small quantities
$E(t)$ and $V(t)$ of the system
\[
\widetilde{S}(\widetilde{E},\widetilde{V},\widetilde{N})\simeq\widetilde
{S}(E_{0},V_{0},\widetilde{N})-\left.  \left(  \frac{\partial\widetilde{S}%
}{\partial\widetilde{E}}\right)  \right\vert _{\mathbf{X}_{0}}E(t)-\left.
\left(  \frac{\partial\widetilde{S}}{\partial\widetilde{V}}\right)
\right\vert _{\mathbf{X}_{0}}V(t),
\]
where the derivatives are evaluated at $E_{0},V_{0},\widetilde{N}$. However,
as $\widetilde{N}$ is very close to $N_{0}$, there is no harm in evaluating
the derivatives at $E_{0},V_{0},N_{0}$. This is the reason that we have used
$\mathbf{X}_{0}$ above for the derivative. This approximation will be made
throughout in this work. The error is inconsequential when the system is a
very small part of the isolated system. It follows from the internal
equilibrium of $\widetilde{\Sigma}$ that%
\begin{equation}
\left.  \left(  \frac{\partial\widetilde{S}}{\partial\widetilde{E}}\right)
\right\vert _{\mathbf{X}_{0}}=\frac{1}{T_{0}},\ \ \left.  \left(
\frac{\partial\widetilde{S}}{\partial\widetilde{V}}\right)  \right\vert
_{\mathbf{X}_{0}}=\frac{P_{0}}{T_{0}}. \label{Medium_Fields}%
\end{equation}
We observe that $\widetilde{S}\equiv\widetilde{S}(E_{0},V_{0},\widetilde{N})$
is a constant, which is independent of the system $\Sigma$. Thus,%
\begin{equation}
S_{0}(t)-\widetilde{S}\simeq S(E,V,N,t)-E(t)/T_{0}-P_{0}V(t)/T_{0}.
\label{Total_Subtracted_Entropy_0}%
\end{equation}
In terms of
\begin{equation}
G(t)\equiv H(t)-T_{0}S(t),\ H(t)\equiv E(t)+P_{0}V(t), \label{Free_Energies}%
\end{equation}
we finally have
\begin{equation}
S_{0}(t)-\widetilde{S}=S(t)-H(t)/T_{0}=-G(t)/T_{0}.
\label{Gibbs_Free_Energy_Entropy_Relation}%
\end{equation}

It immediately follows from Eq. (\ref{Second_Law_0}) that the Gibbs free
energy $G(t)$ of the system in Eq. (\ref{Gibbs_Free_Energy}) continuously
decreases as the system relaxes towards equilibrium, a result quite well known
in classical thermodynamics \cite{Landau}:
\begin{equation}
\frac{dG(t)}{dt}\leq0. \label{Gibbs_Free_Energy_Variation}%
\end{equation}
The function $G(t)$ continues to decrease and finally becomes identical to the
equilibrium Gibbs free energy at the current temperature and pressure
$T_{0},P_{0}$. If we abruptly change the temperature $T^{\prime}$ and pressure
$P^{\prime}$ of the system in some state A$^{\prime}$, where the system was in
equilibrium, to a new state A where the temperature and pressure are
$T_{0},P_{0},$ respectively, at time $t=0$, then the \emph{initial} values of
the energy, volume and entropy at the new temperature and pressure remain
equal to their respective equilibrium values in the previous state A$^{\prime
}$ as the microstate probabilities $p_{\alpha}(t)$ at $t=0\ $have not had any
time to change. Thus, initially
\[
G(0)=E_{\text{A}}^{\prime}-T_{0}S_{\text{A}}^{\prime}+P_{0}V_{\text{A}%
}^{\prime}%
\]
in the state A; the quantities with a prime are the equilibrium values in the
state A$^{\prime}$. The Gibbs free energy decreases in accordance with Eq.
(\ref{Gibbs_Free_Energy_Variation}) and eventually becomes equal to its new
equilibrium value%
\[
G_{\text{A}}=E_{\text{A}}-T_{0}S_{\text{A}}+P_{0}V_{\text{A}},
\]
where the quantities with the subscript A denote the equilibrium values in the
new state A.

It should be noted that the equilibrium Gibbs free energy in the state
A$^{\prime}$ before the abrupt change is%
\[
G_{\text{A}}^{^{\prime}}=E_{\text{A}}^{\prime}-T^{\prime}S_{\text{A}}^{\prime
}+P^{\prime}V_{\text{A}}^{\prime},
\]
so that the Gibbs free energy undergoes a discontinuity at $t=0$ due to the
abrupt change:%
\[
\Delta G_{\text{A}}^{^{\prime}}=(T^{\prime}-T_{0})S_{\text{A}}^{\prime
}-(P^{\prime}-P_{0})V_{\text{A}}^{\prime}.
\]
Its magnitude and sign has nothing to do with the second law as the abrupt
change is not a spontaneous process.

A similar looking quantity $\widehat{G}(t)$, see Eq.
(\ref{Local_Gibbs_Free_Energy}) for its local analog in the local
non-equilibrium thermodynamics \cite{Donder,deGroot,Prigogine,Prigogine-old},
\begin{equation}
\widehat{G}(t)\equiv\widehat{H}(t)-T(t)S(t),\ \widehat{H}(t)\equiv
E(t)+P(t)V(t), \label{Free_Energies_Incorrect}%
\end{equation}
which can be defined \emph{only} when the system is under internal equilibrium
and not otherwise, was shown to increase with time \cite{GujratiNETI} during
relaxation
\begin{equation}
\frac{d\widehat{G}(t)}{dt}\geq0 \label{Gibbs_Free_Energy_Variation_Incorrect}%
\end{equation}
in a cooling process. Since it does not always decrease with time, it cannot
be taken as the Gibbs free energy; the latter is supposed to never increase as
the system equilibrates spontaneously as happens with $G(t)$; see Eq.
(\ref{Gibbs_Free_Energy_Variation}).

\subsection{Fixed Volume $V$ of the System $\Sigma$}

Instead of keeping the number of particles in $\Sigma$ fixed, let us keep its
volume $V$ fixed so that the volume of the medium is also kept fixed. The
number of particles $\widetilde{N}$ of the medium is no longer fixed. The
entropy $\widetilde{S}(\widetilde{E},\widetilde{V},\widetilde{N})$ of the
medium in Eq. (\ref{Total_Entropy}) is expanded in terms of small quantities
$E$ and $N$ of the system. We follow the steps similar to those above and
obtain%
\begin{equation}
\widetilde{S}(\widetilde{E},\widetilde{V},\widetilde{N})\simeq\widetilde
{S}(E_{0},\widetilde{V},N_{0})-\left.  \left(  \frac{\partial\widetilde{S}%
}{\partial\widetilde{E}}\right)  \right\vert _{\mathbf{X}_{0}}E(t)-\left.
\left(  \frac{\partial\widetilde{S}}{\partial\widetilde{N}}\right)
\right\vert _{\mathbf{X}_{0}}N(t), \label{Medium_Expansion_V}%
\end{equation}
where $\mathbf{X}_{0}$ stands for $E_{0},V_{0},N_{0}$ for reasons explained
above in deriving Eq. (\ref{Medium_Fields}). Let us now introduce the chemical
potential $\mu_{0}$\ of the particle in the medium by the standard definition
\[
\ \ \left.  \left(  \frac{\partial\widetilde{S}}{\partial\widetilde{N}%
}\right)  \right\vert _{\mathbf{X}_{0}}=-\frac{\mu_{0}}{T_{0}},
\]
We thus find that in terms of $\widetilde{S}\equiv\widetilde{S}(E_{0}%
,\widetilde{V},N_{0})$
\begin{equation}
S_{0}(t)-\widetilde{S}\simeq S(E,V,N,t)-E(t)/T_{0}+\mu_{0}N(t)/T_{0}%
=-[E(t)-T_{0}S(t)-\mu_{0}N(t)]/T_{0}, \label{Total_Subtracted_Entropy}%
\end{equation}
which identifies a different thermodynamic potential in this case as
\begin{equation}
\Omega(t)\equiv E(t)-T_{0}S(t)-\mu_{0}N(t); \label{Thrmo_Potential}%
\end{equation}
this thermodynamic potential also uses the field variables of the medium. We
should emphasize again that no assumption about the internal equilibrium of
the system has been made. The system may or may not be in internal
equilibrium. The application of the second law in Eq. (\ref{Second_Law_0}) now
gives%
\begin{equation}
\frac{d\Omega(t)}{dt}\leq0. \label{Second_Law_P(t)}%
\end{equation}

\subsection{Fixed $N$ and $V$ of the System $\Sigma$}

If both $N$ and $V$ are kept fixed, it is easy to follow the above derivation
to conclude that the Helmholtz free energy%
\[
F(t)\equiv E(t)-T_{0}S(t)
\]
must continuously decrease as the system reaches equilibrium:%
\[
\frac{dF(t)}{dt}\leq0.
\]
Thus, the second law for an open system is expressed in terms of different
thermodynamic potentials depending on which variables are held fixed.

\section{Internal variables\label{Marker_Internal Variables}}

As said above, a suitable equilibrium macrostate description of the system
requires a set of \emph{independent} macroscopic observables that can be
controlled by an experimentalist and whose values will allow the
experimentalist to differentiate between different macrostates of the same
system. It normally happens that experimentalists have a far less number of
external controls than the possible extensive variables that can be used to
characterize the macrostates. Thus, one does not characterize a macrostate,
especially an equilibrium macrostate, by specifying all of the relevant
extensive system quantities. For example, for a single component system, one
normally uses $E,V$ and $N$ to specify the macrostate if there are no external
shearing forces. Let us for the moment consider a system without external
shear. Usually, one considers a system with fixed $N$; then $E$ and $V$ can be
controlled by the two external variables $T_{0}$ and $P_{0}$ associated with
the medium. However, these external variables need not necessarily control the
local or internal structures in the system at all times during its evolution
towards equilibrium. As Frenkel has observed, the local structures can be
important when considering the structural relaxation in a glass or other
non-equilibrium systems \cite[p. 208]{Frenkel}. For example, one can consider
the average numbers of neighbors and next-neighbors of a given particle to
describe the local structure in the system. These quantities multiplied by the
number $N$ can play the role of internal variables. The corresponding
conjugate variables, normally identified as "chemical potentials" or
"affinity" for these internal variables usually vanish in equilibrium. Frenkel
goes on and calculates viscoelastic effects due to structural changes and
compares them with Maxwell's model of elastic relaxation\emph{ }or an
RC-circuit. This investigation by Frenkel\ \cite{Frenkel}\ shows that internal
variables can play an important role in the temporal evolution towards
equilibrium in some systems such as glasses. As such, they become an integral
part of the description of any non-equilibrium system and determine the
relaxation of the system \cite[Sect. 78]{Landau-FluidMechanics}. The internal
variables are also called hidden variables or internal order parameters.

To introduce the concept of internal variables, let us consider our isolated
system $\Sigma_{0}$ for which one can identify a set of \emph{conserved}
quantities, i.e. integrals of motion. For a mechanical system of $s$ degrees
of freedom, the number of such integrals of motion are $2s-1$
\cite{Landau-Mechanics}. Of these integrals of motion, those that are
\emph{additive} play an important role in thermodynamics and statistical
mechanics. The notable ones are the energy, and linear and angular momenta of
the system, among others such as the polarization, magnetization, etc. For the
moment, let us consider $\Sigma_{0}$ to be stationary. Its macrostate
$\mathcal{M}_{0}$ is characterized by fixed internal energy $E_{0}$, volume
$V_{0}$, particle number $N_{0}$ and other extensive observables, collectively
denoted by $\mathbf{X}_{0}$. All these observables are constant for
$\Sigma_{0}$. Let us consider the energy $E_{0}$, which is an integral of
motion. It usually happens (see below for an example) that there are many
different components $E_{0}^{(k)}$ of the energy whose total sum is the energy
of $\Sigma_{0}$:
\begin{equation}
E_{0}\equiv\sum_{k=1}^{n+1}E_{0}^{(k)}(t), \label{Energy_Sum}%
\end{equation}
where $n+1>1\ $is the number of\ energy components. It is $E_{0}$ that is a
constant of motion, not the individual components $E_{0}^{(k)}(t)$; the latter
will continue to change\ as the system evolves in time while maintaining Eq.
(\ref{Energy_Sum}). Let $W_{0}(\mathbf{X}_{0},t)$\ denote the number of
microstates corresponding to the macrostate $\mathcal{M}_{0}$ at time $t$. At
each instant $t$, the microstates in $W_{0}(\mathbf{X}_{0},t)$\ can be
partitioned into groups according to the possible values of $E_{0}^{(k)}(t)$.
Because of the sum rule in Eq. (\ref{Energy_Sum}), only $n$ of the components
are independent for a given $E_{0}$, which we take to be given by
$k=1,2,\cdots,n$. We will denote this set by an $n$-vector $\mathbf{I}_{0}(t)$
whose elements are$\left\{  I_{0}^{(k)}=E_{0}^{(k)}(t),k=1,2,\cdots,n\right\}
$. Then,
\[
E_{0}^{(n+1)}(t)\equiv E_{0}-\sum_{k=1}^{n}I_{0}^{(k)}(t).
\]
Let $W_{0}(\mathbf{X}_{0},\mathbf{I}_{0}(t),t)$ denote the number of
microstates for a given $\mathbf{Z}_{0}(t)$. These microstates define a new
macrostate, which we denote by $\mathcal{N}_{0}$. Obviously,
\begin{equation}
W_{0}(\mathbf{X}_{0},t)\equiv\sum_{\mathbf{I}_{0}(t)}W_{0}(\mathbf{X}%
_{0},\mathbf{I}_{0}(t),t); \label{NumberMicrostate_Sum}%
\end{equation}
the sum is over all possible $\mathbf{I}_{0}(t)$. As the system evolves,
different components $E_{0}^{(k)}(t)$ of $\mathbf{I}_{0}(t)$\ evolve in time
$t$, but $E_{0}$ and $\mathbf{X}_{0}$\ remain fixed. Thus, a better
understanding of the evolution of the system can be obtained by monitoring how
the various components $E_{0}^{(k)}(t)$ change in time. For this, it is better
to use $\mathbf{Z}_{0}\equiv(\mathbf{X}_{0},\mathbf{I}_{0}(t))$ to identify
the macrostate $\mathcal{N}_{0}$ even though individual $E_{0}^{(k)}(t)$
cannot be controlled by the observer. As $E_{0}$ can be controlled by the
observer, it is still the choice observable to be used for identifying a
macrostate.\ This is even more true for the isolated system for which $E_{0}$
is a constant of motion. The $n$ components of $\mathbf{I}_{0}(t)$\ then play
the role of \emph{internal variables} in developing non-equilibrium
thermodynamics of the isolated system.

It is evident that the same "extended" description should also be useful for
an open system $\Sigma$.\ The only difference between an open system and an
isolated system is that not all elements of $\mathbf{X}$ remain fixed. Some of
the observables, denoted by $\mathbf{X}^{\prime}$ are controlled by external
field parameters $\mathbf{Y}_{0}^{\prime}$ (such as $T_{0}$,$P_{0}$, etc.)\ of
the medium so that they do not remain fixed but continue to fluctuate about
their mean $\mathbf{X}^{\prime}(t)$ that keeps changing in time. However, at
least one of the extensive observables such as $N$ must be kept
\emph{constant} to quantify the size of the system
\cite{Gujrati-review:Fluctuations}. Thus, for an open system, these
observables can be replaced by \ the fields $\mathbf{Y}_{0}^{\prime}$, with
the remaining observables remaining constant. We will denote the latter
observables by $\mathbf{C}$ to remind us that they are constant. The open
system can be either specified by $\mathbf{X}^{\prime}(t),\mathbf{C}$ or
$\mathbf{Y}_{0}^{\prime},\mathbf{C}$. However, for the sake of convenience, we
will continue to use $\mathbf{X}(t)$ rather than $\mathbf{X}^{\prime
}(t),\mathbf{C}$ or $\mathbf{Y}_{0}^{\prime},\mathbf{C}$. Let us now consider
$\Sigma$ which is not in internal equilibrium so that it undergoes internal
deformation due to relative motions between its various parts. If there are
external strains on the system, they can be controlled by us from the outside.
Hence, they will not be considered as internal variables. However, internal
stresses acting on various parts of the system when there are no external
strains on the system are beyond our control and must be treated as internal
variables in describing the system. As said earlier, we can describe the
internal forces acting on each part in terms of translation and rotation of
its various parts; see Sect. \ref{marker_Helmholtz_theorem}. These motions
must be described by the use of suitable internal variables such as the linear
and angular momenta, as was discussed in Sect.
\ref{marker_important_assumptions_I}.

As the internal variables are uncontrollable, their affinity in equilibrium
must vanish as we prove now as a theorem.

\begin{theorem}
\label{marker_chemical_potential}The affinity of an internal variable must
vanish in equilibrium.
\end{theorem}

\begin{proof}
It is sufficient to prove the theorem for an isolated system. Also, we will
prove it for the energy components in Eq. (\ref{Energy_Sum}). The extension to
the general case is a trivial extension and will not be done here. As we are
dealing with equilibrium, we consider equilibrium values of all the
quantities, which are going to be represented by suppressing the argument $t$
as they are stationary. We now construct the following partition function for
the isolated system%
\[
Z_{0}(\mathbf{X}_{0},\mathbf{A}_{0})\equiv\sum_{\mathbf{I}_{0}}W_{0}%
(\mathbf{X}_{0},\mathbf{I}_{0})\exp\left\{  -\sum_{k=1}^{n}a_{0}^{(k)}%
I_{0}^{(k)}\right\}  ,
\]
where $\mathbf{a}_{0}$\ is the the $n$-vector $\left\{  a_{0}^{(k)}\right\}  $
equilibrium affinity. Such a partition function correctly describes the
situation in which the $n$ internal variables are not constant but keep
changing from microstate to microstate. We now observe that this partition
function reduces to the equilibrium value\ $W_{0}(\mathbf{X}_{0})$ in Eq.
(\ref{NumberMicrostate_Sum}) (where we take the limit $t\rightarrow\infty$),
provided $A_{0}^{(k)}\equiv0$:
\[
A_{0}^{(k)}\equiv0,\ k=1,2,\cdots,n,
\]
for each of the internal variable in the set $\mathbf{I}_{0\text{,eq}}$. This
proves the theorem.
\end{proof}

The above theorem deals with equilibrium affinities, and says nothing about
the affinities of the internal variables when the system is out of equilibrium.

Let $\alpha$ denote one of the microstates associated with the macrostate
$\mathcal{M}_{0},$ and $\beta$ one of the microstates associated with the
macrostate $\mathcal{N}_{0}$. Then, using their probabilities $p_{\alpha}(t)$
and $p_{\beta}(t),$\ we can determine the entropy of the two macrostates using
the Gibbs formulation in Eq. (\ref{Gibbs entropy}):
\begin{subequations}
\label{Macrostate_entropy}%
\begin{align}
S_{0}(\mathbf{X}_{0},t)  &  \equiv-\sum_{\alpha}p_{\alpha}(t)\ln p_{\alpha
}(t),\label{M0_entropy}\\
S_{0}(\mathbf{X}_{0},\mathbf{I}_{0}(t),t)  &  \equiv-\sum_{\beta}p_{\beta
}(t)\ln p_{\beta}(t). \label{N0_entropy}%
\end{align}
For a macroscopically large system, the following standard statistical
mechanical arguments can be used to highlight the maximum of the summand in
Eq. (\ref{NumberMicrostate_Sum}). Let the maximum of the summand be denoted by
$M_{0}(t)$, which occurs for some particular value $\overline{\mathbf{I}}%
_{0}(t)$ of $\mathbf{I}_{0}(t):$%
\end{subequations}
\[
M_{0}(t)\equiv W_{0}(\mathbf{X}_{0},\overline{\mathbf{I}}_{0}(t),t)
\]
We separate the maximum contribution from the sum and rewrite Eq.
(\ref{NumberMicrostate_Sum}) as follows:%
\[
W_{0}(\mathbf{X}_{0},t)\equiv M_{0}(t)\left[  1+\sum_{\mathbf{I}_{0}%
(t)\neq\overline{\mathbf{I}}_{0}(t)}\frac{W_{0}(\mathbf{X}_{0},\mathbf{I}%
_{0}(t),t)}{M_{0}(t)}\right]  ,
\]
where the sum is over all remaining $\mathbf{I}_{0}(t)$. It is normally the
case that the ratio in the above sum is vanishingly small for a macroscopic
system and that the sum can be neglected. In this case, we have%
\begin{equation}
W_{0}(\mathbf{X}_{0},t)\approx W_{0}(\mathbf{X}_{0},\overline{\mathbf{I}}%
_{0}(t),t). \label{Equality_W0}%
\end{equation}
For a macroscopically large open system such as $\Sigma$, the above equation
is formally valid, except that we must replace $\mathbf{X}_{0}$ by
$\mathbf{X}(t),$which stands for $\mathbf{X}^{\prime}(t),\mathbf{C}$, and
$\overline{\mathbf{I}}_{0}(t)$ by $\overline{\mathbf{I}}(t)$:%
\begin{equation}
W(\mathbf{X}(t),t)\approx W(\mathbf{X}(t),\overline{\mathbf{I}}(t),t).
\label{Equality_W}%
\end{equation}

We now prove an important theorem about the nature of the entropy.

\begin{theorem}
\label{marker_absence_of_internal_equilibrium}The entropy expressed only in
terms of the observables when (independent) internal variables are present
must explicitly depend on $t.$
\end{theorem}

\begin{proof}
We first consider the isolated system $\Sigma_{0}$ and prove the theorem for
it. For $\overline{\mathbf{I}}_{0}(t)$ to be independent of (fixed)
$\mathbf{X}_{0}$, it must surely have an explicit dependence on time. In other
words, $\overline{\mathbf{I}}_{0}(t)$ must be a function of $\mathbf{X}_{0}$
and $t$. \ Let us assume that there is an explicit $t$-dependence in both
$W_{0}$-functions in Eq. (\ref{Equality_W0}). As the entropy of the macrostate
$\mathcal{M}_{0}$ is given by the sum over all microstates $W_{0}%
(\mathbf{X}_{0},t)~$in Eq. (\ref{System_Entropy}), it must explicitly depend
on $t$. Thus, the theorem is satisfied. If, however, neither of the $W_{0}%
$-functions in Eq. (\ref{Equality_W0}) have any explicit $t$-dependence, then
this is possible only if $\overline{\mathbf{I}}_{0}(t)$ becomes a function of
$\mathbf{X}_{0}$ as the left side is only a function of $\mathbf{X}_{0}$.
Since $\mathbf{X}_{0}$ is constant, $\overline{\mathbf{I}}_{0}(t)$ itself must
be constant. The latter is the situation in equilibrium:
\begin{equation}
\overline{\mathbf{I}}_{0\text{,eq}}=\overline{\mathbf{I}}_{0}(\mathbf{X}%
_{0})\text{ a constant}. \label{Equilibrium_I}%
\end{equation}

It follows that $\overline{\mathbf{I}}_{0}\equiv\overline{\mathbf{I}%
}_{0\text{,eq}}$ is no longer an independent variable when the system is in
equilibrium. Obviously, this case is not covered by the theorem since
$\overline{\mathbf{I}}_{0}$ is not independent. The entropy in this case is
given by the Boltzmann formulation, cf. Eq. (\ref{Boltzmann entropy}), and we
have from Eq. (\ref{Equality_W0})%
\[
S_{0}(\mathbf{X}_{0})\approx S_{0}(\mathbf{X}_{0},\overline{\mathbf{I}}%
_{0})\ \text{a constant}.
\]

Let us now consider the special case when the macrostate $\mathcal{N}_{0}$
satisfies the condition of internal equilibrium. In this case, $W_{0}%
(\mathbf{X}_{0},\overline{\mathbf{I}}_{0}(t),t)$ does not explicitly depend on
$t$ and should be written as $W_{0}(\mathbf{X}_{0},\overline{\mathbf{I}}%
_{0}(t))$ with $\overline{\mathbf{I}}_{0}(t)~$having an explicit time
dependence. This entropy is again given by\ Eq. (\ref{Boltzmann entropy}):%
\begin{equation}
S_{0}(\mathbf{X}_{0},\overline{\mathbf{I}}_{0}(t))=\ln W_{0}(\mathbf{X}%
_{0},\overline{\mathbf{I}}_{0}(t)). \label{N0_entropy_Internal_Equilibrium}%
\end{equation}
It now follows from Eq. (\ref{Equality_W0}) that $W_{0}(\mathbf{X}_{0},t)$
must have an explicit time-dependence due to the explicit $t$-dependence of
$\overline{\mathbf{I}}_{0}(t)~$in $W_{0}(\mathbf{X}_{0},\overline{\mathbf{I}%
}_{0}(t))$. This is because different values of $\overline{\mathbf{I}}_{0}(t)$
will result in different values of $W_{0}(\mathbf{X}_{0},\overline{\mathbf{I}%
}_{0}(t))$, which can be treated as $W_{0}(\mathbf{X}_{0},t)$ associated with
the macrostate $\mathcal{M}_{0}$ at different times. This is the first case
considered above. Thus, $S_{0}(\mathbf{X}_{0},t)$ will have an explicit
$t$-dependence even though $S_{0}(\mathbf{X}_{0},\overline{\mathbf{I}}%
_{0}(t))$ does not.

This proves the theorem for an isolated system.

Let us consider an open system such as $\Sigma$. Again, $\overline{\mathbf{I}%
}(t)$ must be a function of $\mathbf{X}(t)$ and $t$ to remain independent of
$\mathbf{X}(t)$. Let us assume that there is an explicit $t$-dependence in the
$W$-functions in Eq. (\ref{Equality_W}). As the entropy of the macrostate
$\mathcal{M}$ of $\Sigma$\ is given by the sum over all microstates
$W(\mathbf{X}(t),t)~$in Eq.(\ref{M0_entropy}), it must explicitly depend on
$t$. If, however, neither of the $W$-functions have any explicit
$t$-dependence, then $\overline{\mathbf{I}}(t)$ becomes a function of
$\mathbf{X}(t)$. In this case, it is not independent of $\mathbf{X}(t)$. This
situation is then not relevant for the theorem.

The possibility in which $\overline{\mathbf{I}}(t)$ is independent of
$\mathbf{X}(t)$, but $W(\mathbf{X}(t),\overline{\mathbf{I}}(t))$ has no
explicit $t$-dependence, when the system is under internal equilibrium, is
very important. Fixing $\mathbf{X}_{\text{IS}}\equiv\mathbf{X}(t)$ allows us
to think of the system as an isolated system. Now, we can use the argument
given above for the isolated system to conclude that different values of
$\overline{\mathbf{I}}(t)$ will result in different values of $W(\mathbf{X}%
_{\text{IS}},\overline{\mathbf{I}}(t))$, which can be treated as
$W(\mathbf{X}_{\text{IS}},t)$ associated with the macrostate $\mathcal{M}$ at
different times. In other words, the macrostate $\mathcal{M}$ does not
represent an internal equilibrium state. Thus, we conclude that \ a macrostate
$\mathcal{N}$ under internal equilibrium results in a macrostate $\mathcal{M}%
$; the latter is, however, not in internal equilibrium.

This proves the theorem.
\end{proof}

It follows from the above discussion that a general thermodynamic state can be
taken to be a function of internal variables along with other observables and
time $t$ when we deal with non-equilibrium states. For an open system in which
many of the observables are controlled by external field parameters
$\mathbf{Y}_{0}$ (such as $T_{0}$,$P_{0}$, etc.)\ of the medium, we can
express $\overline{\mathbf{I}}_{\text{eq}}$ either as
\[
\overline{\mathbf{I}}_{\text{eq}}=\overline{\mathbf{I}}(\mathbf{X}_{\text{eq}%
}^{\prime},\mathbf{C}),
\]
or as%
\[
\overline{\mathbf{I}}_{\text{eq}}\equiv\overline{\mathbf{I}}(\mathbf{Y}%
_{0}^{\prime},\mathbf{C}).
\]
Away from equilibrium, the internal variable $\overline{\mathbf{I}}%
(\mathbf{X}^{\prime}(t),\mathbf{C})$ differs from its equilibrium values
$\overline{\mathbf{I}}_{\text{eq}}$, and is normally treated as an independent
variable and plays an important role in the dynamics of the system as the
latter strives to reach equilibrium. Thus, it is not surprising that internal
variables are employed to specify the macrostate of a glass. In
non-equilibrium thermodynamics, this fact has been recognized for quite some
time \cite{Donder,deGroot,Prigogine,Prigogine-old}.

Internal variables can also be related to the presence of internal degrees of
freedom in the particles of interest. The internal degrees are more common in
polymers but can also occur in small molecules in the form of rotation about
some internal axes. An example will clarify the point much better. Consider a
polymerization process resulting in a system of polydisperse linear polymer
chains of average molecular weight $\overline{M}$ in a solution
\cite{GujratiRC}. The model is defined on a lattice of $N$ sites and volume
$V=Nv_{0}$, with $v_{0}$ a constant representing the volume occupied by a
lattice site. One normally uses $E$, $V$, $\overline{M}$ defined below in Eq.
(\ref{Molecular_weight}), and the number of chains $p$ as the standard
observables that can be used to identify the macrostate (equilibrium or not)
of the polymer solution. In turn, these quantities are controlled by the
temperature, pressure and the initiation-termination and propagation rates;
the last two can be related to the initiation-termination activity controlling
the number of endgroups, two for each polymer, and the middlegroup activity.
These activities determine the corresponding affinity or "chemical
potentials." Let $N_{\text{m}}\equiv N-N_{\text{v}}$ denote the number of
monomers, each monomer occupying a lattice site, in terms of the number of
voids or sites not covered by monomers $N_{\text{v}}$ so that
\begin{equation}
\overline{M}\equiv\frac{N_{\text{m}}}{p}. \label{Molecular_weight}%
\end{equation}
In terms of the number of middle groups $N_{\text{M}}\equiv N_{\text{m}}%
-2p,$\ or $N_{\text{m}}$,\ the number of chemical bonds in the $p$ polymers is
given by%
\[
N_{\text{B}}\equiv N_{\text{M}}+p=N_{\text{m}}-p.
\]

There are two kinds of energy in the model \cite{GujratiRC}. One kind of
energy is due to mutual interactions of voids (v) with the end\ (E) and middle
(M) groups, and the mutual interactions between chemically unbonded M and E.
Let $N_{ij},i,j=$v,M or E, denote the number of nearest-neighbor contacts
$ij,i\neq j$, and $\varepsilon_{ij}$ the corresponding interaction energies,
respectively. The other kind of energy is due to intrachain gauche bonds (g),
\ and hairpin turns (hp). Their energies are $E_{\text{g}}\equiv
\varepsilon_{\text{g}}N_{\text{g}}$ for gauche bonds and $E_{\text{hp}}%
\equiv\varepsilon_{\text{hp}}N_{\text{hp}}$ for hairpin turns; here
$N_{\text{g}},$ and $N_{\text{hp}}$ denote the number of gauche bonds and
hairpin turns and parallel bonds and $\varepsilon_{\text{g}},$ and
$\varepsilon_{\text{hp}}$ are their energies. In addition, there is a mutual
interaction energy between two parallel (chemical) bonds, which may belong to
the same or different polymers. Let $N_{\text{P}}\,$denote the number parallel
bonds, each of energy $\varepsilon_{\text{P}}$. Thus,%
\begin{equation}
E\equiv%
{\displaystyle\sum\limits_{i\neq j\text{:v,M,E}}}
\varepsilon_{ij}N_{ij}+\varepsilon_{\text{g}}N_{\text{g}}+\varepsilon
_{\text{hp}}N_{\text{hp}}+\varepsilon_{\text{P}}N_{\text{P}}\equiv%
{\displaystyle\sum\limits_{i\neq j\text{:v,M,E}}}
E_{ij}+E_{\text{g}}+E_{\text{hp}}+E_{\text{P}}, \label{Energy_Partition}%
\end{equation}
where we have introduced $E_{ij},E_{\text{g}},E_{\text{hp}}$ and $E_{\text{P}%
}$ with obvious definitions. We thus observe that the energy can be
partitioned into six extensive energies, five of which can be taken as
internal variables.

To summarize, we conclude that the quantities that \emph{cannot} be controlled
by the observer can be identified as the \emph{internal variables}. This
statement should not be taken literary as what is considered uncontrollable
today may not remain so in the future. Thus, to some degree, the decision to
identify the internal variables is left to the observer. For us, any variable
that cannot be controlled to have a fixed value when the system is out of
equilibrium will be taken as an internal variable \cite{Maugin}. It should
also be noted that the number of internal variables is not a unique number for
a given system. For example, to describe local structures in a monatomic
system \cite{Frenkel}, one can consider any number of neighboring particles
(neighbors, next-neighbors, next-to-next neighbors, and so on). Thus, a choice
will have to be made to see how many of them are useful in a given experiment
or investigation. This certainly gives rise to an additional complication in
the study of non-equilibrium system.

Our approach allows us to associate affinity in a formal sense with all
internal variables. This is how the classical non-equilibrium thermodynamics
has been developed \cite{Donder,deGroot,Prigogine,Prigogine-old}. As observed
by Landau and Lifshitz \cite{Landau-FluidMechanics}, the use of internal
variables in a modern way can be traced to Mandelstam and Leontovich
\cite{Mandelstam}; see also Pokrovski \cite{Pokrovski}. Under the internal
equilibrium assumption, Prigogine addresses the issue of internal variables
(orientation of a molecule, deformation due to flow, elastic deformation,
etc.) in Sect. 11, Chapter III of his classic book \cite{Prigogine-old}, or in
Sect. 10.4 in the modern version \cite{Prigogine}, and couples them to their
"chemical potentials" or affinities. Indeed, Prigogine and Mazur were the
first one to do this in their classic paper \cite{Prigogine-Mazur}; see also
Coleman and Gurtin \cite{Coleman}. The issue of the internal variables is also
discussed in Sect. 6, Ch. 10 in \cite{deGroot}. Pokrovski \cite{Pokrovski}
provides a very illuminating discussion of internal variables and their role
in determining the internal energy. Thus, we will treat internal variables as
additional thermodynamic extensive quantities or "observables" similar to the
number of chemical species in chemical reactions that can be controlled by
affinities or chemical potentials. More recently, the idea has also been
visited by Bouchbinder and Langer \cite{Langer}.

\section{Thermodynamics of a simple rotating
body\label{marker_Rotating_Systems}}

\subsection{General Case}

We will find it convenient for later use to consider observing a body in
different frames of reference; see also Appendices \ref{Appd_Frames} and
\ref{Appd_Rotating_System_0}. For concreteness, we consider the system
$\Sigma$ and assume that no internal variables and no other observables
besides the energy, volume and number of particles are present; the latter can
be added easily as we will discuss later. We will consider three special
frames: the lab frame denoted by $\mathcal{L}$, an intermediate frame
$\mathcal{I}$, with its axes parallel to those of \ and moving with respect to
$\mathcal{L}$ with a velocity $\mathbf{V}(t)$, and a frame $\mathcal{C}$ with
its origin common with $\mathcal{I}$ and rotating with respect to it with an
angular velocity $\boldsymbol{\Omega}(t)$. Let $\mathbf{R}(t)$ denote the
location of the origins of $\mathcal{I}$ and $\mathcal{C}$ in the lab frame
$\mathcal{L}$ at time $t$ with $\mathbf{R}(t=0)=0.$ Let $\mathbf{r}%
_{\mathcal{C}}(t)$ denote the coordinate of a particle of $\Sigma$\ in the
$\mathcal{C}$ frame, and $\mathbf{v}_{\mathcal{C}}(t)$ its velocity in this
frame at time $t$. Its coordinate $\mathbf{r}_{\mathcal{L}}$ in the lab frame
$\mathcal{L}$ is given by%
\begin{equation}
\mathbf{r}_{\mathcal{L}}\equiv\mathbf{R}+\mathbf{r}_{\mathcal{C}};
\label{L_C_r_Relation}%
\end{equation}
its velocity is given by Eq. (\ref{velocity_Relation}). As shown in the
Appendix \ref{Appd_Frames}, the energy of the particle in the two frames are
related as shown in Eq. (\ref{Energy_Transformation_particle}). Let us
consider $\mathcal{I}$\ to be the frame in which the center of mass of the
body is at the origin. we will call it the center of mass frame for the body.
Then, applying the above two relations to all the particles in the system and
averaging over all allowed microstates \cite{Gujrati-Symmetry}, which is
carried out later in Sect. \ref{marker_Statistical_averaging}, we obtain that
the energy of the system in the three frames are related as shown in Eqs.
(\ref{Internal_Energy}) and (\ref{Internal_Energy_Iframe}):
\begin{subequations}
\label{Internal_Energy_LICFrame}%
\begin{align}
E_{\mathcal{C}}  &  =E_{\mathcal{L}}-\frac{\mathbf{P}^{2}}{2M}-\mathbf{M}%
\cdot\mathbf{\Omega=}E_{\mathcal{I}}\mathbf{-\mathbf{M}\cdot\mathbf{\Omega}%
,}\label{Internal_Energy_LCFrame}\\
E_{\mathcal{I}}  &  =E_{\mathcal{L}}-\frac{\mathbf{P}^{2}}{2M},
\label{Internal_Energy_ICFrame}%
\end{align}
where $\mathbf{P}$ and $\mathbf{M}$ are introduced in Eq.
(\ref{System_Definition}). We have not used the overbar to express the
statistical averages as explained in Sect. \ref{marker_Statistical_averaging}
but is implied.

We first prove the following theorem:
\end{subequations}
\begin{theorem}
\label{Theorem_Equal_Entropy}The entropy of a system is the same in all three
frames $\mathcal{L}$,$\mathcal{I}$, and $\mathcal{C}$.
\end{theorem}

\begin{proof}
To prove the theorem, we proceed as follows. Only for simplicity of the
argument and presentation, we focus on a system with \emph{fixed} $V=V(t)$ and
$N$ at some instant $t$. The extension to considering other extensive
variables is trivial. Consider observing the system simultaneously at $t$ in
these frames. It is evident that corresponding to each pair $\mathbf{r}%
_{\mathcal{L}}\mathbf{,p}_{\mathcal{L}}$ of the coordinates and momenta of a
given particle in the lab frame $\mathcal{L}$ at this moment, there is a
unique pair $\mathbf{r}_{\mathcal{I}}\mathbf{,p}_{\mathcal{I}}$ and
$\mathbf{r}_{\mathcal{C}}\mathbf{,p}_{\mathcal{C}}$ in the other two frames.
This is true of all the particles. The collection of positions and momenta of
all the particles defines a point in the phase space. In classical statistical
mechanics, a microstate of the system is identified by a small volume element
of size $(2\pi\hbar)^{3N}$ about a point in the phase space. Thus,
corresponding to each microstate $i$ ($=i_{\mathcal{L}},i_{\mathcal{I}},$ or
$i_{\mathcal{C}}$) in one frame, there exists a \emph{unique} microstate in
the other two frames. The uniqueness of microstate-mapping ensures that their
probabilities in the three frames are also equal:%
\begin{equation}
p_{i_{\mathcal{L}}}=p_{i_{\mathcal{I}}}=p_{i_{\mathcal{C}}.}
\label{Equal_Microstate_Probabilities}%
\end{equation}

Let us consider all the microstates of the same energy $E_{\mathcal{C}}$ in
the $\mathcal{C}$ frame at time $t$, and let $W(t)\equiv$ $W(E_{\mathcal{C}%
},t)$ denote their number and $\boldsymbol{p}_{\mathcal{C}}(t)$ the set of
their probabilities (not to be confused with momenta $\mathbf{p}_{\mathcal{C}%
}$, etc.). Because of the uniqueness of the mapping of these microstates noted
above, not only the number of microstates in the three frames are the same at
that instant%
\[
W(t)\equiv W(E_{\mathcal{C}},t)\equiv W(E_{\mathcal{L}},t)\equiv
W(E_{\mathcal{I}},t),
\]
but also the set of their probabilities
\[
\boldsymbol{p}_{\mathcal{C}}(t)\equiv\boldsymbol{p}_{\mathcal{L}}%
(t)\equiv\boldsymbol{p}_{\mathcal{I}}(t);
\]
however, their energies are different as given by Eqs.
(\ref{Internal_Energy_LICFrame})(\ref{Internal_Energy_LCFrame}) and
(\ref{Internal_Energy_ICFrame}). This immediately shows that the entropies
using the general Gibbs formulation in Eq. (\ref{Gibbs entropy}) are equal in
the three frames:
\begin{equation}
S_{\mathcal{C}}(E_{\mathcal{C}},t)=S_{\mathcal{I}}(E_{\mathcal{I}%
},t)=S_{\mathcal{L}}(E_{\mathcal{L}},t), \label{Equal_Entropy}%
\end{equation}
whether the system is in internal equilibrium or not. This proves the theorem.
\end{proof}

It should be noted that the center-of-mass kinetic energy $\mathbf{P}^{2}/2M$
is the same for all microstates in $W(t)$. Similarly, it follows from Eq.
(\ref{Ec_Mc_Relation}) that even $\mathbf{\mathbf{M}\cdot\mathbf{\Omega}}$ is
the same for all microstates in $W(t);$ see also the discussion leading to Eq.
(\ref{Av_M}). Thus, the three energies only differ by some constants at each
instant $t$.

\subsection{System under Internal Equilibrium}

We now specialize and assume the existence of the internal equilibrium, so
that all microstates are equally probable%
\begin{equation}
p_{i}(t)=1/W(t),i=1,2,\cdots,W(t). \label{Instantaneous probabilities}%
\end{equation}
Hence, the three entropies are each equal to
\begin{equation}
S(t)=\ln W(t). \label{Entropy_Frames}%
\end{equation}

It follows from Eq. (\ref{Equal_Entropy}) that there is no reason to use
different subscripts to distinguish the entropies. Accordingly, we will use
$S$ to represent the entropies in different frames; their energy arguments
will of course depend on the frame of reference. The arguments $\mathbf{V}$
and $\mathbf{\mathbf{\Omega}}$ above are actually external parameters that are
not extensive. We will show below that the entropies in the $\mathcal{I}$ and
$\mathcal{L}$ frames are actually functions of extensive quantities
$\mathbf{P}$ and $\mathbf{\mathbf{M}}$ that are conjugate to $\mathbf{V}$ and
$\mathbf{\mathbf{\Omega}}$, respectively; cf. Eq. (\ref{Field_Variables}).

\subsection{Statistical Averaging over Allowed
Microstates\label{marker_Statistical_averaging}}

We now investigate the consequences of statistical averaging over microstates
with non-zero probabilities \cite{Gujrati-Symmetry}\ and show that its
consequences are the same as expressed in Eqs. (\ref{Internal_Energy_LICFrame}%
) and (\ref{Equal_Entropy}). We first note that $\mathbf{M}$ in Eq.
(\ref{System_Definition}) depends on the coordinates and momenta of the
particles, but this is not the case with $\mathbf{P}$, even though both are
extensive quantities. As $E_{\mathcal{C}}$ in Eq. (\ref{Internal_Energy}) or
(\ref{Internal_Energy_LCFrame}) is for a microstate determined by the
coordinates and momenta of the particles, we need to average it using
microstate probabilities in Eq. (\ref{Instantaneous probabilities}). Averaging
over various microstates relates the average energies in the two frame. We use
an overbar, see Eqs. (\ref{Average 0}) and (\ref{Average}), to denote the
average. We find that the same form also describes the desired relation
between the average energies:%
\begin{equation}
\overline{E}_{\mathcal{C}}(t)=\overline{E}_{\mathcal{L}}(t)-\frac
{\mathbf{P}(t)^{2}}{2M}-\overline{\mathbf{M}}(t)\cdot\mathbf{\Omega
}(t)=\overline{E}_{\mathcal{I}}(t)\text{ }\mathbf{-\overline{\mathbf{M}}%
}(t)\mathbf{\cdot\mathbf{\Omega}}(t)\mathbf{,} \label{Internal_Energy_0}%
\end{equation}
where
\begin{equation}
\overline{E}_{\mathcal{I}}(t)=\overline{E}_{\mathcal{L}}(t)-\frac
{\mathbf{P}^{2}}{2M}. \label{Internal_Energy_1}%
\end{equation}
The momentum $\mathbf{P}$, of course, does not require any averaging as noted
above. Eq. (\ref{Internal_Energy_0}) is valid at each instance $t$. We can
also take the statistical average of Eq. (\ref{Angular_momentum_C}) to obtain%
\begin{equation}
\overline{E}_{\mathcal{C}}(t)=\overline{E}_{\mathcal{L}}(t)-\frac
{\mathbf{P}(t)^{2}}{2M}-%
{\displaystyle\sum_{j}}
m_{j}\overline{\mathbf{r}_{j}\cdot\left(  \mathbf{v}_{j}\times\mathbf{\Omega
}\right)  }-%
\frac12
{\displaystyle\sum_{j}}
m_{j}\overline{(\mathbf{\Omega\times r}_{j})^{2}}, \label{Internal_Energy_2}%
\end{equation}
where the two sums are over all the particles in the system. Here
$\mathbf{r}_{j}$ and $\mathbf{v}_{j}$ are the instantaneous position and
velocity of the $j$th particle in a microstate with respect to the
$\mathcal{C}$ frame; we have suppressed the subscript $\mathcal{C}$ from
$\mathbf{r}_{j}$ and $\mathbf{v}_{j}$ for the sake of notational simplicity.
In the last equation, the third contribution is due to the relative motion of
the particles with respect to the $\mathcal{C}$ frame. Indeed, the average of
Eq. (\ref{Angular_Momentum_Work}) immediately yields%
\begin{equation}
\mathbf{\overline{\mathbf{M}}}(t)\mathbf{\cdot\mathbf{\Omega}}(t)=%
{\displaystyle\sum_{j}}
m_{j}\overline{\mathbf{r}_{j}\cdot\left(  \mathbf{v}_{j}\times\mathbf{\Omega
}\right)  }+%
{\displaystyle\sum_{j}}
m_{j}\overline{(\mathbf{\Omega\times r}_{j})^{2}}.
\label{Angular_Momentum_Work_0}%
\end{equation}
The third contribution in Eq. (\ref{Internal_Energy_2}) and the first
contribution in Eq. (\ref{Angular_Momentum_Work_0}) vanish when the system is
in internal equilibrium because of the absence of any relative motion in that
case; see Theorem \ref{marker_Uniform_Motion}.

Since $\overline{E}_{\mathcal{I}}$ also does not depend on the velocity
$\mathbf{V}$, a similar averaging of Eqs. (\ref{Ec_V_Relation}) and
(\ref{Ec_Mc_Relation}) gives us
\begin{subequations}
\label{Av_Derivatives}%
\begin{align}
\left(  \frac{\partial\overline{E}_{\mathcal{C}}(t)}{\partial\mathbf{V}%
(t)}\right)  _{\overline{E}_{\mathcal{I}}\mathbf{,}V,N,\mathbf{\Omega}}  &
=0\mathbf{,}\label{Av_P}\\
\left(  \frac{\partial\overline{E}_{\mathcal{C}}(t)}{\partial\mathbf{\Omega
}(t)}\right)  _{\overline{E}_{\mathcal{I}}\mathbf{,}V,N}  &  =-\overline
{\mathbf{M}}(t). \label{Av_M}%
\end{align}
Comparing the above equations with the equations in the Appendix
\ref{Appd_Rotating_System_0}, we see that there is no reason to make a
distinction between $\mathbf{\mathbf{M}}(t)$, used in the proof above, and
$\mathbf{\overline{\mathbf{M}}}(t)$ or the average energies and the energy
used above in the proof. This justifies not using overbars to indicate
statistical averages in Eq. (\ref{Internal_Energy_LICFrame}).

Since the entropy $S(t)$ in Eq. (\ref{Entropy_Frames}) is fixed for fixed
$\overline{E}_{\mathcal{I}},V,N,$ and $\mathbf{\mathbf{\Omega,}}$\ we can
express the above two derivatives at fixed $S$\textbf{ }instead of\textbf{
}fixed $\overline{E}_{\mathcal{I}}$:
\end{subequations}
\begin{equation}
\left(  \frac{\partial\overline{E}_{\mathcal{C}}(t)}{\partial\mathbf{V}%
(t)}\right)  _{S\mathbf{,}V,N,\mathbf{\Omega}}=0\mathbf{,}\left(
\frac{\partial\overline{E}_{\mathcal{C}}(t)}{\partial\mathbf{\Omega}%
(t)}\right)  _{S\mathbf{,}V,N,\mathbf{V}}=-\overline{\mathbf{M}}(t).
\label{Av_EC_Derivatives}%
\end{equation}
The above equation is similar to the well known result \cite[Sect. 11]{Landau}
in equilibrium statistical mechanics that the statistical average of the
derivatives of the the energy with respect to external parameters
($\mathbf{V}$ and $\boldsymbol{\Omega}$) should be taken at \emph{constant
entropy} and other extensive quantities. We have extended this result to
internal equilibrium now. Introducing the following standard derivatives%

\begin{equation}
\left(  \frac{\partial\overline{E}_{\mathcal{C}}(t)}{\partial S(t)}\right)
_{V,N,\mathbf{V},\mathbf{\Omega}}=T(t),\left(  \frac{\partial\overline
{E}_{\mathcal{C}}(t)}{\partial V(t)}\right)  _{S\mathbf{,}N,\mathbf{V}%
,\mathbf{\Omega}}=-P(t) \label{Av_EC_Derivatives_1}%
\end{equation}
defining the temperature and pressure of the system, we can write down the
following differential identity%
\begin{equation}
d\overline{E}_{\mathcal{C}}=T(t)dS(t)-P(t)dV(t)\ \mathbf{-\ }\overline
{\mathbf{M}}\ (t)\cdot d\boldsymbol{\Omega}(t)\boldsymbol{.}
\label{EC_differential}%
\end{equation}
It should be noted that because of Eq. (\ref{Av_P}), the average energy
$\overline{E}_{\mathcal{C}}(t)$ does not depend on the velocity of the frames
$\mathcal{I~}$and $\mathcal{C}$. Thus, there is no reason to keep $\mathbf{V}$
fixed in the various derivatives in Eqs. (\ref{Av_Derivatives}%
-\ref{Av_EC_Derivatives_1}).

For
\[
\overline{E}_{\mathcal{I}}=\overline{E}_{\mathcal{C}}+\mathbf{\overline
{\mathbf{M}}\cdot\mathbf{\Omega,}}%
\]
we find that
\begin{equation}
d\overline{E}_{\mathcal{I}}=T(t)dS(t)-P(t)dV(t)\ \mathbf{+\ }%
\boldsymbol{\Omega}(t)\cdot d\overline{\mathbf{M}}(t), \label{EI_differential}%
\end{equation}
which is an extension of the result given in Landau and Lifshitz \cite[Sect.
26]{Landau} to the internal equilibrium. The point to note is that the entropy
$S(t)$\ in the $\mathcal{I}$ frame is a function of the conjugate variable
$\overline{\mathbf{M}}(t)$ instead of $\boldsymbol{\Omega}(t)$. However, for
\[
\overline{E}_{\mathcal{L}}(t)=\overline{E}_{\mathcal{C}}(t)+\frac
{\mathbf{P}(t)^{2}}{2M}+\overline{\mathbf{M}}(t)\cdot\mathbf{\Omega}(t),
\]
we also find an additional contribution due to $\mathbf{V}$:%
\begin{equation}
d\overline{E}_{\mathcal{L}}(t)=T(t)dS(t)-P(t)dV(t)+\mathbf{V}(t)\cdot
d\mathbf{P}(t)\mathbf{+}\boldsymbol{\Omega}(t)\cdot d\overline{\mathbf{M}}(t)
\label{EL_differential}%
\end{equation}
given in terms of all extensive quantities. The additional contribution due to
the momentum differential $d\mathbf{P}(t)$\ is due to the velocity of the
system as a whole and is important to include in the lab frame. For example,
such a contribution is needed to describe the flow of a superfluid in which
the normal and superfluid components have different velocities so that the
superfluid cannot be considered at rest in any frame\ \cite[see Eq.
(130.9)]{Landau-FluidMechanics}. We will need to allow for this possibility
when we extend our approach of nonequilibrium thermodynamics to inhomogeneous
systems where different subsystems will undergo relative motion. It follows
form Eq. (\ref{EL_differential}) that the \emph{drift velocity} of the center
of mass of the system is given by%
\begin{equation}
\left(  \frac{\partial\overline{E}_{\mathcal{L}}(t)}{\partial\mathbf{P}%
(t)}\right)  _{S\mathbf{,}V,N,\overline{\mathbf{M}}}=\mathbf{V}(t).
\label{Drift_Velocity}%
\end{equation}
Similarly, the angular velocity is given by%
\begin{equation}
\left(  \frac{\partial\overline{E}_{\mathcal{L}}(t)}{\partial\overline
{\mathbf{M}}(t)}\right)  _{S\mathbf{,}V,N,\mathbf{P}}=\boldsymbol{\Omega}(t).
\label{Angular_Velocity}%
\end{equation}

We again observe that the entropy in the lab frame $\mathcal{L}$ is a function
of the extensive conjugate quantities $\mathbf{P}(t)$ and $\overline
{\mathbf{M}}(t)$ rather than the external parameters $\mathbf{V}$ and
$\mathbf{\mathbf{\Omega}}$.

From now on, we will not use the overbar to show statistical averages for the
sake of notational simplicity.

It is clear from Eq. (\ref{EC_differential}) that we must treat
$E_{\mathcal{C}}(t)$ as a function of $S(t),V(t)$ and $\boldsymbol{\Omega}(t)$
for constant $N$. Alternatively, we must treat $S(t)$ as a function of
$E_{\mathcal{C}}(t),V(t)$ and $\boldsymbol{\Omega}(t):$%
\begin{equation}
S_{\mathcal{C}}(t)\equiv S_{\mathcal{C}}(E_{\mathcal{C}}%
(t),V(t),\boldsymbol{\Omega}(t),N), \label{Functional_dependence}%
\end{equation}
which is identical to the functional dependence shown in Eq.
(\ref{Equal_Entropy}), except that we no longer have an explicit
$t$-dependence because of internal equilibrium. The important point to observe
is that the entropy is a function of not only the energy in the $\mathcal{C}$
frame, but is also a function of the angular velocity of the reference frame
when rotation is involved.

\subsection{Same Temperature and Pressure in Different Frames}

We now make an important observation. It follows from Eqs.
(\ref{EI_differential}) and (\ref{EL_differential}) that
\begin{equation}
\left(  \frac{\partial E_{\mathcal{I}}(t)}{\partial S(t)}\right)
_{V,N,\mathbf{M}}=\left(  \frac{\partial E_{\mathcal{L}}(t)}{\partial
S(t)}\right)  _{V,N,\mathbf{P},\mathbf{M}}=T(t), \label{EI_EL_Derivatives}%
\end{equation}
obtained by differentiating with respect to $S(t).$\ Similar equations are
obtained when we differentiate with respect to $V(t).$\ \
\begin{equation}
\left(  \frac{\partial E_{\mathcal{I}}(t)}{\partial V(t)}\right)
_{S\mathbf{,}N,\mathbf{M}}=\left(  \frac{\partial E_{\mathcal{L}}(t)}{\partial
V(t)}\right)  _{S\mathbf{,}N,\mathbf{V},\mathbf{M}}=-P(t).
\label{EI_EL_Derivatives_1}%
\end{equation}
These equations are identical to the derivatives in Eq.
(\ref{Av_EC_Derivatives_1}) and show that the internal temperature $T(t)$ and
the internal pressure $P(t)$ are the same in the three frames. Moreover, it is
the same entropy function $S(t)$ that appears in Eq. (\ref{EC_differential})
also appears in Eqs. (\ref{EI_differential}) and (\ref{EL_differential}). In
other words, the entropy is the same in all frames, except that the arguments
are different.

\section{Thermodynamics Potentials and Gibbs Fundamental Relation for a
Homogeneous System with Translational Motion\label{marker_homogeneous}}

\subsection{Thermodynamic Potentials for a System under Arbitrary
Conditions\label{marker_homogeneous_arbitrary}}

\subsubsection{Fixed number of particles $N$ of the system $\Sigma$}

Before we discuss the inhomogeneous case, let us consider the homogeneous
situation considered in I and revisited briefly in Sect.
\ref{marker_thermodynamic_potentials}, and extend it to the case when the
system $\Sigma$ moves as a whole with a linear momentum $\mathbf{P}$. We still
assume that $\Sigma_{0}$ is at rest. Because of the linear momentum
conservation, the linear momentum of the center of mass of $\widetilde{\Sigma
}$ is $-\mathbf{P}$. Thus, the centers of mass $\Sigma$ and $\widetilde
{\Sigma}$ are moving towards each other. For simplicity, we will assume the
absence of overall intrinsic rotation for $\Sigma$ and $\widetilde{\Sigma}$
individually. This can easily be incorporated as we do in the next section.
Then we only need to consider the orbital angular momentum $\mathbf{L}_{0}$ of
$\Sigma_{0}$\ in terms of the locations $\mathbf{R}$ and $\widetilde
{\mathbf{R}}$ of the centers of mass of $\Sigma$ and $\widetilde{\Sigma}$,
respectively. It is clear that $\mathbf{L}_{0}$ always vanishes since the
centers of mass\ of $\Sigma$ and $\widetilde{\Sigma}$ are moving towards or
away from each other so that $\mathbf{P}$\ and $\mathbf{R-}\widetilde
{\mathbf{R}}$\ are colinear :%
\[
\mathbf{R}\times\mathbf{P-}\widetilde{\mathbf{R}}\times\mathbf{P}=0.
\]

The (average) internal energies of $\Sigma$ and $\widetilde{\Sigma}$ in their
center of mass frames (the $\mathcal{C}$ frame) are
\begin{align*}
E^{\text{i}}  &  =E-\mathbf{P}^{2}/2M,\\
\widetilde{E}^{\text{i}}  &  =\widetilde{E}-\mathbf{P}^{2}/2\widetilde{M},
\end{align*}
while $E$ and $\widetilde{E}$ denote their total energies in the lab frame
$\mathcal{L}$, respectively; see Eq. (\ref{Internal_Energy_LICFrame}).
However, because of the extreme large size of $\widetilde{\Sigma}$, its mass
$\widetilde{M}$ satisfies the inequality $\widetilde{M}>>M$, so that we can
replace $\widetilde{E}^{\text{i}}$ by$\ \widetilde{E}~$without any appreciable
error. The entropy $S$ of $\Sigma$ is a function of the internal energy
$\widetilde{E}^{\text{i}};$\ however, this is not relevant for our argument
here if we are only interested in identifying the appropriate thermodynamic
potential for the system. The energy
\[
E_{0}=E+\widetilde{E}%
\]
of $\Sigma_{0}$ remains constant in time. As discussed above, the additivity
of energy is valid under the assumption that the interaction energy
$E_{0}^{(\text{int})}(t)$ between the system and the medium is negligible.
This ensures that the entropies are also additive. In the lab frame
$\mathcal{L}$, the entropies of the $\Sigma$ and $\widetilde{\Sigma}$ are
obtained by considering their entropies\ in respective rest frames
$\mathcal{C}$\ and $\widetilde{\mathcal{C}};$ they are $S(E^{\text{i}},V,N,t)$
and $\widetilde{S}(\widetilde{E}^{\text{i}},\widetilde{V},\widetilde{N},t)$,
respectively; recall that we have set $\mathbf{\Omega}=0$ for each of them$.$

Using the fact that the medium is under internal equilibrium, we modify Eq.
(\ref{Total_Entropy}) to reflect the dependence on internal energies to obtain%
\begin{equation}
S_{0}(E_{0},V_{0},N_{0},t)=S(E^{\text{i}},V,N,t)+\widetilde{S}(\widetilde
{E}^{\text{i}},\widetilde{V},\widetilde{N})\simeq S(E^{\text{i}}%
,V,N,t)+\widetilde{S}(\widetilde{E},\widetilde{V},\widetilde{N}).
\label{Total_Entropy_0}%
\end{equation}
We now expand and follow the steps in arriving at Eq.
(\ref{Total_Subtracted_Entropy_0}); the steps are unaffected by the motion of
$\Sigma$. We thus obtain%
\[
S_{0}(t)-\widetilde{S}\simeq S(E^{\text{i}},V,N,t)-E(t)/T_{0}-P_{0}V(t)/T_{0}%
\]
in terms of the energy and volume of the system. We can now identify the Gibbs
free energy and enthalpy in the lab frame $\mathcal{L}$ in terms of the energy
of the system:%
\begin{equation}
G(t)=E(t)-T_{0}S(t)+P_{0}V(t),\ H(t)=E(t)+P_{0}V(t); \label{G-H_lab}%
\end{equation}
compare with the Gibbs free energy in Eq. (\ref{Gibbs_Free_Energy}). Thus, the
second law in terms of the Gibbs free energy remains unchanged and is given by
Eq. (\ref{Gibbs_Free_Energy_Variation}).

In the center of mass frame $\mathcal{C}$ of the system, the Gibbs free energy
and the enthalpy of the system are given by%
\begin{equation}
G^{\text{i}}(t)=E^{\text{i}}(t)-T_{0}S(t)+P_{0}V(t),\ H^{\text{i}%
}(t)=E^{\text{i}}(t)+P_{0}V(t). \label{G-H_CM}%
\end{equation}
Note that the above functions depend on the internal energy and not the energy
of the system $\Sigma$. But they are not useful in the lab frame in which the
system is being observed. Thus, we conclude that the overall motion of the
system does not change the enthalpy and the Gibbs free energy; we must use the
appropriate energy in the frame of observation; the temperature and the
pressure of the medium are not affected by the choice of the frame as noted
near the end of Sect. \ref{marker_Rotating_Systems}. Similarly, the entropy of
the system is unaffected by the choice of the frame as shown by Theorem
\ref{Theorem_Equal_Entropy}.

\subsubsection{Fixed volume $V$ of the system $\Sigma$}

Let us assume that instead of keeping $N$ fixed, we keep the volume of the
system fixed. Then following the procedure given in Sect.
\ref{marker_thermodynamic_potentials}, we find find that the correct
thermodynamic potential now in the two frames $\mathcal{L}$ and $\mathcal{C}$
are\ given by%
\begin{align*}
\Omega(t)  &  \equiv E(t)-T_{0}S(t)-\mu_{0}N(t),\\
\Omega^{\text{i}}(t)  &  \equiv E^{\text{i}}(t)-T_{0}S(t)-\mu_{0}N(t),
\end{align*}
respectively.

\subsubsection{Fixed $N$ and $V$ of the system $\Sigma$}

Let us assume that we keep $N$ and $V$ fixed. Then following the procedure
given in Sect. \ref{marker_thermodynamic_potentials}, we find find that the
correct thermodynamic potential now in the two frames $\mathcal{L}$ and
$\mathcal{C}$ are\ given by the Helmholtz free energy%
\begin{align*}
F(t)  &  \equiv E(t)-T_{0}S(t),\\
F^{\text{i}}(t)  &  \equiv E^{\text{i}}(t)-T_{0}S(t),
\end{align*}
respectively.

\subsubsection{Extension to many state variables}

From now on, we will list energy, volume and particle number for any body
separately and use $\mathbf{X}$ and $\mathbf{Z}$\ to denote the rest of the
observables and state variables, respectively. We will only fix the number of
particles $N,\widetilde{N}$ but allow all other state variables to fluctuate.
In this case, $\widetilde{S}(\widetilde{E}^{\text{i}},\widetilde{V}%
,\widetilde{N})$ in Eq. (\ref{Total_Entropy_0}) must be written as
$\widetilde{S}(\widetilde{E}^{\text{i}},\widetilde{V},\widetilde{N}%
,\widetilde{\mathbf{Z}})$ and its expansion in terms of small quantities gives%
\[
\widetilde{S}(\widetilde{E}^{\text{i}},\widetilde{V},\widetilde{N}%
,\widetilde{\mathbf{Z}})\simeq\widetilde{S}(E_{0},V_{0},\widetilde
{N},\mathbf{Z}_{0})-\left.  \left(  \frac{\partial\widetilde{S}}%
{\partial\widetilde{E}}\right)  \right\vert _{0}E(t)-\left.  \left(
\frac{\partial\widetilde{S}}{\partial\widetilde{V}}\right)  \right\vert
_{0}V(t)-\left.  \left(  \frac{\partial\widetilde{S}}{\partial\widetilde
{\mathbf{Z}}}\right)  \right\vert _{0}\cdot\mathbf{Z}(t).
\]
Here, $\left.  {}\right\vert _{0}$ corresponds evaluating the derivative at
$E_{0},V_{0},N_{0},\mathbf{Z}_{0}$ ($\mathbf{X}_{0}$ and $\mathbf{I}_{0})$ so
that these derivatives are \emph{constant}, independent of the properties of
the system; see the discussion in deriving Eq. (\ref{Medium_Fields}).
Introducing the corresponding "chemical potential vector" $\boldsymbol{\mu
}_{0}$\ for $\mathbf{X}_{0}$ and the "affinity vector" $\mathbf{a}_{0}%
\equiv\mathbf{A}_{0}/T_{0}=0$ (see Sect. \ref{Marker_Internal Variables})\ for
$\mathbf{I}_{0}$ because of the internal equilibrium of $\widetilde{\Sigma}$
\begin{equation}
\left.  \left(  \frac{\partial\widetilde{S}}{\partial\widetilde{\mathbf{X}}%
}\right)  \right\vert _{0}=-\frac{\boldsymbol{\mu}_{0}}{T_{0}},\left.  \left(
\frac{\partial\widetilde{S}}{\partial\widetilde{\mathbf{I}}}\right)
\right\vert _{0}=\frac{\mathbf{A}_{0}}{T_{0}}=0,
\label{Chemical_Potential_Internal_Variables}%
\end{equation}
we can identify a new thermodynamic potential $G^{\mathbf{X}}(t)\equiv
-T_{0}[S_{0}(E_{0},V_{0},N_{0},\mathbf{Z}_{0},t)-\widetilde{S}(E_{0}%
,V_{0},\widetilde{N},\mathbf{Z}_{0})]:$%
\begin{equation}
G^{\mathbf{X}}(t)=E(t)-T_{0}S(t)+P_{0}V(t)+\boldsymbol{\mu}_{0}\cdot
\mathbf{X}(t)=G(t)+\boldsymbol{\mu}_{0}\cdot\mathbf{X}(t)
\label{General_Thermodynamic_Potential}%
\end{equation}
in the lab frame $\mathcal{L}$. As $E_{0},V_{0},N_{0},\mathbf{X}_{0}$ and
$\widetilde{N}$ are constant, we have%
\[
\frac{d}{dt}\widetilde{S}(E_{0},V_{0},\widetilde{N},\mathbf{Z}_{0})=\left.
\left(  \frac{\partial\widetilde{S}}{\partial\mathbf{I}_{0}}\right)
\right\vert _{0}\cdot\frac{d\mathbf{I}_{0}(t)}{dt}=0
\]
because $\mathbf{a}_{0}=0$. Thus, $\widetilde{S}(E_{0},V_{0},\widetilde
{N},\mathbf{Z}_{0})$ is a constant, and the second law tells us that%
\begin{equation}
\frac{dS_{0}}{dt}=-\frac{1}{T_{0}}\frac{dG^{\mathbf{X}}}{dt}\geq0,
\label{General_Thermodynamic_Potential_Variation0}%
\end{equation}
as expected in any spontaneous process. In the $\mathcal{C}$ frame, we will
instead have%
\begin{equation}
G^{\text{i}\mathbf{X}}(t)=E^{\text{i}}(t)-T_{0}S(t)+P_{0}V(t)+\boldsymbol{\mu
}_{0}\cdot\mathbf{X}(t)=G^{\text{i}}(t)+\boldsymbol{\mu}_{0}\cdot
\mathbf{X}(t). \label{General_Thermodynamic_Potential_i}%
\end{equation}
The important point to note is that chemical potential vector $\boldsymbol{\mu
}_{0}$ and the affinity vector $\mathbf{A}_{0}=0$\ are associated with the
medium, just as $T_{0},P_{0}$ are. The analogue of the thermodynamic potential
$\Omega(t)$ and $F(t)$ are%
\begin{align*}
\Omega^{\mathbf{X}}(t)  &  =\Omega(t)+\boldsymbol{\mu}_{0}\cdot\mathbf{X}%
(t),\Omega^{\text{i}\mathbf{X}}(t)=\Omega^{\text{i}}(t)+\boldsymbol{\mu}%
_{0}\cdot\mathbf{X}(t),\\
F^{\mathbf{X}}(t)  &  =F(t)+\boldsymbol{\mu}_{0}\cdot\mathbf{X}(t),F^{\text{i}%
\mathbf{X}}(t)=F^{\text{i}}(t)+\boldsymbol{\mu}_{0}\cdot\mathbf{X}(t),
\end{align*}
respectively. Again, it follows from the second law that%
\begin{equation}
\frac{d\Omega^{\mathbf{X}}}{dt}\leq0,\frac{dF^{\mathbf{X}}}{dt}\leq0,
\label{General_Thermodynamic_Potential_Variation}%
\end{equation}
as expected in any spontaneous process.

\subsection{Gibbs Fundamental Relation for a System under Internal
Equilibrium\label{marker_homogeneous_internal}}

\subsubsection{No extra state variables $\mathbf{Z}(t)$}

We will first assume that there are no internal variables, but the system
$\Sigma$ satisfies the condition of internal equilibrium so that
$S(E^{\text{i}},V,N)$ no longer depends explicitly on time. Then, we can
identify the temperature, pressure, and the chemical potential of the system
by
\begin{subequations}
\label{CM_Derivatives}%
\begin{align}
\left(  \frac{\partial S}{\partial E^{\text{i}}}\right)   &  =\frac{1}%
{T(t)},\ \label{CM_Derivatives-1}\\
\left(  \frac{\partial S}{\partial V}\right)   &  =\frac{P(t)}{T(t)}%
,\ \label{CM_Derivatives-2}\\
\left(  \frac{\partial S}{\partial N}\right)   &  =-\frac{\mu(t)}{T(t)}.
\label{CM_Derivatives-3}%
\end{align}
This allows us to write down the Gibbs fundamental relation for constant $N$
as%
\end{subequations}
\[
T(t)dS=dE^{\text{i}}(t)+P(t)dV(t)-+\mu(t)dN(t),
\]
which can be rearranged to write down the first law of thermodynamics%
\begin{equation}
dE^{\text{i}}(t)=T(t)dS-P(t)dV(t)+\mu(t)dN(t). \label{First_Law_int}%
\end{equation}
We now turn to Eq. (\ref{Internal_Energy_0}). In terms of the momentum of the
center of mass frame $\mathcal{C}$ of the system, we have%
\begin{equation}
dE^{\text{i}}=dE-\mathbf{V}\cdot d\mathbf{P}; \label{Diff_Energy_Relation}%
\end{equation}
recall that according to our assumption, the system has no intrinsic angular
momentum. This allows us to use Eq. (\ref{First_Law_int}) to write down the
differential form%
\begin{equation}
dE=T(t)dS+\mathbf{V}\cdot d\mathbf{P}-P(t)dV(t)+\mu(t)dN(t)
\label{First-Law_Total}%
\end{equation}
for the first law of thermodynamics in terms of the energy $E(t)$ rather than
the internal energy $E^{\text{i}}(t)$. This has some important consequences
and will be extremely useful in the following. The first consequence is that
it allows us to think of $S(E^{\text{i}},V,N)$ as a function of four variables
$S(E,\mathbf{P},V,N)$:%
\[
T(t)dS=dE-\mathbf{V}\cdot d\mathbf{P}+P(t)dV(t)-\mu(t)dN(t).
\]
This equation is the Gibbs fundamental relation relating the entropy with
$E(t)$ rather than $E^{\text{i}}(t)$. The second consequence, which follows
from the Gibbs fundamental relation is that
\begin{subequations}
\label{Lab_Derivatives}%
\begin{align}
\left(  \frac{\partial S}{\partial E}\right)   &  =\frac{1}{T(t)}%
,\label{Lab_Derivatives-1}\\
\left(  \frac{\partial S}{\partial\mathbf{P}}\right)   &  =-\frac
{\mathbf{V}(t)}{T(t)},\ \label{Lab_Derivatives-2}\\
\left(  \frac{\partial S}{\partial V}\right)   &  =\frac{P(t)}{T(t)}%
,\ \label{Lab_Derivatives-3}\\
\left(  \frac{\partial S}{\partial N}\right)   &  =-\frac{\mu(t)}{T(t)}.
\label{Lab_Derivatives-4}%
\end{align}
The \emph{drift velocity} $\mathbf{V}$\ (of the center of mass) of the system
is given a thermodynamic interpretation in terms of the derivative in Eq.
(\ref{Lab_Derivatives-2}) at fixed $E,V$ and $N$. For the case when the number
of particles is held fixed, $dN(t)\equiv0$ and the last term in $dE^{\text{i}%
}(t),dE(t)$ and $dS(t)$ will be absent. For fixed $V$, the third term in
$dE^{\text{i}}(t),dE(t)$ and $dS(t)$ will be absent, and so on.

\subsubsection{Inclusion of state variables $\mathbf{Z}(t)$}

In the presence of internal variables $\mathbf{X}$ and $\mathbf{I}$, the
extension of the above relations is quite obvious. Introducing the
instantaneous chemical potential vector $\boldsymbol{\mu}(t)$ associated with
$\mathbf{X}$ and the affinity vector $\mathbf{A}$\ associated with
$\mathbf{I}$ using
\end{subequations}
\begin{equation}
\left(  \frac{\partial S}{\partial\mathbf{X}}\right)  =-\frac{\boldsymbol{\mu
}(t)}{T(t)},\left(  \frac{\partial S}{\partial\mathbf{I}}\right)
=\frac{\mathbf{A}(t)}{T(t)}, \label{Chemical_Potential_Affinity}%
\end{equation}
\ we generalize the entropy differential to
\begin{subequations}
\label{Gibbs_Fundamental_Relation_System}%
\begin{align}
T(t)dS  &  =dE^{\text{i}}(t)+P(t)dV(t)-\mu(t)dN(t)-\boldsymbol{\mu}(t)\cdot
d\mathbf{X}(t)+\mathbf{A}(t)\cdot d\mathbf{I}%
(t)\label{Gibbs_Fundamental_Relation_System_Internal}\\
&  =dE(t)-\mathbf{V}(t)\cdot d\mathbf{P}(t)+P(t)dV(t)-\mu
(t)dN(t)-\boldsymbol{\mu}(t)\cdot\mathbf{X}(t)+\mathbf{A}(t)\cdot
d\mathbf{I}(t). \label{Gibbs_Fundamental_Relation_System_General}%
\end{align}

\section{Inhomogeneous System with Relative Motion}

\subsection{Subsystems undergoing Relative Motion}

We now consider the isolated system\ to be stationary so that it has
no\ linear and angular momenta. However, $\Sigma$ and $\widetilde{\Sigma}%
$\ may have relative motion so that they each may possess linear and angular
momenta that individually must cancel out:
\end{subequations}
\begin{equation}
\mathbf{P}+\widetilde{\mathbf{P}}=0,\ \mathbf{M}+\widetilde{\mathbf{M}}=0.
\label{Net_Momenta}%
\end{equation}
\ Moreover, we will now treat $\Sigma$\ as inhomogeneous and assume that it
can be decomposed into a \emph{collection} of a large number $N_{\text{S}}$ of
\emph{subsystems} $\sigma_{k}$, $k=1,2,\cdots,N_{\text{S}},$ which may be in
different macrostates to allow for inhomogeneity and for relative motion
between different subsystems and within each subsystem. Each subsystem is
still \emph{macroscopically} large so that we can not only introduce a
legitimate entropy function $s_{k}$ for the macrostate $\mathcal{M}_{k}$ via
Gibbs's formulation%
\[
s_{k}(e_{k}^{\text{i}},\boldsymbol{\Omega}_{k},n_{k},v_{k},\mathbf{z}%
_{k},t)\equiv-%
{\textstyle\sum\limits_{\alpha_{k}}}
p_{\alpha_{k}}(t)\ln p_{\alpha_{k}}%
\]
where $\alpha_{k}$ denotes one of the allowed microstates of the subsystem
$\sigma_{k}$ corresponding to the macrostate $\mathcal{M}_{k}$ characterized
by observables $e_{k}^{\text{i}},n_{k},v_{k}$ and $\mathbf{z}_{k}$, but we
also have these entropies satisfy the \emph{additive property}%
\begin{equation}
S(E^{\text{i}},\boldsymbol{\Omega},N,V,\mathbf{Z},t)=%
{\textstyle\sum\limits_{k=1}^{N_{\text{S}}}}
s_{k}(e_{k}^{\text{i}},\boldsymbol{\Omega}_{k},n_{k},v_{k},\mathbf{z}_{k},t),
\label{Entropy_Sum_0}%
\end{equation}
which requires their \emph{quasi-independence} as discussed in Sect.
\ref{marker_entropy_additivity} at each instant $t$. Using the entropy
$s_{k},$ we can introduce the appropriate thermodynamic functions, but care
must be exercised in identifying these functions in the lab frame
$\mathcal{L}$, where the experiments are done. This is because the energies
depend on the frame of reference, which will result in different values of the
energies and thermodynamic potentials in different frames, such as the lab
frame $\mathcal{L}$ and the rotating frame frame $\mathcal{C}_{k}$ attached to
the center of mass of $\sigma_{k}$, which is translating with a linear
velocity $\mathbf{v}_{k}(t)$ and rotating with an angular velocity $\Omega
_{k}(t)$.

To make further progress, we will make another assumption later as we did in I
that each subsystem is in \emph{internal equilibrium.} This occurs when all
microstates contributing to the entropy are \emph{equiprobable}%
\[
p_{\alpha_{k}}(t)=1/W_{k},\ \ \alpha=1,2,,\cdots W_{k}(t);
\]
here $W_{k}(t)$ represents the number of microstates of the subsystem
$\sigma_{k}$ at time $t$. Under the equiprobability assumption or the internal
equilibrium assumption,
\[
s_{k}(t)=\ln W_{k}(t),
\]
which is what one would obtain by applying the Boltzmann formulation of the
entropy \cite{Gujrati-Symmetry}. It also follows from Theorem
\ref{marker_Uniform_Motion} that the entire subsystem is uniformly translating
with a linear velocity $\mathbf{v}_{k}(t)$ and rotating with an angular
velocity $\Omega_{k}(t)$ so that the system in internal equilibrium is
\emph{stationary} in the frame $\mathcal{C}_{k}$.

\subsection{System under Arbitrary
Conditions\label{marker_inhomogeneous_arbitrary}}

The internal energy for each subsystem $\sigma_{k}$ is related to its energy
$e_{k}$ in the lab frame $\mathcal{L}$
\begin{equation}
e_{k}^{\text{i}}\equiv e_{k}-\mathbf{p}_{k}^{2}/2m_{k}-\mathbf{m}_{k}%
\cdot\boldsymbol{\Omega}_{k}; \label{subsystem_energy_relation0}%
\end{equation}
see Eq. (\ref{Internal_Energy_0}). Alternatively, we can use Eq.
(\ref{Internal_Energy_2}) to express this relation as
\begin{equation}
e_{k}^{\text{i}}=e_{k}-\mathbf{p}_{k}^{2}/2m_{k}-%
{\displaystyle\sum_{j_{k}}}
m_{j_{k}}\mathbf{r}_{j_{k}}\cdot\left(  \mathbf{v}_{j_{k}}\times
\mathbf{\Omega}_{k}\right)  -%
\frac12
{\displaystyle\sum_{j_{k}}}
m_{j_{k}}(\mathbf{\Omega}_{k}\mathbf{\times r}_{j_{k}})^{2},
\label{subsystem_energy_relation1}%
\end{equation}
where each sum is over $n_{k}$ particles in the subsystem $k$. The third term
in the last equation vanishes when the subsystem is in internal equilibrium as
commented earlier. We also have the additivity laws
\begin{subequations}
\label{Conservation_Laws}%
\begin{align}
V  &  =%
{\displaystyle\sum}
v_{k},N=%
{\displaystyle\sum}
n_{k},\mathbf{X}=%
{\displaystyle\sum}
\mathbf{x}_{k},\mathbf{I}=%
{\displaystyle\sum}
\mathbf{i}_{k},\label{Conservation_V}\\
E  &  =%
{\displaystyle\sum}
e_{k}=%
{\displaystyle\sum}
(e_{k}^{\text{i}}+\mathbf{p}_{k}^{2}/2m_{k}+\mathbf{m}_{k}\cdot\Omega
_{k}),\label{Conservation_E}\\
\mathbf{P}  &  =%
{\displaystyle\sum}
\mathbf{p}_{k},\mathbf{M}=%
{\displaystyle\sum}
(\mathbf{m}_{k}+\mathbf{r}_{k}\times\mathbf{p}_{k}), \label{Conservation_P}%
\end{align}
at each instant $t$. The angular momentum $\mathbf{l}_{k}\equiv$
$\mathbf{r}_{k}\times\mathbf{p}_{k}$\ is the orbital angular momentum of
$\sigma_{k}$ with $\mathbf{r}_{k},$ $\mathbf{p}_{k}$ representing the location
and momentum of the center of mass of $\sigma_{k}$, respectively, and should
not be confused with its intrinsic angular momentum $\mathbf{m}_{k}$
introduced in Eq. (\ref{System_Definition}). We should emphasize that the
additivity of the energy requires that the interaction energy between
subsystems be negligible. This condition is necessary for the entropy to be
additive as discussed in Sect. \ref{marker_entropy_additivity}. As a
consequence, we do not have the contribution analogous to $\psi(\mathbf{r})dV$
in Eq. (\ref{Energy_Additivity_Thermodynamics}). This distinguishes our
approach with that taken in local non-equilibrium thermodynamics
\cite{Donder,deGroot,Prigogine,Prigogine-old}.

We keep $n_{k}$ fixed for simplicity so that $N$ is also fixed, and allow
$e_{k}^{\text{i}},v_{k}$ and $\mathbf{z}_{k}$\ to vary in time. We assume as
above that $\widetilde{\Sigma}$ has well-defined field variables ($T_{0}%
,P_{0}$, etc.) which do not change with time. This is ensured by assuming
internal equilibrium for the medium. This is the only assumption we make
regarding the isolated system $\Sigma_{0}$ for which we have%
\end{subequations}
\[
E_{0}=\widetilde{E}+E,V_{0}=\widetilde{V}+V,N_{0}=\widetilde{N}+N,\mathbf{X}%
_{0}=\widetilde{\mathbf{X}}+\mathbf{X,I}_{0}=\widetilde{\mathbf{I}}%
+\mathbf{I}.
\]
Because the medium is taken to be in internal equilibrium, its energy is
related to its internal energy according to Eq. (\ref{Internal_Energy_0})%
\begin{equation}
\widetilde{E}^{\text{i}}\equiv\widetilde{E}-\mathbf{P}^{2}/2\widetilde
{M}-\widetilde{\mathbf{M}}\cdot\widetilde{\boldsymbol{\Omega}}=E-\mathbf{P}%
^{2}/2M-\mathbf{M}^{2}/2\widetilde{I}, \label{General_Ei_E_Relation_Medium}%
\end{equation}
where according to Eq. (\ref{Net_Momenta}), we have taken $\widetilde
{\mathbf{M}}=-\mathbf{M}$, and where $\widetilde{I}$ is the moment of inertia
of the medium about its axis of rotation. The axis of rotation must be one of
its principal axis of rotation; see the comment after the proof of Theorem
\ref{marker_Uniform_Motion} in Sect. \ref{marker_internal_equilibrium}. The
contribution coming from the internal motion, which is similar to the third
contribution in Eq. (\ref{subsystem_energy_relation1}) applied to the medium
vanishes because of its uniform translation and rotation following Theorem
\ref{marker_Uniform_Motion} applied to the medium. The contribution similar to
the last term in Eq. (\ref{subsystem_energy_relation1}) is the standard
rotational kinetic energy of the medium treated as a rigid body. The angular
momentum $\widetilde{\mathbf{M}}$ is given by%
\[
\widetilde{M}_{ij}=\widetilde{I}_{ij}\widetilde{\Omega}_{j}=-M_{ij}.
\]
Assuming that the motions are finite, we conclude that $\widetilde{M}_{ij}$
must be finite. Therefore, for an extremely large medium, $\widetilde
{\mathbf{\Omega}}$ and $\widetilde{\mathbf{V}}$\ must be extremely small,
which ensure that the last two terms in Eq.
(\ref{General_Ei_E_Relation_Medium}) are extremely small. This allows us to
approximate%
\begin{equation}
\widetilde{E}\simeq\widetilde{E}^{\text{i}}
\label{Medium_Energy_Internal_Energy}%
\end{equation}
without any appreciable error. As far as the system is concerned, the
relationship between its internal energy and the energy is still given by Eq.
(\ref{Internal_Energy})
\begin{equation}
E^{\text{i}}\equiv E-\mathbf{P}^{2}/2M-\mathbf{M}\cdot\boldsymbol{\Omega,}
\label{General_Ei_E_Relation}%
\end{equation}
except that the system may not be uniformly translating and rotating about any
of its principal axis of inertia. We take $\widetilde{N}$ and $N$ as
constants, but allow for $E,V,$ and $\mathbf{Z}$ ($\mathbf{X}$ and
$\mathbf{I}$) to change with time. Now, we use the entropy additivity
\[
S_{0}(E_{0}^{\text{i}},N_{0},V_{0},\mathbf{Z}_{0},t)=S(E^{\text{i}%
},\boldsymbol{\Omega,}N,V,\mathbf{Z},t)+\widetilde{S}(\widetilde{E}^{\text{i}%
},\widetilde{\boldsymbol{\Omega}}\boldsymbol{,}\widetilde{N},\widetilde
{V},\widetilde{\mathbf{Z}})
\]
and expand
\begin{align*}
\widetilde{S}(\widetilde{E}^{\text{i}},\widetilde{\boldsymbol{\Omega}%
}\boldsymbol{,}\widetilde{N},\widetilde{V},\widetilde{\mathbf{Z}})  &
=\widetilde{S}(E_{0},\widetilde{\boldsymbol{\Omega}}\boldsymbol{,}%
\widetilde{N},V_{0},\mathbf{Z}_{0})-E\left.  \frac{\partial\widetilde{S}%
}{\partial\widetilde{E}}\right\vert _{0}-V\left.  \frac{\partial\widetilde{S}%
}{\partial\widetilde{V}}\right\vert _{0}-\mathbf{Z\cdot}\left.  \frac
{\partial\widetilde{S}}{\partial\widetilde{\mathbf{Z}}}\right\vert _{0}\\
&  =\widetilde{S}(E_{0},\widetilde{\boldsymbol{\Omega}}\boldsymbol{,}%
\widetilde{N},V_{0},\mathbf{Z}_{0})-\frac{E}{T_{0}}-V\frac{P_{0}}{T_{0}%
}-\mathbf{X\cdot}\frac{\boldsymbol{\mu}_{0}}{T_{0}},
\end{align*}
by treating $E,V,\mathbf{X}$ and $\mathbf{I}$ as small parameters; here,
$\left.  {}\right\vert _{0}$ corresponds to evaluating the derivatives at
$E_{0},\widetilde{\boldsymbol{\Omega}}\boldsymbol{,}V_{0},N_{0},\mathbf{Z}%
_{0}$ ($\mathbf{X}_{0}$ and $\mathbf{I}_{0}$). We have also used the
definitions of the field variables of the medium in terms of the derivatives
of the entropy $\widetilde{S}$ in Eqs. (\ref{Medium_Fields}) and
(\ref{Chemical_Potential_Internal_Variables}), and set $\mathbf{A}_{0}=0$ as
established in Sect. \ref{Marker_Internal Variables}. We thus finally obtain%
\begin{align*}
S_{0}(E_{0},N_{0},V_{0},\mathbf{Z}_{0},t)-\widetilde{S}(E_{0}\widetilde
{\boldsymbol{\Omega}}\boldsymbol{,}\widetilde{N},V_{0},\mathbf{Z}_{0})  &
=S-\frac{E}{T_{0}}-V\frac{P_{0}}{T_{0}}=-\frac{G}{T_{0}}-\mathbf{X\cdot}%
\frac{\boldsymbol{\mu}_{0}}{T_{0}}\\
&  =-\frac{G^{\text{i}}}{T_{0}}-\frac{\mathbf{P}^{2}}{2MT_{0}}-\frac
{\mathbf{M}\cdot\mathbf{\Omega}}{T_{0}}-\frac{\mathbf{X\cdot}\boldsymbol{\mu
}_{0}}{T_{0}},
\end{align*}
where
\[
G\equiv E-T_{0}S+P_{0}V,\ \ G^{\text{i}}\equiv E^{\text{i}}-T_{0}S+P_{0}V.
\]
Observe that the expansion of $\widetilde{S}(\widetilde{E}^{\text{i}%
},\widetilde{\boldsymbol{\Omega}}\boldsymbol{,}\widetilde{N},\widetilde
{V},\widetilde{\mathbf{Z}})$ does not require any knowledge of what is
happening with the system $\Sigma$. For example, we have \emph{not} assumed
internal equilibrium in the system or any of its subsystems. Since
$\widetilde{S}(E_{0},\widetilde{N},V_{0},\mathbf{Z}_{0})$ is a constant as
established earlier, the second law tells us that
\begin{equation}
\frac{dS_{0}}{dt}=-\frac{1}{T_{0}}\frac{dG^{\mathbf{X}}}{dt}\geq0,
\label{Gibbs_Free_Energy_Variation_0}%
\end{equation}
the desired result for the system. The inhomogeneity of the system has no
relevance in the above conclusion.

Let us now assume that $\mathbf{X}=0$ and introduce%
\begin{equation}
g_{k}\equiv e_{k}-T_{0}s_{k}+P_{0}v_{k},\ \ g_{k}^{\text{i}}\equiv
e_{k}^{\text{i}}-T_{0}s_{k}+P_{0}v_{k} \label{Inhomogeneous_Gibbs_Free_Energy}%
\end{equation}
for the subsystem $\sigma_{k}$. It is clear that%
\[
G=%
{\displaystyle\sum\limits_{k}}
g_{k},G^{\text{i}}=%
{\displaystyle\sum\limits_{k}}
g_{k}^{\text{i}}.
\]
Because of the quasi-independence of various subsystems, Eq.
(\ref{Gibbs_Free_Energy_Variation_0}) immediately leads to
\begin{equation}
\frac{dg_{k}}{dt}\leq0. \label{subsystem_Gibbs_Free_Energy_Variation}%
\end{equation}
Thus, $g_{k}(t)$ can be identified as the Gibbs free energy of the subsystem
$\sigma_{k}$ in the lab frame $\mathcal{L}$. Comparing this definition with
the definition in Eq, (\ref{Local_Gibbs_Free_Energy}) used in the conventional
non-equilibrium thermodynamics, we see that the discrepancy in the two Gibbs
free energy has not disappeared by taking into account the inhomogeneity
inherent in the system.

It may be argued that the above identification of $g_{k}$ in Eq.
(\ref{Inhomogeneous_Gibbs_Free_Energy}) is based on considering the entire
system in contact with the medium. One can alternatively consider a particular
subsystem $\sigma$ of the system in contact with the medium and the remaining
subsystems. However, a little bit of reflection shows that this will not
affect the behavior of $g_{k}$\ for the simple reason that the remaining
subsystems still form a very small part of the isolated system so that they
\emph{cannot} affect the internal equilibrium of the medium. To see this more
clearly, let us introduce a new medium $\widetilde{\Sigma}^{\prime}$
consisting of $\widetilde{\Sigma}$ and the remaining subsystems. Only for the
sake of simplicity, we do not consider any additional state variables. The
argument can be easily extended to include them. Let $\widetilde{S}^{\prime}$
denote the entropy and $\widetilde{E}^{\text{i}\prime}\simeq\widetilde
{E}^{\prime},$ $\widetilde{V}^{\prime}$\ and $\widetilde{N}^{\prime}$\ the
internal energy, volume and the number of particles of $\widetilde{\Sigma
}^{\prime}$, the latter of which is kept fixed. Expanding this entropy in
terms of the small quantities $e^{\text{i}}$ and $v$ of the chosen subsystem
requires calculating the derivatives%
\[
\left.  \frac{\partial\widetilde{S}^{\prime}}{\partial\widetilde{E}^{\prime}%
}\right\vert _{E_{0},V_{0}}\text{ and }\left.  \frac{\partial\widetilde
{S}^{\prime}}{\partial\widetilde{V}^{\prime}}\right\vert _{E_{0},V_{0}}.
\]
Because of the small size of the system relative to the original medium
$\widetilde{\Sigma}$, these derivatives are not different from%
\[
\left.  \frac{\partial\widetilde{S}}{\partial\widetilde{E}}\right\vert
_{E_{0},V_{0}}\text{ and }\left.  \frac{\partial\widetilde{S}}{\partial
\widetilde{V}}\right\vert _{E_{0},V_{0}},
\]
respectively. Thus, using $e=E_{0}-\widetilde{E}$ for the energy of the
subsystem $\sigma_{k}$, we find that
\[
\widetilde{S}^{\prime}(\widetilde{E}^{\text{i}\prime},\widetilde
{\boldsymbol{\Omega}}\boldsymbol{,}\widetilde{N}^{\prime},\widetilde
{V}^{\prime})=\widetilde{S}(E_{0},\widetilde{\boldsymbol{\Omega}%
}\boldsymbol{,}\widetilde{N},V_{0})-e\left.  \frac{\partial\widetilde
{S}^{\prime}}{\partial\widetilde{E}^{\prime}}\right\vert _{E_{0},V_{0}%
}-v\left.  \frac{\partial\widetilde{S}^{\prime}}{\partial\widetilde{V}%
^{\prime}}\right\vert _{E_{0},V_{0}}=\widetilde{S}(E_{0},\widetilde{N}%
,V_{0})-\frac{e}{T_{0}}-v\frac{P_{0}}{T_{0}}.
\]
Using this in $S_{0}(E_{0}^{\text{i}},N_{0},V_{0},t)=s(e^{\text{i}%
},\boldsymbol{\Omega,}n,v,t)+\widetilde{S}^{\prime}(\widetilde{E}%
^{\text{i}\prime},\widetilde{\boldsymbol{\Omega}}\boldsymbol{,}\widetilde
{N}^{\prime},\widetilde{V}^{\prime})$, where $\boldsymbol{\Omega}$\ is the
angular velocity of the $\mathcal{C}$ frame of the selected subsystem, allows
us to identify
\begin{equation}
g\equiv e-T_{0}s+P_{0}v \label{Subsystem_Gibbs_free_energy}%
\end{equation}
as the Gibbs free energy of the particular subsystem $\sigma$. It now follows
from Eq. (\ref{Second_Law_0}) that%
\[
\frac{dg}{dt}\leq0,
\]
a property we expect from the Gibbs free energy of a system. Incidentally,
this also provides an independent justification of the inequality in Eq.
(\ref{subsystem_Gibbs_Free_Energy_Variation}).

In terms of $g_{k}^{\text{i}}$, we immediately have%
\[
G^{\text{i}}=%
{\displaystyle\sum\limits_{k}}
(g_{k}^{\text{i}}+\frac{\mathbf{p}_{k}^{2}}{2m_{k}}+\mathbf{m}_{k}\cdot
\Omega_{k}),
\]
which is expected in view of the sum rule in Eqs. (\ref{Conservation_E}) and
(\ref{General_Ei_E_Relation}).

\subsection{System under Internal
Equilibrium\label{marker_inhomogeneous_internal}}

The above derivation only uses the second law, and the assumption that the
medium satisfies the condition of internal equilibrium. The situation with the
first law is very different. In general, the differential $ds_{k}$ of the
entropy of the subsystem $\sigma_{k}$ is given by
\[
ds_{k}=\frac{\partial s_{k}}{\partial e_{k}}de_{k}^{\text{i}}+\frac{\partial
s_{k}}{\partial v_{k}}dv_{k}+\frac{\partial s_{k}}{\partial n_{k}}dn_{k}%
+\frac{\partial s_{k}}{\partial\boldsymbol{\Omega}_{k}}\cdot
d\boldsymbol{\Omega}_{k}+\frac{\partial s_{k}}{\partial\mathbf{z}_{k}}\cdot
d\mathbf{z}_{k}(t)+\frac{\partial s_{k}}{\partial t}dt;
\]
cf. Eqs. (\ref{Equal_Entropy}) and (\ref{Functional_dependence}). If we now
assume internal equilibrium so that $s_{k}$ does not have an explicit
$t$-dependence, we can get rid of the last term above. In this case only, the
remaining derivatives identify the field variables $1/T_{k},P_{k}/T_{k},$
$-\mu_{k}/T_{k}$ etc. of the subsystem $\sigma_{k}:$%
\begin{equation}
\frac{\partial s_{k}}{\partial e_{k}}=\frac{1}{T_{k}(t)},\frac{\partial s_{k}%
}{\partial v_{k}}=\frac{P_{k}(t)}{T_{k}(t)},\frac{\partial s_{k}}{\partial
n_{k}}=-\frac{\mu_{k}(t)}{T_{k}(t)}dn_{k},\frac{\partial s_{k}}{\partial
\boldsymbol{\Omega}_{k}}=\frac{\mathbf{m}_{k}(t)}{T_{k}(t)},\frac{\partial
s_{k}}{\partial\mathbf{x}_{k}}=-\frac{\boldsymbol{\mu}^{(k)}(t)}{T_{k}%
(t)},\frac{\partial s_{k}}{\partial\mathbf{i}_{k}}=\frac{\mathbf{A}^{(k)}%
(t)}{T_{k}(t)}. \label{subsystem_parameters0}%
\end{equation}
These derivatives then allows us to obtain the Gibbs fundamental relation%
\begin{equation}
de_{k}^{\text{i}}(t)=T_{k}(t)ds_{k}-P_{k}(t)dv_{k}+\mu_{k}(t)dn_{k}%
-\mathbf{m}_{k}(t)\cdot d\boldsymbol{\Omega}_{k}+\boldsymbol{\mu}%
^{(k)}(t)\cdot d\mathbf{x}_{k}(t)-\mathbf{A}^{(k)}(t)\cdot d\mathbf{i}_{k}(t).
\label{subsystem_First_Law_internal}%
\end{equation}
Using%
\[
de_{k}(t)=de_{k}^{\text{i}}(t)+\mathbf{v}_{k}(t)\cdot d\mathbf{p}%
_{k}(t)+d[\mathbf{\Omega}_{k}(t)\cdot\mathbf{m}_{k}(t)],
\]
we find the first law of thermodynamics can be expressed in terms of the
energy as%
\begin{align}
de_{k}  &  =T_{k}(t)ds_{k}(t)+\mathbf{v}_{k}(t)\cdot d\mathbf{p}%
_{k}(t)+\mathbf{\Omega}_{k}(t)\cdot d\mathbf{m}_{k}(t)-P_{k}(t)dv_{k}%
(t)\label{First_Law_General}\\
&  +\mu_{k}(t)dn_{k}(t)+\boldsymbol{\mu}^{(k)}(t)\cdot d\mathbf{x}%
_{k}(t)-\mathbf{A}^{(k)}(t)\cdot d\mathbf{i}_{k}(t);\nonumber
\end{align}
compare with Eq. (\ref{EL_differential}). This allows us to think of the
entropy $s_{k}((e_{k}^{\text{i}},\boldsymbol{\Omega}_{k},n_{k},v_{k}%
,\mathbf{z}_{k},t)$ as a function%
\[
s_{k}(e_{k},\mathbf{p}_{k},\mathbf{m}_{k},n_{k},v_{k},\mathbf{z}_{k},t)
\]
so that the drift and the angular velocities in internal equilibrium are given
by%
\begin{equation}
\frac{\mathbf{v}_{k}(t)}{T_{k}(t)}=-\frac{\partial s_{k}(t)}{\partial
\mathbf{p}_{k}(t)},\ \ \frac{\mathbf{\Omega}_{k}(t)}{T_{k}(t)}=-\frac{\partial
s_{k}(t)}{\partial\mathbf{m}_{k}(t)}. \label{subsystem_parameters1}%
\end{equation}
However, different subsystems will undergo relative motions with respect to
each other as $\mathbf{v}_{k}(t)$ and $\mathbf{\Omega}_{k}(t)$ are different
for them. In addition, their temperatures $T_{k}(t)$ and pressures $P_{k}(t)$
are also different for each other, and so are $\boldsymbol{\mu}^{(k)}(t)\ $and
$\mathbf{A}^{(k)}(t)$. Thus, there would be viscous dissipation and,
consequently, entropy production as they come to equilibrium with each other
and with the medium. We now turn to this issue.

\section{Reversible and Irreversible Contributions}

\subsection{General Considerations: Bodies in a Medium}

It is a well-known fact \cite[Sect. 10]{Landau} that in equilibrium, the
system $\Sigma$ has a uniform translational motion as a whole with a constant
velocity $\mathbf{V}_{0}$ and a uniform rotation of the whole system with a
constant angular velocity $\boldsymbol{\Omega}_{0}$. Thus, there cannot be
relative motions between different subsystems in equilibrium. Thus, no
internal macroscopic motion is possible in equilibrium. In equilibrium, the
coefficients of the differential quantities in Eq. (\ref{First_Law_General})
take their constant values of the medium. Writing
\begin{equation}
ds_{k}(t)\equiv d_{\text{e}}s_{k}(t)+d_{\text{i}}s_{k}(t)
\label{entropy_generation_0}%
\end{equation}
as a sum of the change in the entropy $d_{\text{e}}s_{k}(t)$ due to reversible
exchange with the medium and the production of the entropy $d_{\text{i}}%
s_{k}(t)\geq0$ due to irreversible processes within the system, we have
\[
d_{\text{e}}s_{k}=\frac{1}{T_{0}}[de_{k}(t)-\mathbf{V}_{0}(t)\cdot
d\mathbf{p}_{k}(t)-\mathbf{\Omega}_{0}\cdot d\mathbf{m}_{k}(t)+P_{0}%
dv_{k}(t)-\mu_{0}dn_{k}(t)-\boldsymbol{\mu}_{0}\cdot d\mathbf{x}_{k}(t)]
\]
and
\begin{align}
d_{\text{i}}s_{k}  &  =F\left[  \frac{1}{T_{k}(t)}\right]  ds_{k}(t)+F\left[
-\frac{\mathbf{v}_{k}(t)}{T_{k}(t)}\right]  \cdot d\mathbf{p}_{k}(t)+F\left[
-\frac{\mathbf{\Omega}_{k}(t)}{T_{k}(t)}\right]  \cdot d\mathbf{m}%
_{k}(t)+F\left[  \frac{P_{k}(t)}{T_{k}(t)}\right]  dv_{k}(t)\nonumber\\
&  +F\left[  -\frac{\mu_{k}(t)}{T_{k}(t)}\right]  dn_{k}(t)+F\left[
-\frac{\boldsymbol{\mu}^{(k)}}{T_{k}(t)}\right]  \cdot d\mathbf{x}%
_{k}(t)+F\left[  \frac{\mathbf{A}^{(k)}}{T_{k}(t)}\right]  \cdot
d\mathbf{i}_{k}(t)\label{entropy_generation}\\
&  \geq0, \label{entropy_generation_limit}%
\end{align}
where
\begin{equation}
F_{z}\equiv F\left[  w\right]  =w(t)-w_{\text{eq}} \label{Irreversible_Change}%
\end{equation}
is the difference of the conjugate field $w(t)$ at time $t$ and its value in
equilibrium, i.e. as $t\rightarrow\infty$, and represents the thermodynamic
force associated with the conjugate extensive quantity $z(t)$; see also Eq.
(\ref{Force_y}). For the internal variables, the equilibrium value of
$\mathbf{A}_{0}$ is zero. According to the second law of thermodynamics, not
only $d_{\text{i}}s_{k}\geq0,$ but each term in Eq. (\ref{entropy_generation})
is non-negative:
\[
F\left[  w\right]  dz\geq0,
\]
where $z$ and $w$ form a conjugate pair. In terms of these pairs, we can
express the two parts of the entropy as follows:%
\begin{equation}
ds_{k}=\mathbf{w}_{k}\cdot d\mathbf{z}_{k}\mathbf{,\ \ }d_{\text{e}}%
s_{k}=\mathbf{w}_{0k}\cdot d\mathbf{z}_{k}\mathbf{,\ \ }d_{\text{i}}%
s_{k}=\mathbf{F}\left[  \mathbf{w}_{k}\right]  \cdot d\mathbf{z}_{k}\mathbf{,}
\label{Generalized_Entropy_Change}%
\end{equation}
which is the general form of the entropy differenmtial and its two components.

Let us now turn back to the current case under investigation. We can express
the generalized Gibbs fundamental relation as%
\[
de_{k}\equiv d_{\text{e}}e_{k}(t)+d_{\text{i}}e_{k}(t),
\]
where%
\[
d_{\text{e}}e_{k}=T_{0}ds_{k}(t)+\mathbf{V}_{0}(t)\cdot d\mathbf{p}%
_{k}(t)+\mathbf{\Omega}_{0}\cdot d\mathbf{m}_{k}(t)-P_{0}dv_{k}(t)+\mu
_{0}dn_{k}(t)+\boldsymbol{\mu}_{0}\cdot d\mathbf{x}_{k}(t),
\]
and
\begin{align*}
d_{\text{i}}e_{k}  &  =F\left[  T_{k}(t)\right]  ds_{k}(t)+F\left[
\mathbf{v}_{k}(t)\right]  \cdot d\mathbf{p}_{k}(t)+F\left[  \mathbf{\Omega
}_{k}(t)\right]  \cdot d\mathbf{m}_{k}(t)-F\left[  P_{k}(t)\right]
dv_{k}(t)\\
&  +F\left[  \mu_{k}(t)\right]  dn_{k}(t)+F\left[  \boldsymbol{\mu}%
^{(k)}\right]  \cdot d\mathbf{x}_{k}(t)-F\left[  \mathbf{A}^{(k)}\right]
\cdot d\mathbf{i}_{k}(t).
\end{align*}
The generalized form for the Gibbs fundamental relation is evidently
\[
de_{k}=T_{k}ds_{k}-\mathbf{W}_{k}\cdot d\mathbf{z}^{\prime}\mathbf{,\ \ }%
d_{\text{e}}e_{k}=T_{0}ds_{k}-\mathbf{W}_{0k}\cdot d\mathbf{z}^{\prime
}\mathbf{,\ \ }d_{\text{i}}e_{k}=F\left[  T_{k}\right]  ds_{k}-\mathbf{F}%
\left[  \mathbf{W}_{k}\right]  \cdot d\mathbf{z}^{\prime}\mathbf{,}%
\]
where $\mathbf{z}^{\prime}$ represents all state variables except the energy
$e_{k}$, and where $\mathbf{W}_{k}=T\mathbf{w}$ and $\mathbf{W}_{0}%
\mathbf{=T}_{0}\mathbf{w}_{0}$.

It is easy to see that performing the Legendre transform to obtain the Gibbs
free energy in Eq. (\ref{Inhomogeneous_Gibbs_Free_Energy}) only affects the
form of $d_{\text{e}}s_{k},$ but leaves $d_{\text{i}}s_{k}$ unaffected. Thus,
it is easy to see that
\[
dg_{k}(t)\equiv d_{\text{e}}g_{k}(t)+d_{\text{i}}g_{k}(t),
\]
where
\[
d_{\text{e}}g_{k}=-s_{k}(t)dT_{0}+\mathbf{V}_{0}(t)\cdot d\mathbf{p}%
_{k}(t)+\mathbf{\Omega}_{0}\cdot d\mathbf{m}_{k}(t)+v_{k}(t)dP_{0}+\mu
_{0}dn_{k}(t)+\boldsymbol{\mu}_{0}\cdot d\mathbf{x}_{k}(t),
\]
and
\[
d_{\text{i}}g_{k}\equiv d_{\text{i}}e_{k}.
\]
The general form of $d_{\text{e}}g_{k}$ is
\[
d_{\text{e}}g_{k}=-s_{k}(t)dT_{0}+v_{k}(t)dP_{0}-\mathbf{W}_{k}\cdot
d\mathbf{z}^{^{\prime\prime}},
\]
where $\mathbf{z}^{^{\prime\prime}}$ represents all state variables except the
energy $e_{k}$ and the volume $v_{k}$.

\subsection{Bodies forming a Finite Isolated System without a
Medium\label{marker_comparable_bodies}}

So far, we have considered a system or a collection of subsystems in a very
large medium $\widetilde{\Sigma}$ so that its field variables are held fixed
at $\mathbf{Y}_{0}$. We now consider a collection of bodies of comparable
sizes forming an isolated system $\Sigma_{0}$. In this case, we cannot treat
any of the bodies as a (macroscopically extensively large) medium with a fixed
field at $\mathbf{Y}_{0}$. As the collection strives towards equilibrium,
their field variables continue to change in time. This causes a problem in
identifying reversible contributions to various quantities. To solve this
problem, we discuss below a simple case; the generalization to more complex
situation will be obvious.

\subsubsection{Two Bodies}

Let us consider the simplest possible case of two comparable bodies $1$ and
$2$ in internal equilibrium at temperature $T_{1}$ and $T_{2}>T_{1}$,
respectively, that are brought in thermal contact at time $t=0.$ We will
simplify the discussion by assuming that all other extensive observables
besides the energy are held fixed. The case of two bodies in the shape of
rectangular cuboid along the $x$-axis are shown schematically in Fig.
\ref{Fig_two_boxes_isolated}(a), with the rectangular interface of area $A$
lying in the $yz$-plane.
\begin{figure}
[ptb]
\begin{center}
\includegraphics[
trim=0.000000in 0.000000in 0.000000in -0.199447in,
height=4.3716in,
width=4.8862in
]%
{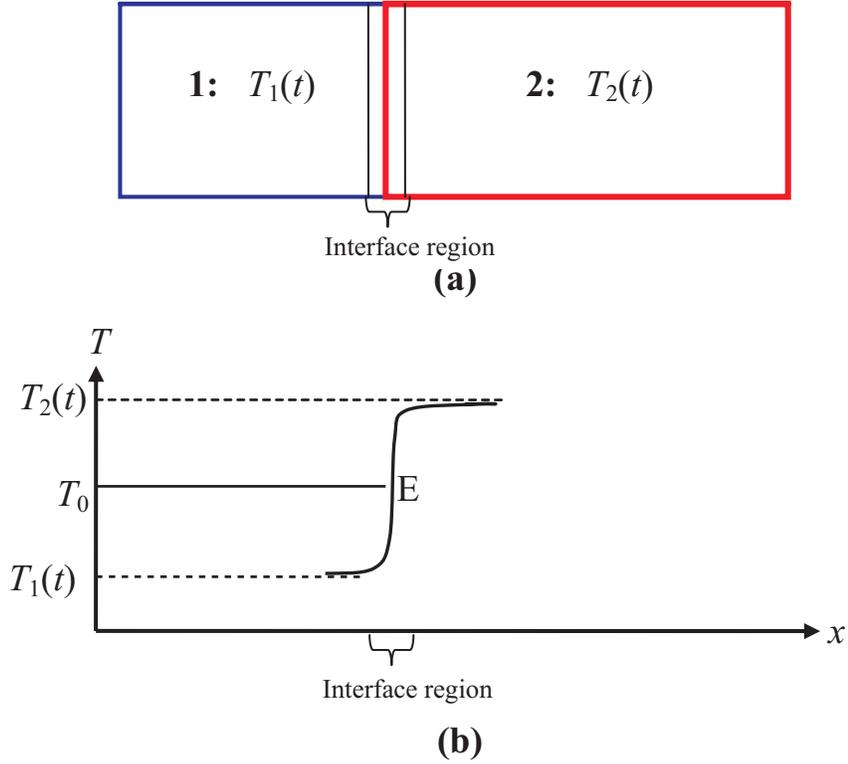}%
\caption{We show in Fig. (a) a simplified situation in which two bodies $1$
and $2$ in thermal contact are alligned along the $x$-axis. Initially, the two
bodies are in equilibrium at temperatures $T_{1}$ and $T_{2}>T_{1}$. Their
contact interface is taken to be a plane orthogonal to the $x$-axis as shown
in the figure, but the discussion is valid for any shape of the interface. In
reality, however, any interface is an interface region of some very small
width $\Delta x$\ over which the temperature continuously changes between
$T_{1}(t)$ and $T_{2}(t)$ so that there is a very narrow region of width
$dx<<\Delta x$ around a point E on the $x$-axis, where the temperature is
exactly $T_{0}$, the equilibrium temperature of the two bodies, as shown in
Fig. (b) for the simple case shown in Fig. (a). }%
\label{Fig_two_boxes_isolated}%
\end{center}
\end{figure}
Let $dQ=dE$ be the infinitesimal heat given to the body $1$ by the body $2$.
As the amount of heat is infinitesimal, it does not alter the temperatures of
the bodies in any significant way so that the entropy of the isolated system
$\Sigma_{0}$ changes by%
\[
dS_{0}=d_{\text{i}}S_{0}=dQ\left(  \frac{1}{T_{1}}-\frac{1}{T_{2}}\right)  >0
\]
at $t=0$.; compare with Eq. (\ref{Isolated_S_Generation}). This expression
will not be correct as time goes on as the temperature continues to change.
The general expression for the irreversible entropy generation is
\begin{equation}
dS_{0}(t)=d_{\text{i}}S_{0}(t)=dQ(t)\left(  \frac{1}{T_{1}(t)}-\frac{1}%
{T_{2}(t)}\right)  >0, \label{Irreversible_Entropy_Change_Total}%
\end{equation}
where $dQ(t)$ is the infinitesimal heat given to the body $1$ by the body $2$
between $t$ and $t+dt$, and $T_{1}(t),T_{2}(t)$ are the instantaneous internal
temperatures of the two bodies. It is clear from the discussion following Eq.
(\ref{Isolated_S_Generation}) that the irreversible entropy generation in Eq.
(\ref{Irreversible_Entropy_Change_Total}) is the sum of entropy changes in the
two bodies. Thus, Eq. (\ref{Irreversible_Entropy_Change_Total}) is exact for
the isolated system consisting of the two bodies. However, the point to
remember is that it is the irreversible entropy generation in the two bodies,
and says nothing about the irreversible entropy generation within each body.

If the system above were not isolated, then Eq.
(\ref{Irreversible_Entropy_Change_Total}) would still give the entropy change
due to the direct flow of heat between the two bodies, but it not represent
the irreversible entropy generation due to this heat flow. Moreover, there
would also be changes in the entropy of the system and each of the bodies due
to heat exchange with the medium or other bodies. These entropy changes will
also have their own irreversible entropy generations.\ The presence of the
medium at the \emph{equilibrium} temperature $T_{0}$ of the isolated system is
discussed below. The situation when $T_{0}$ is not the equilibrium temperature
of the isolated system is very different, as discussed later; see the
discussion after Eq. (\ref{Entropy_Generation_Each_Body_0}).

What can we say about the irreversible entropy generation within each body?
This question is, to the best of our knowledge, is not answered within the
local thermodynamics approach. To answer this question in our approach, we
proceed as follows. We know that the entropy generation in each body must
vanish when the bodies come to equilibrium. To obtain the desired result, we
introduce the common temperature $T_{0}$, when the two bodies come to
equilibrium. We now discuss two alternative approaches to determine the
individual entropy generation on the way to prove Theorem
\ref{Theorem_Isolated_Bodies}.

\paragraph{Introduction of a Medium at constant $T_{0}$}

Let us imagine that we insert the two bodies in an extensively large medium
$\widetilde{\Sigma}$ kept at a fixed temperature $T_{0}$, with all three
bodies now forming a new isolated system $\Sigma_{0}^{\prime}$. The situation
is schematically shown in Fig. \ref{Fig_two_boxex_in_medium}. There is no
direct contact between the two bodies, as shown. Despite this, the heat loss
$dQ(t)\geq0$ by the body $2$ to $\widetilde{\Sigma}$ is completely transferred
to body $1$. An alternative is to insert the medium between the two bodies
along the $x$-axis but not surrounding them from all sides with the same
effect of transferring the entire heat $dQ(t)$ from body $2$ to body $1$. The
width of the medium along the $x$-axis may be infinitesimally small of order
$dx$, but must have a macroscopically large cross-sectional area in the
$yz$-plane to ensure its constant temperature $T_{0}$ at all times. In either
case, the entropy of the medium does not change so that $d\widetilde
{S}=0,d_{\text{e}}\widetilde{S}=0$, and $d_{\text{i}}\widetilde{S}=0.$
Therefore, the irreversible entropy generation in the new isolated system
$\Sigma_{0}^{^{\prime}}$ is identical to the irreversible entropy generation
in $\Sigma_{0}$. This artificial introduction of $\widetilde{\Sigma}$ can now
be exploited to obtain the irreversible entropy generation in each smaller
body. The method of calculation above can be easily applied to this case, see
Eq. (\ref{entropy_generation}), to yield
\begin{equation}
d_{\text{i}}S_{1}(t)=dQ(t)\left(  \frac{1}{T_{1}(t)}-\frac{1}{T_{0}}\right)
,\ \ d_{\text{i}}S_{2}(t)=dQ(t)\left(  \frac{1}{T_{0}}-\frac{1}{T_{2}%
(t)}\right)  , \label{Entropy_Generation_Each_Body_0}%
\end{equation}
where $dQ(t)\geq0$\ is the heat addede to body $1$ or given out by body $2$.
By comparing with Fig. \ref{Fig_two_boxes_modified_system}, we note that
$dQ(t)=dQ_{1}(t)+dQ^{\prime}(t)=$ $dQ_{2}(t)+dQ^{\prime}(t)$, giving
\[
dQ_{1}(t)=dQ_{2}(t).
\]
It is evident that these entropy productions vanish in equilibrium, as
required by the notion of equilibrium between the two bodies. If we had
introduced a medium held at a constant temperature $T_{0}^{^{\prime}}\neq
T_{0}$, then the equilibrium will occur at $T_{0}^{^{\prime}}$, and not at
$T_{0}$, and we would end up with a different final state of $\Sigma
_{0}^{^{\prime}}$ than that of $\Sigma_{0}$. Thus, the medium must have the
constant temperature $T_{0}$.
\begin{figure}
[ptb]
\begin{center}
\includegraphics[
height=2.674in,
width=5.9802in
]%
{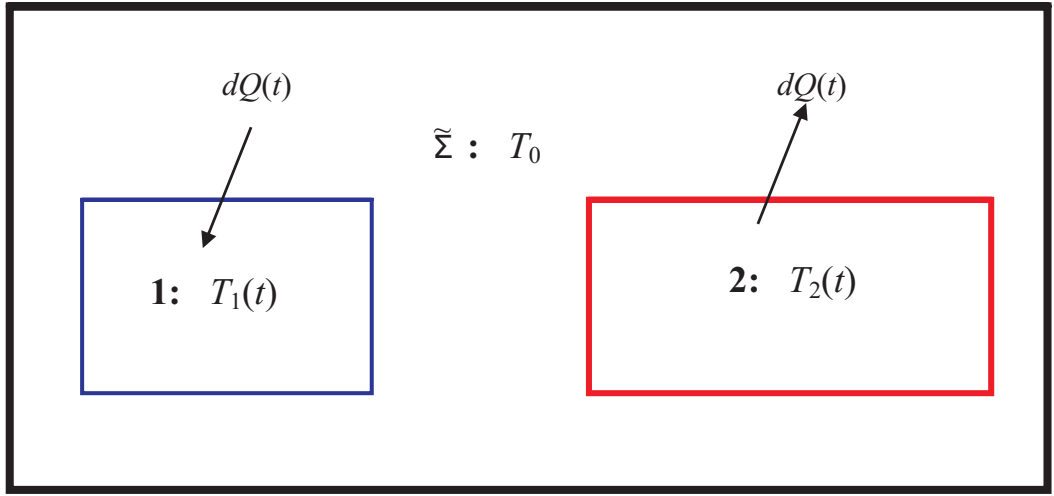}%
\caption{A schematic representation of a new isolated system $\Sigma
_{0}^{^{\prime}}$ consisting of the two bodies $1$ and $2$, not in thermal
conatct, surrounded by an extensively large medium $\widetilde{\Sigma}$, which
is so large that the presence of $1$ and $2$ does not alter its internal
temperature, which therefore remains constant.\ We choose the medium, as
exlained in the text, to be at the temperature $T_{0}$. }%
\label{Fig_two_boxex_in_medium}%
\end{center}
\end{figure}

The situation will be very different if the medium is taken to be at a
themperature $T_{0}^{\prime}\neq T_{0}$. In this case, the heat given to body
$1$ is different from the heat given out by body $2$. Comparing with Fig.
\ref{Fig_two_boxes_modified_system}, we note for the total amount of heat
$dQ_{2,\text{tot}}(t)$ given out by $2$ and the total amount of heat
$dQ_{1,\text{tot}}(t)$ given to $1$ that $dQ_{1,\text{tot}}(t)\equiv
dQ_{1}(t)+dQ^{\prime}(t)\neq$ $dQ_{2,\text{tot}}(t)\equiv dQ_{2}%
(t)+dQ^{\prime}(t)$, giving
\[
dQ_{1}(t)\neq dQ_{2}(t).
\]
The irreversible entropy generation in the two bodies are now givem by%
\[
d_{\text{i}}S_{1}(t)=dQ_{1,\text{tot}}(t)\left(  \frac{1}{T_{1}(t)}-\frac
{1}{T_{0}^{\prime}}\right)  ,\ \ d_{\text{i}}S_{2}(t)=dQ_{2,\text{tot}%
}(t)\left(  \frac{1}{T_{0}^{\prime}}-\frac{1}{T_{2}(t)}\right)  ,
\]
so that $d_{\text{i}}S(t)\equiv d_{\text{i}}S_{1}(t)+d_{\text{i}}S_{1}(t)$ is
given by
\begin{equation}
d_{\text{i}}S(t)=dQ^{\prime}(t)\left(  \frac{1}{T_{1}(t)}-\frac{1}{T_{2}%
(t)}\right)  +dQ_{1}(t)\left(  \frac{1}{T_{1}(t)}-\frac{1}{T_{0}^{\prime}%
}\right)  +dQ_{2}(t)\left(  \frac{1}{T_{0}^{\prime}}-\frac{1}{T_{2}%
(t)}\right)  . \label{entropy_generation_Diff_T0_Medium}%
\end{equation}
Notice that for $T_{0}^{\prime}>T_{2}(t)>T_{1}(t),$ $dQ_{1}(t)>0,$ and
$dQ_{2}(t)<0$. Similarly, for $T_{0}^{\prime}<T_{1}(t),$ $dQ_{1}(t)<0,$ and
$dQ_{2}(t)<0$. Thus, each of the last two irreversible entropy generation
contributions above is non-negative, as expected. We will see below that the
last two contributions are absent in the local theory.

The above approach can now be extended to many bodies at different initial
temperatures $T_{k}$. We assume that the medium $\widetilde{\Sigma}$ surrounds
all the bodies so that there is no direct contact between the bodies. If
$T_{0}$ still denotes the final temperature of all the bodies, treated as an
isolated system $\Sigma_{0}$, then the medium $\widetilde{\Sigma}$ must be
chosen to be also at the same constant temperature $T_{0}$. It is obvious
that
\begin{equation}
d_{\text{i}}S_{k}(t)=dQ_{k}(t)\left(  \frac{1}{T_{k}(t)}-\frac{1}%
{T_{0}^{\prime}}\right)  , \label{Entropy_Generation_Each_Body_1}%
\end{equation}
where $dQ_{k}(t)$ denotes the infinitesimal heat added to the $k$-th body
between $t$ and $t+dt,$ and $T_{k}(t)$ is its instantaneous internal
temperature; compare with Eq. (\ref{entropy_generation}). This is a
generalization of Eq. (\ref{Irreversible_Entropy_Change}) to many bodies, and
exploits the absence of causality inherent in $d_{\text{e}}S_{k}(t)$, as
exemplified by Eq. (\ref{Reversible_Entropy_Change}). The irreversible entropy
generation in the isolated system $\Sigma_{0}^{^{\prime}}$ consisting only of
the bodies (without the medium) is the sum of all these contributions:%
\[
d_{\text{i}}S_{0}(t)\equiv%
{\textstyle\sum\limits_{k}}
d_{\text{i}}S_{k}(t)=%
{\textstyle\sum\limits_{k}}
\frac{dQ_{k}(t)}{T_{k}(t)}=%
{\textstyle\sum\limits_{k}}
dS_{k}(t)=dS(t),
\]
as it must be for an isolated system; Obviusly, $d_{\text{i}}S_{0}(t)$ does
not depend on the final temperature $T_{0}$. Here, we have used the fact that
\[%
{\textstyle\sum\limits_{k}}
dQ_{k}(t)\equiv0
\]
because of the isolation of $\Sigma_{0}$ and $\Sigma_{0}^{^{\prime}}$.\ It is
quite clear that the discussion is easily extended to include other extensive
variables which results in the expression for $d_{\text{i}}S_{k}(t)$, which is
identical in form to the expression $d_{\text{i}}s_{k}(t)$in Eq.
(\ref{entropy_generation}).

One may feel uneasy that the introduction of the fictitious medium
$\widetilde{\Sigma}$ has changed the problem. To see that this is not the
case, we reemphasize that its introduction does not affect the heat
$dQ_{k}(t)$. Since, it is the heat $dQ_{k}(t)$ that determines the entropy
generation, and since the equilibrium state of each body is the correct
equilibrium state, the result must be correct. To offer an even stronger
argument, we now provide an alternative proof without the introduction of the
medium so that we can feel comfortable in associating a medium in the case
when finite-size bodies form an isolated system. The importance of this trick
is that it allows us to use all the results we have found by the use of a medium.

\paragraph{Without any Medium}

Let us revert to the simple case of two bodies in thermal contact. In reality,
the interface or the contact region between the two bodies is not a sharp
boundary with a discontinuity in temperature; rather, it is a narrow region of
width $\Delta x$\ over which the temperature rapidly changes from $T_{1}(t)$
to $T_{2}(t)$, as shown in Fig. \ref{Fig_two_boxes_isolated}(b). Thus, there
exist a point E over this region where the temperature is precisely $T_{0}$.
It is possible that the location of the point E along the $x$-axis changes
with time, but this point is irrelevant. The relevant point is that the
temperature in a very narrow width $dx<<\Delta x$ around and including this
point will remain constant in time. If we take the narrow neighborhood of a
point whose temperature is different from $T_{0}$, its temperature will
eventually change to $T_{0}$, as equilibrium is achieved. Thus, the
temperature of the narrow region around this point will not be constant in
time. We take the region where the temperature is greater than $T_{0}$ along
with the half-width $dx/2$ around E on the side of $2$ to be the part of the
body $2$, while the region with the temperature less than $T_{0}$ along with
the half-width $dx/2$ around E on the side of $1$\ to be the part of the body
$1$. From the quasi-independence of the two bodies, it is clear that the
inclusion of these regions will not significantly affect the internal
temperatures $T_{1}(t)$ to $T_{2}(t).$ The heat lost $dQ(t)$ by the body $2$
is transferred to the body $1$ at a \emph{constant} temperature $T_{0}$. Thus,
while
\[
dS_{1}(t)=\frac{dQ(t)}{T_{1}(t)},
\]
where $T_{1}(t)$ is the internal temperature of the body $1$,%
\[
d_{\text{e}}S_{1}(t)=\frac{dQ(t)}{T_{0}},
\]
so that
\[
d_{\text{i}}S_{1}(t)=dQ(t)\left(  \frac{1}{T_{1}(t)}-\frac{1}{T_{0}}\right)
,
\]
as discovered above; see Eq. (\ref{Entropy_Generation_Each_Body_0}). In a
similar fashion, we obtain
\[
d_{\text{i}}S_{2}(t)=dQ(t)\left(  \frac{1}{T_{0}}-\frac{1}{T_{2}(t)}\right)
\]
in accordance with Eq. (\ref{Entropy_Generation_Each_Body_0}).

The above discussion justifies the use of $T_{0}$ as a temperature in the
interfacial region. However, the important point is that the determination of
$d_{\text{e}}S(t)$ for any body requires the use of the equilibrium
temperature of the body in accordance with Eq.
(\ref{Reversible_Entropy_Change}). Thus, the discussion is equally valid for
any number of subsystems in the system.

\paragraph{Comaprison with Local Thermodynamics}

The same result as in Eq. (\ref{Irreversible_Entropy_Change_Total}) is also
obtained in the local thermodynamics, as can be easily seen; see for example
\cite[Eq. (3.14)]{Prigogine-old}. In the limit in which the interfacial region
between the two subsystems is infinitesimal in width along the $x$-direction,
we can obtain the continuum analog of the irreversible entropy generation
\emph{between} the two neighboring regions. Denoting $T_{1}(t)$ by $T(x,t)$
and $T_{2}(t)$ by $T(x+dx,t)\simeq T(x,t)+\left(  \partial T/\partial
x\right)  dx,$ we have
\[
dS_{0}(t)=d_{\text{i}}S_{0}(t)\simeq dQ(t)\left[  \partial(1/T)/\partial
x\right]  dx
\]
\ for conduction. Dividing and multiplying by the cross-sectional area $A$ of
the interface, and using $Adx$ as the volume of the interfacial region, we
have for the rate $\sigma$ of entropy production per unit interfacial volume%
\begin{equation}
\sigma(x,t)=\overset{.}{q}(t)\left[  \partial(1/T)/\partial x\right]  ,
\label{Rate_of_Entropy_Production}%
\end{equation}
where $\overset{.}{q}(t)=$ $\overset{.}{Q}(t)/A$ denotes the heat flux. This
expression (in three dimensions, the result will contain $\boldsymbol{\partial
}(1/T)$) in this limit is a standard result for the entropy production due to
heat conduction in local thermodynamics.

It should be stressed that our derivation of $\sigma(x,t)$ above is
independent of the how long the two subsystems are along the $x$-axis, but
assumes implicitly that $T(x,t)$ has a Taylor series expansion. It is also
obvious that $\sigma(x,t)$ must be zero outside the interfacial region. Thus,
$\sigma(x,t)$ should be correctly identified as the rate of entropy per unit
volume in the \emph{interfacial region}. Therefore, as it follows from the
discussion immediately following Eq. (\ref{Irreversible_Entropy_Change_Total}%
), the derivation says nothing about how much of the irreversible entropy is
generated within each body. The issue is avoided in local thermodynamics by
assuming that the entire volume of the system is composed only of such
interfacial regions. This is contrary to the basic postulate of local
equilibrium according to which each local region has a well-defined
temperature $T(t)$, while the interfacial regions have non-zero gradients of
the temperature.

We also observe that the form of $\sigma(x,t)$ in Eq.
(\ref{Rate_of_Entropy_Production}) \ is only valid for the case when Eq.
(\ref{Irreversible_Entropy_Change_Total}) is valid. It is not valid for the
case covered in Eq. (\ref{entropy_generation_Diff_T0_Medium}). This is the
case when the equilibrium temperature $T_{0}^{\prime}$ of the system is
different from the equilibrium temperature $T_{0}$ of the two bodies under
consideration if treated as isolated bodies. Thus, the above expression
$\sigma(x,t)$ will not be valid if our system consists of a large number of
bodies so that the irreversible entropy generation between any two bodies in
contact will be given by Eq. (\ref{entropy_generation_Diff_T0_Medium}). In
this case, the expresasion explicitly contains the equilibrium temperature of
the system in the last two terms, which is not the case with $\sigma(x,t)$ in
Eq. (\ref{Rate_of_Entropy_Production}), thus verifying the earlier statment
made after Eq. (\ref{entropy_generation_Diff_T0_Medium}).

\subsubsection{Several Bodies}

Let us now consider a simple extension of the case shown in Fig.
\ref{Fig_two_boxes_isolated}(a) in which there are several bodies
$1,2,3,\cdots,$B in thermal contact along the $x$-axis forming an isolated
body $\Sigma_{0}$, with their temperatures in an increasing order:
$T_{1}<T_{2}<\cdots<T_{k}<T_{k+1}<\cdots<T_{\text{B}}$. Let $T_{k-1}%
<T_{0}<T_{k}$ again denote the equilibrium temperature $T_{0}$ of the isolated
body $\Sigma_{0}$, so that there exist a point E over their interface region
where the temperature is precisely $T_{0}$. Let us consider the infinitesimal
heat $dQ_{1}(t)$ received by $1$ to be the part of the heat $dQ_{1}(t)$ that
was given out by the body $k$ and was transferred unaltered via
$k-1,k-2,\cdots3,2$ to $1.$ Any heat transfer through the interface region at
constant $T_{0}$ is isothermal. Therefore, the reversible entropy change due
to $dQ_{1}(t)$ \ is precisely $d_{\text{e}}S_{1}=$ $-dQ_{1}(t)/T_{0}$. The
entropy change of $1$ is $dS_{1}=dQ_{1}(t)/T_{1}(t)$ so that it immediately
follows that%
\[
d_{\text{i}}S_{1}(t)=dQ_{1}(t)\left(  \frac{1}{T_{1}(t)}-\frac{1}{T_{0}%
}\right)  ,
\]
as above. Similarly, we find%
\[
d_{\text{i}}S_{\text{B}}(t)=-dQ_{\text{B}}(t)\left(  \frac{1}{T_{\text{B}}%
(t)}-\frac{1}{T_{0}}\right)
\]
for the body B, where $dQ_{\text{B}}(t)$ is the heat rejected by B, which was
transferred isothermally to the body $k$ unaltetred. For the body $2$, which
receives an infinitesimal heat $dQ_{2}(t)$, we similarly find%
\[
d_{\text{i}}S_{2}(t)=dQ_{2}(t)\left(  \frac{1}{T_{2}(t)}-\frac{1}{T_{0}%
}\right)  ,
\]
and so on. All these results can be easily obtained by inserting a medium held
at fixed temperature $T_{0}$ at the interface between each connsecutive pair
of bodies. Thus, we conclude that we can consider an isolated inhomogeneous
body as a body embedded in a medium without affecting any of the consequences.
This then proves Theorem \ref{Theorem_Isolated_Bodies}.

\subsection{General Discussion}

We are now in a position to provide a proof of Theorem
\ref{Theorem_Isolated_Bodies}. We consider the reversible entropy change
$d_{\text{e}}s^{p}$ in a given subsystem due to the $p$-th conjugate field
$W^{p}(t)$ ($Y(t)$ or $A(t))$ due to the change $dz^{p}$. It is given by
\[
d_{\text{e}}s^{p}=\frac{W_{0}}{T_{0}}dz^{p},
\]
while
\[
ds^{p}=\frac{W(t)}{T(t)}dz^{p},
\]
so that
\[
d_{\text{i}}s^{p}=F\left[  \frac{W(t)}{T(t)}\right]  dz^{p}\geq0
\]
in terms of the quantity $F\left[  w\right]  $ in Eq.
(\ref{Irreversible_Change}). This thus proves the theorem.

\section{Discussion and Conclusions}

We have proposed a scheme to extend our previous work in I to an inhomogeneous
system in which subsystems may undergo relative translational and rotational
motions with respect to each other to contribute to irreversibility. Another
source of irreversibility is the temperature variation, which was considered
in I. Each subsystem or the medium is identified by a set of extensive state
variables $\mathbf{Z}(t)$, and some constant parameters $\mathbf{C}$, some of
which may be external parameters and need not be extensive. Examples of
$\mathbf{C}$ may be the number of particles in the system that characterize
the system, or the translational and the angular velocities of the frame of
reference in which the observations are made. Examples of the state variables
are mechanical quantities such as energy, volume, etc. and the internal
variables are the translational and angular momenta of the various subsystems, etc.

\subsection{Quasi-independence and Additivity of Entropy and Energy}

As our approach starts from the second law, the law of increase of entropy,
the entropy as a state function plays the most important role in our approach.
Accordingly, we need to ensure that this quantity can always be determined by
or at least formally defined in terms of some basic quantities pertaining to
the macrostate of the system of interest. We use the Gibbs formulation of the
entropy such as that in Eq. (\ref{Gibbs entropy}), which is applicable in all
cases as discussed in a recent review \cite{Gujrati-Symmetry}. This
formulation is well-established for an isolated system \cite{Gujrati-Symmetry}%
, but we show in Sect. \ref{marker_entropy_additivity} that it is also
applicable to an open system even when it is\ not in equilibrium with its
surrounding medium. This formulation of entropy for open system is a well
known result in equilibrium thermodynamics \cite{Landau} and our demonstration
generalizes this result to non-equilibrium systems so that this entropy can be
used as the central quantity in developing a non-equilibrium thermodynamics
with the second law as the starting point. It is required that the open system
be quasi-independent of the medium. We have shown in Sect.
\ref{marker_entropy_additivity} that this quasi-independence is a very
important property, which is required for the \emph{additivity} of the entropy
and of energy. For example, quasi-independence of the system and the medium
ensures that their energy of interaction is negligible so that dropping it
makes their energies additive, as seen from Eqs. (\ref{Interaction_Energy})
and (\ref{Energy_Sum_Approx}). Indeed, without the additivity of the energy,
the entropy cannot be additive as discussed in Sect.
\ref{marker_entropy_additivity}. The basic idea is very simple but profound.
The entropy being a state function must only depend on state variables. By
definition, all state variables of a body must be solely determined by what
happens within the body. If there are other bodies, their influence on the
body must not destroy the additivity of all extensive state variables. Let us
consider the isolated system consisting of the system $\Sigma$ and the medium
$\widetilde{\Sigma}$ it is in contact with, which is considered in Sect.
\ref{marker_entropy_additivity}; in particular, consider the energies in Eq.
(\ref{Interaction_Energy}). We will assume no relative motion between $\Sigma$
and the medium $\widetilde{\Sigma}$ to make the discussion simple. The
energies $E_{0},E(t)$ and $\widetilde{E}(t)\ $then represent the energies of
the isolated system, the system and the medium, respectively. However, the
presence of the interaction energy $E_{0}^{(\text{int})}(t)$ in this equation
destroys the required property that the right side of Eq.
(\ref{Interaction_Energy}) is the sum of internal energies of the system and
the medium; the interaction energy is a property of both bodies. Accordingly,
the entropy $S_{0}(E_{0})$ of $\Sigma_{0}$ cannot depend only on the energies
of the two bodies separately. The difference represented by $S_{0}%
^{(\text{int})}(t)$ in Eq. (\ref{Entropy_Loss_Sum}) depends on both bodies.
Thus, the entropy of the isolated system cannot be a sum of entropies, each
representing the entropy of one of the bodies alone. In other words,
generalizing this result to a collection of bodies b$_{j}$ forming a body B,
we conclude that the entropy $S($B$)$ of B \emph{cannot} be expressed as a sum
of entropies, each term $S($b$_{j})$ representing the entropy of the body
b$_{j}$%
\[
S(\text{B})\neq\sum_{j}S(\text{b}_{j}).
\]
However, the additivity requires not only that the sum of the energies of
various parts b$_{j}$ of a body must yield the energy of the body B itself%
\[
E(\text{B})=\sum_{j}E(\text{b}_{j}),
\]
which is satisfied if the linear sizes of the bodies are large compared with
the range of interaction, but also requires that various bodies are
quasi-independent, which is satisfied if the linear sizes of the bodies are
large compared with thecorrelation length. Unless thess properties hold, the
entropy cannot be expressed as a sum of the entropies of its parts, with each
entropy depending only on the properties of the part alone. It is the latter
property that makes entropy a state function.

\subsection{Concept of Internal Equilibrium}

For a body (or its parts, such as subsystems) out of equilibrium, the entropy
usually has an \emph{explicit} dependence on time in addition to the implicit
time-dependence%
\[
S(\text{B})=S(\mathbf{X},\mathbf{I},t);
\]
the latter arises from the dependence of the entropy on state variables
$\mathbf{X}(t)$ and $\mathbf{I}(t)$ that explicitly depend on time; we will
suppress this time-dependence in the following for notational simplicity. (We
will assume, as discussed in \cite{Gujrati-review:Fluctuations}, that at least
one observable, which we take to be the number of particles $N$ is held fixed
and is not contained in $\mathbf{X}$.) Thus, such an entropy will continue to
change (increase if the body is isolated) for \emph{fixed} state variables.
For a homogeneous body, this will happen if the state variables in
$\mathbf{X}$ and $\mathbf{I}$ do not uniquely specify its state. In other
words, there may be other internal variables not contained in $\mathbf{I}$.
This is a connsequence of a simple generalization of Theorem
\ref{marker_absence_of_internal_equilibrium}. The other possibility is that
the body is not homogeneous. If the isolated body is not in equilibrium, its
entropy continues to increase according to the second law in Eq.
(\ref{Second_Law_0}). However, if it happens that the entropy of the body does
not explicitly depend on time%
\[
S_{\text{IE}}(\text{B})=S(\mathbf{X},\mathbf{I}),
\]
then its entropy must be at its \emph{maximum} for fixed $\mathbf{X}$ and
$\mathbf{I}$. This is how we had introduced the concept of internal
equilibrium in I: in this state, the entropy of a system, which is not in
equilibrium with the medium $\widetilde{\Sigma}(\mathbf{Y}_{0},0)$, has no
explicit time-dependence. It can be brought in contact with another medium
$\widetilde{\Sigma}^{\prime}$, representing the medium $\widetilde{\Sigma
}(\mathbf{Y}_{\text{IE}},\mathbf{A}_{\text{IE}})$, so that it remains in
equilibrium with the new medium. In other words, there is no difference
between a body in internal equilibrium with the medium $\widetilde{\Sigma
}(\mathbf{Y}_{0},0)$, and the body in equilibrium with the medium
$\widetilde{\Sigma}(\mathbf{Y}_{\text{IE}},\mathbf{A}_{\text{IE}})$; in the
latter case, the medium ensures to keep the averages $\mathbf{X}_{\text{IE}}$
and $\mathbf{I}_{\text{IE}}$ fixed. In both cases, the entropy is maximum for
the fixed values of their state variables $\mathbf{X}_{\text{IE}}$ and
$\mathbf{I}_{\text{IE}}$. Accordingly, all properties of a body in equilibrium
also hold for a body in internal equilibrium at each instant. For example,
Theorem \ref{marker_Uniform_Motion} established in Sect.
\ref{marker_internal_equilibrium_0} shows that there \emph{cannot} be any
relative motion within such a body. The only possible motion is a
\emph{uniform translation }and a\emph{ rigid-body rotation}. As a consequence,
as discussed in Sect. \ref{marker_internal_equilibrium_0}, there is no viscous
dissipation \emph{within} the body in internal equilibrium. Thus, while there
is no relative motion between the system and $\widetilde{\Sigma}%
(\mathbf{Y}_{\text{IE}},\mathbf{A}_{\text{IE}})$, there will in general be
relative motions between the system and the medium $\widetilde{\Sigma
}(\mathbf{Y}_{0},0)$. This result can be generalized to the following
statement: The only source of viscous dissipation in a collection of bodies is
then due to relative motions between different bodies, when each of which is
postulated to be in internal equilibrium.%

\begin{figure}
[ptb]
\begin{center}
\includegraphics[
trim=1.227302in 0.158563in 0.081295in 0.238845in,
height=2.6844in,
width=5.175in
]%
{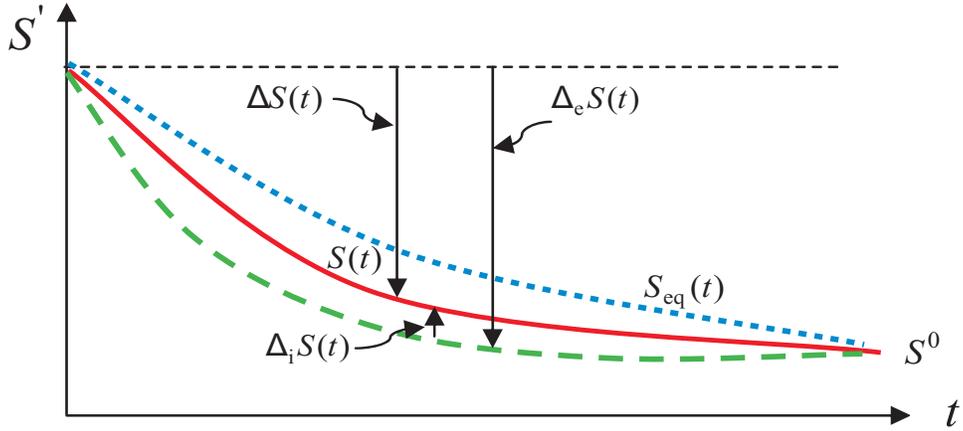}%
\caption{Schematic form of the entropy $S(t)$ (solid curve) and $S_{\text{eq}%
}(t)$ (dotted upper curve) in a cooling experiment. Both curves start from
$S^{\prime}$ at $t=0$ and end at $S^{0}$ as $t\rightarrow\infty$. The entropy
difference $\Delta S(t)$ is shown by the leftmost downward pointing arrow and
$\Delta_{\text{e}}S(t)$ by the rightmost pointing arrow. The irreversible
entropy generation $\Delta_{\text{i}}S(t)>0$ is shown by the small upward
pointing arrow. }%
\label{Fig_entropy}%
\end{center}
\end{figure}

It was assumed in I that the system $\Sigma$ at each instant remains in
internal equilibrium as it goes from one equilibrium macrostate $\mathcal{M}%
^{\prime}$\ in contact with a medium $\widetilde{\Sigma}(\mathbf{Y}%
_{0}^{^{\prime}},0)$ at time $t=0$ to another equilibrium macrostate
$\mathcal{M}^{0}$ in contact with a medium $\widetilde{\Sigma}(\mathbf{Y}%
_{0},0)$ as $t\rightarrow\infty$. Let us try to understand the behavior of its
entropy $S(t)$ during this process, which can then be extended to a subsystem
in internal equilibrium. The field variables for the two macrostates are
$\mathbf{W}^{\prime}\equiv(\mathbf{Y}^{^{\prime}}=\mathbf{Y}_{0}^{^{\prime}%
},\mathbf{A}^{^{\prime}}=0)$ and $\mathbf{W}^{0}\equiv(\mathbf{Y}%
=\mathbf{Y}_{0},\mathbf{A}=0)$, respectively.\ As discussed several times, the
intermediate macrostate $\mathcal{N}(t)$ at some instant $t$ can be thought of
as an equilibrium state of the system $\Sigma$ in contact with a medium
$\widetilde{\Sigma}(\mathbf{Y}_{\text{IE}}=\mathbf{Y}(t),\mathbf{A}%
_{\text{IE}}=\mathbf{A}(t))$ after being disconnected from the medium
$\widetilde{\Sigma}(\mathbf{Y}_{0},0)$. In this equilibrium state, the system
has a well-defined field $\mathbf{W}(t)=(\mathbf{Y}(t),\mathbf{A}(t))$ and the
entropy $S(t)$, which is shown by the solid curve in Fig. \ref{Fig_entropy}.
However, this equilibrium state of the system does not represent the
equilibrium state of the system in contact with either the medium
$\widetilde{\Sigma}(\mathbf{Y}_{0},0)$ or the medium $\widetilde{\Sigma
}(\mathbf{Y}_{\text{IE}},0)$. If we bring the system in contact with the
medium $\widetilde{\Sigma}(\mathbf{Y}_{\text{IE}},0),$ its entropy will
increase as the internal variables relax towards the new equilibrium, to be
denoted by $\mathcal{M}(t)$ with fixed $\mathbf{X}_{\text{IE}}$. The
equilibrium entropy of the system in $\mathcal{M}(t)$ will be denoted by
$S_{\text{eq}}(t)$, which is shown by the dotted curve in Fig.
\ref{Fig_entropy}.

Because of our assumption of internal equilibrium, $S(t)$ represents the
maximum possible entropy for the macrostate $\mathcal{N}(t)$ at $t$. The
change in the entropy $\Delta S(t)$ is given by
\[
\Delta S(t)=S(t)-S^{\prime},
\]
where $S^{\prime}$ is the equilibrium entropy in the macrostate $\mathcal{M}%
^{\prime}$ at $t=0$, and is shown by the left-most arrow in Fig.
\ref{Fig_entropy}. The total entropy change in the entire process is
\[
\Delta S_{\infty}=S^{0}-S^{\prime},
\]
where $S^{0}$ as the equilibrium entropy in the macrostate $\mathcal{M}^{0}$
as $t\rightarrow\infty$.

The entropy difference $\Delta S(t)$ can always be broken into two parts in
analogy with the two terms in eq. (\ref{entropy_generation_0}):%
\[
\Delta S(t)=\Delta_{\text{e}}S(t)+\Delta_{\text{i}}S(t),
\]
with $\Delta_{\text{e}}S(t),$ see the right-most arrow in Fig.
\ref{Fig_entropy}, representing the entropy exchange with the medium
$\widetilde{\Sigma}(\mathbf{Y}_{0},0)$ and $\Delta_{\text{i}}S^{0}\geq0$ the
irreversible entropy production within the system, see the middle small arrow
in Fig. \ref{Fig_entropy}. It does not matter whether the process is
reversible or irreversible, the entropy change is always $\Delta S(t)$. In the
former process, $\Delta_{\text{e}}S(t)=\Delta S(t)$, $\Delta_{\text{i}}S(t)=0$
and $S(t)=S_{\text{eq}}(t)$. In the latter process,
\[
\Delta_{\text{e}}S(t)<\Delta S(t)
\]
and there is irreversible entropy production during the passage from
$\mathcal{M}^{\prime}$\ to $\mathcal{N}(t)$,
\[
\Delta_{\text{i}}S(t)>0
\]
in accordance with Eq. (\ref{entropy_generation_limit}). We thus see that
irreversible generation of entropy has been accounted for in our approach; see
also I and Eq. (\ref{entropy_generation_limit}). We can easily extend the
above discussion to the subsystems, where we now also have the possibility of
irreversible entropy generation through relative motions between them, which
is also accounted for in our approach.

We have found (Theorem \ref{marker_absence_of_internal_equilibrium}) that
while a body is in internal equilibrium in the presence of internal variables,
it is not in internal equilibrium when internal variables are not considered
in its description. In other words, the entropy of a body, see the solid curve
for $S(t)$ in Fig. \ref{Fig_entropy},%
\begin{equation}
S(\text{B})=S_{\text{IE}}(\text{B})=S(\mathbf{X},\mathbf{I})
\label{Body_Internal_Equilibrium}%
\end{equation}
has no explicit time-dependence; however, it becomes an explicit function of
time if expressed only in terms of observables:%
\begin{equation}
S(\text{B})=S(\mathbf{X},t). \label{Body_Not_Internal_Equilibrium}%
\end{equation}
This theorem shows how important the internal variables are in describing
glasses. Because of the presence of $t$, this entropy will continue to
increase and approach $S_{\text{eq}}(t)$, shown in Fig. \ref{Fig_entropy}, if
we keep $\mathbf{X}$ fixed.

As experimentalists can only control the observables in $\mathbf{X}$, it is
the expression of the entropy in Eq. (\ref{Body_Not_Internal_Equilibrium})
that is experimentally relevant. Because of the explicit $t$-dependence, the
body will continue to relax even if the observables are held fixed by
isolating the body from its surroundings and bringing it in contact with the
medium $\widetilde{\Sigma}(\mathbf{Y}_{\text{IE}},0)$. This relaxation occurs
because the internal variable $\mathbf{I}$ is no longer held fixed. It will
continue to change with time until finally the entropy reaches its maximum
possible value for fixed $\mathbf{X}$. We can apply Theorem
\ref{marker_chemical_potential} for an isolated system to our body. The
conclusion is that the maximum of the entropy occurs when $\mathbf{A}$ for the
internal variable $\mathbf{I}$ is identically zero, which is why we had
selected this particular medium. The final value of $\mathbf{I}$ is some value
$\mathbf{I}_{\text{eq}}(\mathbf{X})$. The body in this equilibrium state is
the fully relaxed body with its entropy given by%
\[
S_{\text{eq}}(\text{B})=S_{\text{eq}}(\mathbf{X});
\]
this entropy is shown as $S_{\text{eq}}(t)$ in Fig. \ref{Fig_entropy}. The
affinity $\mathbf{A}$ of the internal variable is identically zero, and the
fields in this state are given by $\mathbf{Y\equiv Y}_{\text{IE}}$. This
equilibrium state of the isolated body should not be confused with the state
of the body in internal equilibrium in the presence of the internal variable
$\mathbf{I}$, as shown in Eq.(\ref{Body_Internal_Equilibrium}). This body has
$\mathbf{I}$ different from $\mathbf{I}_{\text{eq}}(6\mathbf{X})$, and the
corresponding affinity $\mathbf{A}$ different from zero.

\subsection{Thermodynamic Potentials}

The additivity of the entropy and energy allows us to treat our system as a
collection of quasi-independent subsystems so that we can develop the
thermodynamics of an inhomogeneous system. Under the mildest possible
assumption that the medium is in internal equilibrium so that its field
variables $\widetilde{w}_{k}$ are defined via Eq.
(\ref{Chemical_Potential_Internal_Variables}), but the system itself is (or
the subsystems are) not necessarily in internal equilibrium, we are able to
identify the thermodynamic potential for the system. The field variables of
the system in internal equilibrium vary with time and, therefore, they differ
from the medium's field variables. They become identical only when the system
is in equilibrium with the medium. It is obvious from Eq.
(\ref{Chemical_Potential_Internal_Variables}) that the product
\[
\widetilde{w}_{k}\widetilde{Z}_{k},\text{ no summation implied}%
\]
is dimensionless because the entropy is defined as dimensionless in this work.
Accordingly, the (dimensionless) entropy $\widetilde{S}(\widetilde{\mathbf{Z}%
})$ can be expressed in terms of these dimensionless products. It follows then
that resulting thermodynamic potentials will contain these products, as shown
in Sect. \ref{marker_inhomogeneous_arbitrary}. The actual form of the
thermodynamic potential function depends on the choice of the set $\mathbf{C}%
$, and requires field variables to maintain the dimensions of each term
appearing in it. The fields variables in the thermodynamic potentials are the
field variables of the medium, and not of the system as the latter are not
even defined when the system in not in internal equilibrium; see for example,
Eq. (\ref{Free_Energies}) or (\ref{Thrmo_Potential}). This is true whether we
consider the system or any of its many subsystems; the latter are considered
in Sect. \ref{marker_inhomogeneous_arbitrary}. This is surprising since one
knows that, in equilibrium thermodynamics, the thermodynamic potentials are
state quantities. Therefore, they must contain all quantities related to the
system. It just happens that in equilibrium the field variables are identical
to the medium's field variables. Therefore, our thermodynamic potentials
reduce to the standard thermodynamic potentials in equilibrium. In this case,
they become state functions. But this is not true when the system is not in
equilibrium. What is surprising is that this result remains true even when the
system or subsystem is in internal variable. The thermodynamic potentials
always contain fields of the medium and not of the system or subsystem. These
functions have the required thermodynamic property that they can never
increase in any spontaneous process, as seen from Eq.
(\ref{General_Thermodynamic_Potential_Variation0}) or
(\ref{General_Thermodynamic_Potential_Variation}). One can easily see that
this property is a consequence of the \emph{convexity} property of the
thermodynamic potentials. In contrast, a function obtained from thermodynamic
potentials by replacing medium's fields by the body's field, when the latter
is in internal equilibrium, does not have this required thermodynamic
property. An example of this is the function $\widehat{g}(t)$ or $\widehat
{G}(t)$ from $g(t)$ or $G(t)$ given in Eq. (\ref{Local_Gibbs_Free_Energy}) or
(\ref{Free_Energies_Incorrect}).

While each subsystem is in internal equilibrium, so that there can be no
irreversible processes inside them, the irreversibility occurs due to relative
motions going on among them. Because of the internal equilibrium, each
subsystem has its own internal temperature $T_{k}(t)$, pressure $P_{k}(t)$ or
other conjugate field variables$.$ How are these temperatures and pressures
related to the internal temperature $T(t)$ and pressure $P(t)$ introduced in
I? To answer these questions, we consider their definitions in the
$\mathcal{L}$ frame:%
\begin{align*}
\frac{1}{T(t)}  &  =\frac{\partial S}{\partial E},\frac{P(t)}{T(t)}%
=\frac{\partial S}{\partial V}\\
\frac{1}{T_{k}(t)}  &  =\frac{\partial s_{k}}{\partial e_{k}},\frac{P_{k}%
(t)}{T_{k}(t)}=\frac{\partial s_{k}}{\partial v_{k}};
\end{align*}
see Eq. (\ref{subsystem_parameters0}) for the definition of $T_{k}(t)$ and
$P_{k}(t)$. Let us first consider the temperature. Introducing%
\[
r_{k}(t)\equiv\frac{\partial E}{\partial e_{k}},%
{\textstyle\sum}
r_{k}(t)=1,
\]
we have%
\[
\frac{\partial S}{\partial E}=%
{\textstyle\sum\nolimits_{k}}
\frac{\partial s_{k}}{\partial e_{k}}/\frac{\partial E}{\partial e_{k}}=%
{\textstyle\sum\nolimits_{k}}
\frac{r_{k}(t)}{T_{k}(t)}.
\]
Thus, we have
\[
\frac{1}{T(t)}=\sum\nolimits_{k}\frac{r_{k}(t)}{T_{k}(t)}.
\]
Introducing
\[
r_{k}^{\text{v}}(t)\equiv\frac{\partial V}{\partial v_{k}},%
{\textstyle\sum_{k}}
r_{k}^{\text{v}}(t)=1,
\]
we find that%
\[
\frac{P(t)}{T(t)}=\sum\nolimits_{k}\frac{r_{k}^{\text{v}}(t)P_{k}(t)}%
{T_{k}(t)}.
\]
The same exercise can be carried out for other state variables. Introducing%
\[
r_{k}^{l}(t)\equiv\frac{\partial Z_{l}}{\partial z_{lk}},%
{\textstyle\sum_{k}}
r_{k}^{l}(t)=1,
\]
for the $l$-th state variable $Z_{l}$, we find that
\[
w_{l}(t)\equiv\frac{W_{l}(t)}{T(t)}=\sum\nolimits_{k}\frac{r_{k}^{l}%
(t)y_{k}(t)}{T_{k}(t)}.
\]

\subsection{Contrast with Local Non-equilibrium Thermodynamics}

Our approach differs from the traditional local non-equilibrium thermodynamics
due to de Groot \cite{Donder,deGroot,Prigogine,Prigogine-old}, briefly
discussed in Sect. \ref{marker_averages}, in three important ways.

\begin{enumerate}
\item The first one relates to the principle of additivity of energy. The
energy in the local non-equilibrium thermodynamics is not the sum of the
energies of its various part due to the presence of the mutual interaction
energies expressed by $\psi(\mathbf{r})dV$ as seen from Eq.
(\ref{Energy_Additivity_Thermodynamics}). Indeed, usually one applies local
equilibrium to a volume element $dV$, what is traditionally called a
physically infinitesimal volume in that, while it contains a large number of
atoms, the corresponding volume is still infinitesimally small. Such a volume
is conventionally called a "particle" (not to be confused with our usage which
refers to an atom or molecule). For all practical purposes, it is indeed
considered as a limit $dV\rightarrow0$. It is evident that one must then
consider the interaction energy $\psi(\mathbf{r})dV$ to account for the
interaction of this volume with the rest of the system. Thus, the local
non-equilibrium thermodynamics takes the additivity of entropy is taken as
postulate even when the energy is not additive.

\item Because the volume element $dV$ is treated infinitesimal in the local
theory, all thermodynamic quantities are also treated as continuous in space,
while this continuity is not a prerequisite in our approach.

\item Our approach also differs from the local non-equilibrium theory in the
form of the thermodynamic potentials. The volume element is considered to be
in internal equilibrium from the start so that it has a well-defined
temperature, pressure, etc. These fields are used in identifying the
thermodynamic potentials. For example, the Gibbs free energy of the "particle"
is taken to be $\widehat{g}(t)dV$even if the system consisting of such
"particles" is not in equilibrium with the medium.

\item Another imporatnt difference between the two approaches is that the
reversible entropy change and the irreversible entropy generation in each
subsystem also depends on the equilibrium state of the system. The
irreversible entropy generation in the local theory depends only on the
current local properties.
\end{enumerate}

These differences make our approach very different from the local
non-equilibrium thermodynamics due to de Groot
\cite{Donder,deGroot,Prigogine,Prigogine-old}.

\textbf{Acknowledgements }I would like to thank Arkady Leonov for a discussion
on the local non-equilibrium thermodynamics and him and Peter Wolynes for
their suggestion to apply the approach of I to inhomogeneous systems. I would
also like to thank an anonymous referee of I who wanted the current approach
to include internal dissipation.%

\appendix{}%

\section{Relation Between Lab and Body Frames \label{Appd_Frames}}

Let us consider a particle of mass $m$ in the lab frame $\mathcal{L}$, where
it has a velocity $\mathbf{v}_{\mathcal{L}}$ and the potential energy $U$. The
Lagrangian $L_{\mathcal{L}}$ in this frame is given by
\[
L_{\mathcal{L}}=\frac{1}{2}m\mathbf{v}_{\mathcal{L}}^{2}-U.
\]
Let us transform to a rotating frame $\mathcal{C}$ which is moving with a
velocity $\mathbf{V}$ and rotating with an angular velocity $\mathbf{\Omega.}$
The velocity of the particle in $\mathcal{C}$ is given by $\mathbf{v}%
_{\mathcal{C}}$ and is related to $\mathbf{v}_{\mathcal{L}}$ according to
\begin{equation}
\mathbf{v}_{\mathcal{L}}=\mathbf{v}_{\mathcal{C}}\mathbf{+V+\Omega\times
r}_{\mathcal{C}}\mathbf{,} \label{velocity_Relation}%
\end{equation}
where $\mathbf{r}_{\mathcal{C}}$ is the coordinate of the particle in the
$\mathcal{C}$ frame, and is related to the coordinate $\mathbf{r}%
_{\mathcal{L}}$ of the particle in the lab frame $\mathcal{L}$ by%
\[
\mathbf{r}_{\mathcal{L}}\equiv\mathbf{R}+\mathbf{r}_{\mathcal{C}}.
\]
In the following, it would also be convenient to consider a nonrotating frame
$\mathcal{I}$, which is only moving with the velocity $\mathbf{V}$ with
respect to the lab frame, but whose origin coincides with that of
$\mathcal{C}$. The Lagrangian in the frame $\mathcal{C}$ is given by%
\begin{align}
L_{\mathcal{C}}  &  =\frac{1}{2}m\mathbf{(\mathbf{v}_{\mathcal{C}}%
+V+\Omega\times r}_{\mathcal{C}}\mathbf{)}^{2}-U\nonumber\\
&  =\frac{1}{2}m\mathbf{v}_{\mathcal{C}}^{2}+\frac{1}{2}m\mathbf{(\Omega\times
r}_{\mathcal{C}}\mathbf{)}^{2}+m\mathbf{v}_{\mathcal{C}}\cdot\mathbf{V+}%
m\mathbf{v}_{\mathcal{C}}\cdot\mathbf{\Omega\times\mathbf{r}}_{\mathcal{C}%
}\mathbf{+}m\mathbf{V}\cdot\mathbf{\Omega\times\mathbf{r}}_{\mathcal{C}%
}\mathbf{-}U, \label{Energy_Relation}%
\end{align}
in which we have omitted $%
\frac12
m\mathbf{V}^{2}$, which is a total time derivative.

The canonical momentum is obtained as
\[
\mathbf{p}_{\mathcal{C}}=\frac{\partial L_{\mathcal{C}}}{\partial
\mathbf{v}_{\mathcal{C}}}=m(\mathbf{v}_{\mathcal{C}}\mathbf{+V+\Omega
\times\mathbf{r}}_{\mathcal{C}}\mathbf{)},
\]
so that the energy of the particle becomes%
\begin{equation}
E_{\mathcal{C}}=\mathbf{p}_{\mathcal{C}}\cdot\mathbf{v}_{\mathcal{C}%
}-L_{\mathcal{C}}=\frac{1}{2}m\mathbf{v}_{\mathcal{C}}^{2}-\frac{1}%
{2}m\mathbf{(\Omega\times\mathbf{r}}_{\mathcal{C}}\mathbf{)}^{2}%
-m\mathbf{V}\cdot\mathbf{\Omega\times\mathbf{r}}_{\mathcal{C}}\mathbf{+}U.
\label{Internal_Energy_particle}%
\end{equation}
Expressing $\mathbf{v}_{\mathcal{C}}$ in terms of $\mathbf{v}_{\mathcal{L}}$
from Eq. (\ref{velocity_Relation}), we find that
\begin{equation}
E_{\mathcal{C}}=E_{\mathcal{L}}+\frac{1}{2}m\mathbf{V}^{2}-\mathbf{p}%
_{\mathcal{L}}\cdot\mathbf{V}-m\mathbf{r}_{\mathcal{C}}\times\mathbf{v}%
_{\mathcal{L}}\cdot\mathbf{\Omega=}E_{\mathcal{L}}+\frac{1}{2}%
m\mathbf{\mathbf{V}^{2}-\mathbf{p}}_{\mathcal{L}}\mathbf{\cdot\mathbf{V}%
-m}\cdot\mathbf{\Omega}, \label{Energy_Transformation_particle}%
\end{equation}
where
\begin{equation}
E_{\mathcal{L}}=\frac{1}{2}m\mathbf{v}_{\mathcal{L}}^{2}+U,\ \mathbf{p}%
_{\mathcal{L}}=m\mathbf{v}_{\mathcal{L}},\ \ \mathbf{m=}m\mathbf{r}%
_{\mathcal{C}}\times\mathbf{v}_{\mathcal{L}}. \label{Definition}%
\end{equation}
Incidentally, we also note that
\[
\mathbf{p}_{\mathcal{C}}\mathbf{=p}_{\mathcal{L}}.
\]
Thus,
\[
\mathbf{m}=\mathbf{r}_{\mathcal{C}}\times\mathbf{p}_{\mathcal{L}}%
=\mathbf{r}_{\mathcal{C}}\times\mathbf{p}_{\mathcal{C}},
\]
which explains why there is no need to use $\mathcal{L}$ or $\mathcal{C}$ as a
subscript in $\mathbf{m}$.

\section{A Rotating and Translating System\label{Appd_Rotating_System_0}}

Let us now extend the previous calculation for a single particle to a system
of particles of total mass $M$ at a given instant $t$. The system is also
characterized by the number of particles and its volume. For specificity, we
use $N$ and $V(t)$ to denote these quantities. The notation should not mean
that we are only considering the system $\Sigma$ here. The system we have in
mind is any generic system. We assume that this system is translating with a
velocity $\mathbf{V}(t)$ and rotating with an angular velocity $\mathbf{\Omega
}(t)$ as a whole, both of which can change in time. For notational simplicity,
we will suppress the explicit $t$-dependence of various quantities here. We
focus on one particular instant $t$. We take the center of mass of this system
to coincide with the origin of $\mathcal{C}$, so that $\mathcal{C}$ is fixed
to the body and rotating with it. In this case, $\mathcal{C}$ represents the
\emph{center of mass frame}. We now sum Eq.
(\ref{Energy_Transformation_particle}) over all particles, with the result
that it is replaced by%
\begin{equation}
E_{\mathcal{C}}=E_{\mathcal{L}}+\frac{\mathbf{P}^{2}}{2M}-\mathbf{P}%
_{\mathcal{L}}\cdot\mathbf{V}-\mathbf{M}\cdot\mathbf{\Omega},
\label{Energy_Transformation_System}%
\end{equation}
where we have introduced%
\begin{equation}
\mathbf{P=}M\mathbf{V},\ \ \mathbf{P}_{\mathcal{L}}=%
{\displaystyle\sum}
m\mathbf{v}_{\mathcal{L}},\ \ \mathbf{M\equiv}%
{\displaystyle\sum}
\mathbf{r}_{\mathcal{C}}\times\mathbf{p}_{\mathcal{L}}\equiv%
{\displaystyle\sum}
\mathbf{r}_{\mathcal{C}}\times\mathbf{p}_{\mathcal{C}}\mathbf{.}
\label{System_Definition}%
\end{equation}
Here, $\mathbf{M}_{\mathcal{C}}$ is the \emph{intrinsic} angular momentum of
the system of particles in the $\mathcal{C}$ frame. The summation in the above
formulas, which are also applicable to non-uniform rotation and translation of
the center of mass frame $\mathcal{C}$, is a sum over all the particles in the
system. The value of $E_{\mathcal{C}}$ in the $\mathcal{C}$ frame represents
the internal energy $E^{\text{i}}$ of the system. The equation
(\ref{Energy_Transformation_System}) generalizes the result given by Landau
and Lifshitz \cite[see their Eq. (39.13)]{Landau-Mechanics} to the case when
$\mathbf{V}\neq0.$ The present generalization is not limited to constant
rotation and translation.

We can express $\mathbf{M}$ as follows:%
\[
\mathbf{M=}%
{\displaystyle\sum}
m\mathbf{r}_{\mathcal{C}}\times(\mathbf{v}_{\mathcal{C}}\mathbf{+V+\Omega
\times r}_{\mathcal{C}})=%
{\displaystyle\sum}
m\mathbf{r}_{\mathcal{C}}\times\mathbf{v}_{\mathcal{C}}+\left(
{\displaystyle\sum}
m\mathbf{r}_{\mathcal{C}}\right)  \times\mathbf{V+}%
{\displaystyle\sum}
m\mathbf{r}_{\mathcal{C}}\times(\mathbf{\Omega\times r}_{\mathcal{C}}),
\]
in which the second sum on the right vanishes for $\mathcal{C}$, the center of
mass frame. For the same reason, the third term in the second equation in Eq.
(\ref{Internal_Energy_particle}) does not contribute to the energy of the
system. Thus, we find
\begin{subequations}
\begin{align}
\mathbf{M}  &  \mathbf{=}%
{\displaystyle\sum}
m\mathbf{r}_{\mathcal{C}}\times\mathbf{v}_{\mathcal{C}}+%
{\displaystyle\sum}
m\mathbf{r}_{\mathcal{C}}\times(\mathbf{\Omega\times r}_{\mathcal{C}%
})\label{Angular_momentum_C}\\
E_{\mathcal{C}}  &  =\frac{1}{2}%
{\displaystyle\sum}
m\mathbf{v}_{\mathcal{C}}^{2}-\frac{1}{2}m\mathbf{(\Omega\times\mathbf{r}%
}_{\mathcal{C}}\mathbf{)}^{2}\mathbf{+}U. \label{Internal_Energy_C_system}%
\end{align}
We see that
\end{subequations}
\begin{equation}
\mathbf{M\cdot\mathbf{\Omega}=}%
{\displaystyle\sum}
m\mathbf{r}_{\mathcal{C}}\cdot\left(  \mathbf{v}_{\mathcal{C}}\times
\mathbf{\Omega}\right)  +%
{\displaystyle\sum}
m(\mathbf{\Omega\times r}_{\mathcal{C}})^{2}. \label{Angular_Momentum_Work}%
\end{equation}
The first term vanishes when the motion is radial. Thus, it is the
contribution to the energy from the relative transverse motion in the
$\mathcal{C}$ frame and will vanish if the system is stationary in the this
frame. The latter condition is met when the system is in internal equilibrium
in accordance with Theorem \ref{marker_Uniform_Motion}.

For the center of mass frame $\mathcal{C},$ $\mathbf{P=P}_{\mathcal{L}},$ so
that the energy of the system in the frame $\mathcal{C}$ is given by%
\begin{equation}
E_{\mathcal{C}}=E_{\mathcal{L}}-\frac{\mathbf{P}^{2}}{2M}-\mathbf{M}%
\cdot\mathbf{\Omega,} \label{Internal_Energy}%
\end{equation}
which can be rewritten as
\begin{equation}
E_{\mathcal{L}}=E_{\mathcal{C}}+\frac{\mathbf{P}^{2}}{2M}+\mathbf{M}%
\cdot\mathbf{\Omega} \label{Energy_L_system0}%
\end{equation}
Using Eq. (\ref{Internal_Energy_C_system}), we find that
\begin{equation}
E_{\mathcal{L}}=E_{\mathcal{C}}+\frac{\mathbf{P}^{2}}{2M}+\frac{1}{2}%
{\displaystyle\sum}
m(\mathbf{\Omega\times r}_{\mathcal{C}})^{2}+%
{\displaystyle\sum}
m\mathbf{r}_{\mathcal{C}}\cdot\left(  \mathbf{v}_{\mathcal{C}}\times
\mathbf{\Omega}\right)  , \label{Energy_L_system1}%
\end{equation}
in which the last term is the contribution of the transverse motion.

The energy of the system in the $\mathcal{I}$ frame, in which the system has
only rotation, is given by
\begin{equation}
E_{\mathcal{I}}=E_{\mathcal{L}}-\frac{\mathbf{P}^{2}}{2M}.
\label{Internal_Energy_Iframe}%
\end{equation}
It is obvious that $E_{\mathcal{I}}$ does not depend on the velocity
$\mathbf{V}$. Thus,%
\[
E_{\mathcal{C}}=E_{\mathcal{I}}-\mathbf{M}\cdot\mathbf{\Omega}%
\]
does not depend explicitly on $\mathbf{V}$. We can take $E_{\mathcal{C}}$ as a
function of $E_{\mathcal{I}}$ and $\mathbf{\Omega,}$ from which it follows
that%
\begin{align}
\left.  \left(  \frac{\partial E_{\mathcal{C}}}{\partial\mathbf{V}}\right)
\right\vert _{E_{\mathcal{I}},V,N,\mathbf{\Omega}}  &  =0\mathbf{,}%
\label{Ec_V_Relation}\\
\left.  \left(  \frac{\partial E_{\mathcal{C}}}{\partial\mathbf{\Omega}%
}\right)  \right\vert _{E_{\mathcal{I}}\mathbf{,}V,N}  &  =-\mathbf{M}.
\label{Ec_Mc_Relation}%
\end{align}
All the above results are for a particular microstate of the system undergoing
a Hamiltonian dynamics. To obtain thermodynamics of the system, we need to
average the above equations over all microstates using their probabilities,
which will be taken up in Sect. \ref{marker_Rotating_Systems}. The averaging
takes into account the stochastic nature of the dynamics, which has not been
considered in both appendices.

\end{document}